\documentclass[11pt,reqno]{amsart}
\usepackage{amsmath,amsthm,amssymb,bm,bbm,dsfont,braket,cases}

\usepackage[mathscr]{eucal}
\usepackage{amsaddr}
\usepackage{upgreek}
\usepackage[left=2cm,right=2cm,top=2cm,bottom=2cm]{geometry}
\usepackage{mathtools}\mathtoolsset{centercolon}
\DeclarePairedDelimiter{\ceil}{\lceil}{\rceil}
\DeclareMathOperator*{\argmax}{arg\,max}
\DeclareMathOperator*{\argmin}{arg\,min} % thin space, limits underneath in displays
\usepackage{cite}

\usepackage[shortlabels]{enumitem}
\usepackage{setspace}
\usepackage{graphicx}
\usepackage{verbatim}
\usepackage{float}
\usepackage{placeins}
\usepackage{array}
\usepackage{booktabs}
\usepackage{threeparttable}
\usepackage[update,prepend]{epstopdf}
\usepackage{multirow}
\usepackage{amsfonts,amssymb,dsfont,xfrac,comment}
\usepackage[abs]{overpic}

\usepackage[usenames,dvipsnames]{color}
\usepackage[hidelinks,linktocpage=true]{hyperref}
\hypersetup{
    linktoc=all,     %set to all if you want both sections and subsections linked
	unicode=false,          % non-Latin characters in Acrobat’s bookmarks
	pdftoolbar=true,        % show Acrobat’s toolbar?
	pdfmenubar=true,        % show Acrobat’s menu?
	pdffitwindow=false,     % window fit to page when opened
	pdfstartview={FitH},    % fits the width of the page to the window
	pdftitle={My title},    % title
	colorlinks=true,        % false: boxed links; true: colored links
	linkcolor=BrickRed,          % color of internal links (change box color with linkbordercolor)
	citecolor=OliveGreen,  % color of links to bibliography
	filecolor=Magenta,      % color of file links
	urlcolor=BlueViolet,    % color of external links
}
\usepackage{doi}
\usepackage{url}
\usepackage{caption, subcaption}
\usepackage{enumitem}
\usepackage{shadethm}
\usepackage{tikz,pgf}

\makeatletter
\renewcommand\tocchapter[3]{%
  \indentlabel{\@ifnotempty{#2}{\ignorespaces#2.\quad}}#3%
}
\newcommand\@dotsep{4.5}
\def\@tocline#1#2#3#4#5#6#7{\relax
  \ifnum #1>\c@tocdepth % then omit
  \else
    \par \addpenalty\@secpenalty\addvspace{#2}%
    \begingroup \hyphenpenalty\@M
    \@ifempty{#4}{%
      \@tempdima\csname r@tocindent\number#1\endcsname\relax
    }{%
      \@tempdima#4\relax
    }%
    \parindent\z@ \leftskip#3\relax \advance\leftskip\@tempdima\relax
    \rightskip\@pnumwidth plus1em \parfillskip-\@pnumwidth
    #5\leavevmode\hskip-\@tempdima{#6}\nobreak
    \leaders\hbox{$\m@th\mkern \@dotsep mu\hbox{.}\mkern \@dotsep mu$}\hfill
    \nobreak
    \hbox to\@pnumwidth{\@tocpagenum{#7}}\par
    \nobreak
    \endgroup
  \fi}
\makeatother
\AtBeginDocument{%
\makeatletter
\expandafter\renewcommand\csname r@tocindent0\endcsname{0pt}
\makeatother
}
\makeatletter
\ifx\@NODS\undefined%

\let\mathbb=\mathds
\else%
\fi
\makeatother

\DeclareMathOperator{\Tr}{Tr}

\newcommand{\be}{{\mathbf e}}

\newcommand{\tr}{\operatorname{Tr}}
\newcommand{\al}{{\alpha}}

\newcommand{\ten}{\otimes}
\newcommand{\pl}{\hspace{.1cm}}

\newcommand{\norm}[2]{\parallel \! #1 \! \parallel_{#2}}

\def\0{{\mathbf{0}}}
\def\1{{\mathbf{1}}}
\def\2{{\mathbf{2}}}
\def\3{{\mathbf{3}}}
\def\4{{\mathbf{4}}}
\def\5{{\mathbf{5}}}
\def\6{{\mathbf{6}}}

\def\7{{\mathbf{7}}}
\def\8{{\mathbf{8}}}
\def\9{{\mathbf{9}}}

\def\be{\begin{equation}}
\def\ee{\end{equation}}
\def\bea{\begin{eqnarray}}
\def\eea{\end{eqnarray}}

\def\eps{\varepsilon}

\theoremstyle{plain}
\newshadetheorem{theo}{Theorem} %[section]%[chapter]
\newshadetheorem{conj}{Conjecture} %[section]
\newshadetheorem{prop}[theo]{Proposition} %[section]
\newshadetheorem{lemm}[theo]{Lemma} %[section]
\newshadetheorem{coro}[theo]{Corollary} %[section]
\newtheorem{fact}{Fact} %[section]%[chapter]
\newshadetheorem{theorem}{Theorem}
\newshadetheorem{question}{Question}
\newtheorem*{lemm1}{Lemma~\ref{lemm:key}}
\newtheorem*{lemm2}{Lemma~\ref{lemm:converse}}
\newtheorem*{prop_add}{Proposition~\ref{prop:properties}}

\theoremstyle{definition}
\newtheorem{defn}[theo]{Definition} %[section]

\theoremstyle{remark}
\newtheorem{remark}[theo]{Remark}%[section]

\makeatletter
\newcommand{\opnorm}{\@ifstar\@opnorms\@opnorm}
\newcommand{\@opnorms}[1]{%
	$\left|\mkern-1.5mu\left|\mkern-1.5mu\left|
	#1
	\right|\mkern-1.5mu\right|\mkern-1.5mu\right|$
}
\newcommand{\@opnorm}[2][]{%
	\mathopen{#1|\mkern-1.5mu#1|\mkern-1.5mu#1|}
	#2
	\mathclose{#1|\mkern-1.5mu#1|\mkern-1.5mu#1|}
}
\makeatother

\begin{document}

%%%%% How to disable amsart.cls to capitalize the article title?
\let\origmaketitle\maketitle
\def\maketitle{
	\begingroup
	\def\uppercasenonmath##1{} % this disables uppercasing title
	\let\MakeUppercase\relax % this disables uppercasing authors
	\origmaketitle
	\endgroup
}
%%%%%%%%%%%%%%%%%%%%%%%%%%%%%%%%

\title{\bfseries \Large{
Tight One-Shot Analysis for Convex Splitting\\
with Applications in Quantum Information Theory
}}

\author{ \normalsize {Hao-Chung Cheng$^{1\text{--}5}$ and Li Gao$^6$}}
\address{\small  	
$^1$Department of Electrical Engineering and Graduate Institute of Communication Engineering,\\ National Taiwan University, Taipei 106, Taiwan (R.O.C.)\\
% $^1$Department of Electrical Engineering, Department of Mathematics, National Taiwan University\\
$^2$Department of Mathematics, National Taiwan University\\
$^3$Center for Quantum Science and Engineering,  National Taiwan University\\
$^4$Physics Division, National Center for Theoretical Sciences, Taipei 10617, Taiwan (R.O.C.)\\
$^5$Hon Hai (Foxconn) Quantum Computing Center, New Taipei City 236, Taiwan (R.O.C.)\\
$^6$Department of Mathematics, University of Houston, Houston, TX 77204, USA
}

\email{\href{mailto:haochung.ch@gmail.com}{haochung.ch@gmail.com}}
 \email{\href{mailto:gaolimath@gmail.com}{gaolimath@gmail.com}}

\date{\today}
\begin{abstract}
	Convex splitting is a powerful technique in quantum information theory used in proving the achievability of numerous information-processing protocols such as quantum state redistribution and quantum network channel coding.
	In this work, we establish % tight one-shot characterizations and
	a one-shot error exponent and a one-shot exponential strong converse
	for convex splitting with trace distance as an error criterion. 
	Our results show that the derived error exponent (strong converse exponent) is positive if and only if the rate is in (outside) the achievable region.  This leads to new one-shot exponent results in various tasks such as communication over quantum wiretap channels, secret key distillation, one-way quantum message compression, quantum measurement simulation, and quantum channel coding with side information at the transmitter. We also establish a near-optimal one-shot characterization of the sample complexity for convex splitting, which yields matched second-order asymptotics. This then leads to stronger one-shot analysis in many quantum information-theoretic tasks. %Our technique is a unified functional analytic approach using Kosaki's noncommutative weighted $L_p$ norm. 

\end{abstract}
\maketitle

\vspace{-29pt}
% \tableofcontents
{
  \hypersetup{linkcolor=black}
  \tableofcontents
}

\newpage
\section{Introduction} \label{sec:introduction}
\emph{Convex splitting} \cite{HJM+10, ADJ17} is a powerful technique to decouple a bipartite quantum system, which has various  applications in quantum information theory such as achievability in quantum state redistribution, port-based teleportation \cite{ADJ17},  and quantum  channel coding \cite{HJM+10, YHD11, SW15, Dup10, DHL10, RSW16, ADJ17, WDW17, Wil17b, AJW19a, KKG+19, SHA20, KW20}.
\emph{Quantum covering} \cite{AW02,Hay15,Hay17, CG22, SGC22b}, as another useful technique in quantum information,  aims to approximate a quantum state by sampling from a prior distribution and querying a quantum channel.
%This idea came from  Wyner \cite{Wyn73}, Han and Verd{\'u} \cite{HV93} in classical information theory, and be generalized
Quantum covering was first studied by Ahlswede and Winter \cite{AW02} under the name of \emph{Operator Chernoff Bound} and later developed by Hayashi \cite{Hay15, Hay17}; it has further applications in source coding \cite{DW03}, channel identification \cite{AW02}, channel resolvability \cite[\S 9]{Hay17}, and quantum channel simulation \cite{WHB+12, 6757002}.

In this paper, we approach these two substantial tasks using a unified technique from noncommutative $L_p$ spaces, and obtain  tight one-shot error exponents, one-shot exponential strong converses, and the corresponding \emph{near-optimal} characterizations of the sample complexity. We will show that the proposed tight analysis improves upon existing achievability results in quantum information theory
%Hence, they may serve as substantial building blocks for all the \emph{covering-type} problems.

We consider the following problems.
For density operators  $\rho_{AB}$, $\tau_A$, and any integer $M\in\mathds{N}$,
%quantum states such that $\mathtt{supp}(\rho_B)\subseteq\mathtt{supp}(\sigma_B)$,
we let
\begin{align} 
\begin{split}\label{eq:total_state}
\omega_{A_1\ldots A_M B} := \frac{1}{M}\sum_{m=1}^M
\rho_{A_m B} \otimes \tau_{A_1} \otimes \tau_{A_2} \otimes \cdots \otimes \tau_{A_{m-1}} 
\otimes \tau_{A_{m+1}} \otimes \cdots \otimes \tau_{A_M}, %\label{eq:total_state}
\end{split}
\end{align}
where $\rho_{A_m B} = \rho_{AB}$ and $\tau_{A_m} = \tau_A$ for all $m\in[M]:=\{1,\ldots,M\}$.
%Without loss of generality, we further assume that the density operator $\tau_A$ is invertible here and subsequently; otherwise, we may restrict the associated Hilbert space $\mathsf{A}$ to the support of $\tau_A$.
The statistical mixture in \eqref{eq:total_state} is a \emph{convex-splitting} operation that attempts to decouple the system $A$ from $B$ via increasing the integer $M$ \cite{HJM+10, ADJ17}; see Figure~\ref{fig:convex_splitting} below for an illustration.
This operation serves as a technical primitive in quantum information theory and has
a wealth of applications therein \cite{HJM+10, ADJ17, AJW19a, Wil17b, KKG+19, SHA20, KW20}.
The goal of this paper is to provide tight one-shot characterizations on how well the convex splitting can accomplish decoupling, and accordingly, its applications of improving analysis for the existing quantum information-theoretic tasks.

\begin{figure}[h!]
	\centering
	\resizebox{1\columnwidth}{!}{ 
		\includegraphics{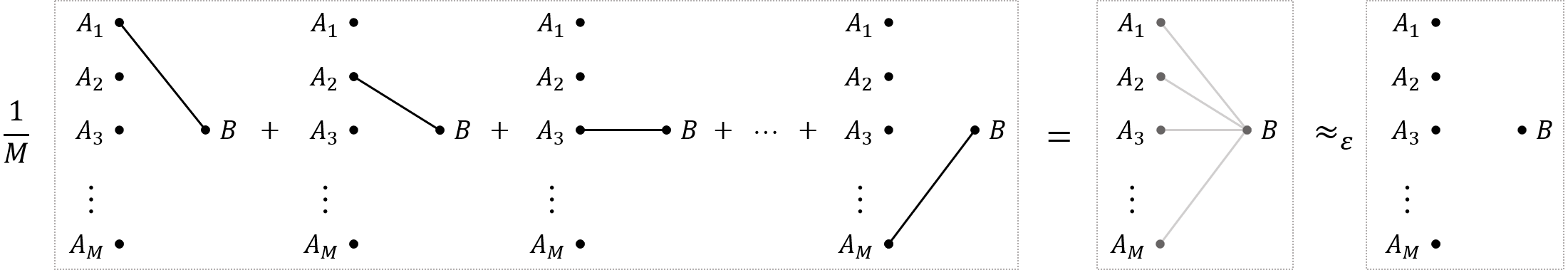}     		   		
	}
	\caption{
    \small
    Depiction of a convex splitting.
    On the left part, solid lines connected each system is a bipartite (correlated) state, while the other systems are left isolated. 
    On the middle part is the statistical mixture of the $M$ states on the left; we can see there are still correlations in the presence between system $B$ and each system $A$ (depicted in light-gray lines).
    When $M$ is sufficiently large, then the statistical mixture is close to the right part of product state (in trace norm), where each isolated dot represents an independent quantum system.
	}
	\label{fig:convex_splitting}
\end{figure}

To measure the closeness between the mixture $\omega_{A_1\ldots A_M B}$ and the product its marginal states, we adopt trace distance as the error criterion\footnote{
We note that some literature adopted quantum relative entropy as an criterion \cite{ADJ17, AJW19a}. In this paper we consider the trace distance, which can be viewed as a generalization of the classical total variation distance \cite{anshu2017unified}.
}, i.e.~denoting $\|\cdot\|_1$ as the trace norm,
\begin{align}
	\begin{split} \label{eq:covering}
&\Delta^\texttt{c}_M\left( \rho_{AB} \, \Vert\, \tau_A\right) :=
\frac12\left\| \omega_{A_1\ldots A_M B} - \bigotimes_{m=1}^M \tau_{A_m} \otimes \rho_B \right\|_1.
%\\
%&\Delta^\texttt{c}_M\left(A:B\right)_{\rho} :=
%\Delta^\texttt{c}_M\left( \rho_{AB} \, \Vert\, \rho_A\right).
	\end{split}
\end{align}
%Here, the superscript `$\texttt{c}$' stands for ``covering''.
It is natural to ask the following questions:
\begin{align}
&\textbf{Given an $\eps\in(0,1)$, how small $M$ can be to achieve $\Delta^\texttt{c}_M\left( \rho_{AB} \, \Vert\, \tau_A\right)\leq \eps$?} \tag{\text{Q1}} \label{eq:Q2} \\
&\textbf{Given an integer $M$, how small is $\Delta^\texttt{c}_M\left( \rho_{AB} \, \Vert\, \tau_A\right)$? } \tag{\text{Q2}} \label{eq:Q1}
\end{align}
Regarding Question \eqref{eq:Q2}, the minimum integer $M$ to achieve $\Delta _M\left( \rho_{AB} \, \Vert\, \tau_A\right)\leq \eps$ is  called the \emph{sample complexity} for achieving an $\eps$-convex-splitting, for which we denoted by $M_\varepsilon^{\textnormal{\texttt{c}}}$.
We obtain a tight one-shot characterization for an $\eps$-convex-splitting in terms of the \emph{hypothesis-testing information} (Theorems~\ref{theo:direct} \& \ref{theo:converse} of Section~\ref{sec:sample}):
	\begin{align} \label{eq:intro_sample}
	\log M_\varepsilon^{\textnormal{\texttt{c}}}\left(\rho_{AB} \,\Vert\, \tau_A\right)
	\approx D_\textnormal{h}^{1-\varepsilon }\left(\rho_{AB}\,\|\,\tau_A\otimes \rho_B \right).
	\end{align}
	Here, $\eps$-hypothesis-testing divergence is $\displaystyle D_\textnormal{h}^\eps(\rho\,\Vert\, \sigma) := \sup_{0\leq T\leq \mathds{1} } \left\{ -\log \Tr[\sigma T] : \Tr[\rho T] \geq 1 - \varepsilon \right\}$ \cite{WR13, TH13, Li14}.
	The approximation ``$\approx$'' is up to some logarithmic terms.
%
%	Moreover, for sufficiently large $n\in\mathds{N}$ and every $\eps\in(0,1)$,% we have
%	\begin{align}
%	\log M_\varepsilon^{\textnormal{\texttt{c}}}\left(\rho_{AB}^{\otimes n} \,\Vert\, \tau_A^{\otimes n} \right)
%	&= n D(\rho_{AB}\,\Vert\, \tau_{A} \otimes \rho_B) - \sqrt{n V(\rho_{AB}\,\Vert\, \tau_{A} \otimes \rho_B) } \, \Phi^{-1}(\eps) + O(\log n),
%	\end{align}
%	where $\Phi^{-1}({\eps}) := \sup \{ \gamma\in \mathds{R} : \frac{1}{\sqrt{2\pi}} \int_{-\infty}^{\gamma} \mathrm{e}^{-u^2/2} \, \mathrm{d} u \leq \eps \}$, and $V$ is the relative information variance \cite{TH13, Li14}.
We argue that the established characterization is \emph{near optimal} in the sense that the second-order rate obtained are matched in the identical and independently distributed (i.i.d.) asymptotic scenario, i.e.~every $\eps\in(0,1)$,
% , wherein the underlying states are $\rho_{AB}^{\otimes n}$ and $\tau_A^{\otimes n}$.
% we have
	\begin{align}
		\log M_\varepsilon^{\textnormal{\texttt{c}}}\left(\rho_{AB}^{\otimes n} \,\Vert\, \tau_A^{\otimes n} \right)
		&= n I(\rho_{AB}\,\Vert\, \tau_{A}) - \sqrt{n V(\rho_{AB}\,\Vert\, \tau_{A} \otimes \rho_B) } \,  \Phi^{-1}(\eps) + O(\log n),
	\end{align}
	where $ \Phi^{-1}$ a inverse cumulative function of standard normal distribution, $I(\rho_{AB}\,\Vert\, \tau_{A})$ is a \emph{generalized mutual information} \cite{HT14}, and $V$ is a generalized relative information variance \cite{TH13, Li14}.

Regarding \eqref{eq:Q1}, we are concerned with whether the error, $\Delta _M^{\texttt{c}}$, can be upper bounded by an exponentially decaying quantity and when the associated (one-shot) error exponent is positive.
We term this an \emph{one-shot error exponent} for convex splitting.
On the other hand, we call an \emph{one-shot strong converse exponent} if the error can be lower bounded by one minus an exponentially decaying quantity.
We show that  (Theorems~\ref{theo:exponent} \& \ref{theo:sc} of Section~\ref{sec:exponet}):
% \begin{align}
\begin{subnumcases}{}
% \displaystyle	
\Delta^\texttt{c}_M\left( \rho_{AB} \Vert \tau_A\right) \leq \mathrm{e}^{ - \sup_{\alpha \in [1,2]} \frac{\alpha-1}{\alpha} ( \log M -  I_\alpha^*(\rho_{AB}\,\|\,\tau_A) ) },
& $\log M > I(\rho_{AB}\,\|\,\tau_A)$, \label{eq:intro_error_exponent} \\
\Delta^\texttt{c}_M\left( \rho_{AB} \Vert \tau_A\right) \geq 1 - 4\, \mathrm{e}^{ - \sup_{\alpha \in [\sfrac{1}{2},{1}]} \frac{1-\alpha}{\alpha} ( D_{2 - \sfrac{1}{\alpha}}\left(\rho_{AB}\,\|\,\tau_A\otimes \rho_B\right) - \log M )}, & $\log M < I(\rho_{AB}\,\|\,\tau_A)$, \label{eq:intro_sc} 
\end{subnumcases}
%	\end{align}
where
% ${E}^*(\log M) := \sup_{\alpha \in [1,2]} \frac{\alpha-1}{\alpha} ( \log M -  I_\alpha^*(\rho_{AB}\,\|\,\tau_A) )$;
% ${E}^\downarrow(\log M) :=  \sup_{\alpha \in [1/2,{1}]} \frac{1-\alpha}{\alpha} ( D_{2 - \frac{1}{\alpha}}\left(\rho_{AB}\,\|\,\tau_A\otimes \rho_B\right) - \log M )$;
$
I_\alpha^*\left(\rho_{AB}\,\|\,\tau_A\right)$ is a \emph{generalized sandwiched R\'enyi information} \cite{HT14} defined via the quantum sandwiched R\'enyi divergence $D_\alpha^*$ \cite{MDS+13, WWY14}, and $D_\alpha$ is the Petz--R\'enyi divergence~\cite{Pet86} (see Section~\ref{sec:exponet} for the precise definitions).
The  error exponent established in \eqref{eq:intro_error_exponent} is expressed in terms of $D_\alpha^*$ with order $\alpha \in [1,2]$, as opposed to previous results using the max-relative entropy $D_\infty^*$ \cite{Dat09,ADJ17}.
Then, the positivity of the error immediately implies an \emph{achievable rate region}: $\log M > I(\rho_{AB}\,\|\,\tau_A) = \lim_{\alpha \to 1} I_\alpha^*\left(\rho_{AB}\,\|\,\tau_A\right)$ without \emph{smoothing} \cite[Corollary 1]{AJW19a}, \cite[Lemma 2]{Wil17b}, \cite[Lemma 8]{KKG+19}.
On the other hand, the exponent in \eqref{eq:intro_sc} is positive if and only if $\log M < I(\rho_{AB}\,\|\,\tau_A)$.
This indicates that the (even in the one-shot setting) the {generalized mutual information} $I(\rho_{AB}\,\|\,\tau_A)$ is a fundamental limit that determines a  \emph{sharp} phase transition\footnote{In Ref.~\cite{ADJ17}, the authors also provided a \emph{weak converse}. Namely, the convex-splitting error is bounded away from $0$ once $\log M < I(\rho_{AB}\,\|\,\tau_A)$, but it did not characterize if the error will converge to $1$.} for the convex-splitting error $\Delta^\texttt{c}_M\left( \rho_{AB} \Vert \tau_A\right)$.
Notably, both the exponents in \eqref{eq:intro_error_exponent} and \eqref{eq:intro_sc} are additive for $n$-fold product $\rho_{AB}^{\otimes n}$ and $\tau_A^{\otimes n}$, demonstrating an exponential convergence of the error and the \emph{exponential strong converse}, for \emph{any} integer $n$.

%\[
%\begin{dcases}
%%\displaystyle	
%\Delta _M\left( \rho_{AB} \Vert \tau_A\right) \leq \mathrm{e}^{ - \mathsf{E}^*(\log M) },
% \\
%\Delta _M\left( \rho_{AB} \Vert \tau_A\right) \geq 1 - 4\mathrm{e}^{ - \sup\limits_{\alpha \in [1/2,{1}]} \frac{1-\alpha}{\alpha} ( D_{2 - \frac{1}{\alpha}}\left(\rho_{AB}\,\|\,\tau_A\otimes \rho_B\right) - \log M ) } . \\
%\end{dcases}
%%	\end{align}
%\]
%$D_\alpha^*(\rho\,\Vert\,\sigma) := \frac{\alpha}{\alpha-1} \log \| \sigma^{\frac{1-\alpha}{2\alpha}} \rho \sigma^{\frac{1-\alpha}{2\alpha}} \|_{\alpha}$

Our results then apply to various quantum information-processing tasks such as private communication over quantum wiretap channels, secret key distillation, one-way quantum message compression, quantum measurement simulation, and entanglement-assisted classical communication with channel state at the transmitter (see Section~\ref{sec:app}).
Specifically, we demonstrate that (as observed in many literature \cite{WKT15, anshu2017unified, AJW19a, Wil17b, Wilde2}) the error of a quantum information-theoretic task is usually composed of two parts: the \emph{packing error} and the \emph{covering error}.
A \emph{one-shot quantum packing lemma} (which slightly improves upon the analysis of the \emph{position-based coding} \cite{AJW19a}) by one of the author \cite{Cheng2022} describes the packing error, while the analysis of convex splitting provided in this work characterizes the covering error.

\medskip
This paper is organized as follows. In Section~\ref{sec:notation}, we formulate the convex splitting via noncommutative $L_p$ norms and present one-shot upper and lower analysis.
Section~\ref{sec:sample} devotes to a near-optimal characterization of the sample complexity regarding Question~\eqref{eq:Q2}.
Section~\ref{sec:exponet} contains a one-shot error exponent and one-shot exponential strong converse regarding Question~\eqref{eq:Q1}.
In Section~\ref{sec:app}, we collect applications of the analysis of convex splitting in quantum information theory.
We conclude the paper in Section~\ref{sec:conclusions}.

\section{ Preliminaries and Convex Splitting} \label{sec:notation}
For $p\geq 1$, the Schatten $p$-norm of an operator $X$ is defined as
\begin{align}
	\|X\|_p := \left( \Tr\left[ \left|X\right|^p \right] \right)^{\sfrac{1}{p}},
\end{align}
where $\Tr$ is the standard matrix trace. Let $\mathsf{A}$ and $\mathsf{B}$ be Hilbert spaces.
We denote $\mathcal{S}(\mathsf{A})$ as the set of all density operators (positive and trace $1$, denoted with subscript, as $\rho_A $) on $\mathsf{A}$. Given a natural number $M\in\mathds{N}^+$, a density $\tau_A\in \mathcal{S}(\mathsf{A})$ and a bipartite density operator $\rho_{AB}\in \mathcal{S}(\mathsf{A}\otimes \mathsf{B})$, we define
%quantum states such that $\mathtt{supp}(\rho_B)\subseteq\mathtt{supp}(\sigma_B)$,
\begin{align} \label{eq:total_state}
	\omega_{A_1\ldots A_M B} := \frac{1}{M}\sum_{m=1}^M
	\rho_{A_m B} \otimes \tau_{A_1} \otimes \tau_{A_2} \cdots \otimes \tau_{A_{m-1}} \otimes \tau_{A_{m+1}} \otimes \cdots \otimes \tau_{A_M} \in \mathcal{S}(\mathsf{A}_1\otimes \cdots \otimes \mathsf{A}_M \otimes \mathsf{B}),
\end{align}
where for each $m\in[M]:=\{1,\ldots,M\} $, $\rho_{A_m B} = \rho_{AB}$ and $\tau_{A_m} = \tau_A$.
Without loss of generality, we will further assume that the density operator $\tau_A$ is invertible; otherwise, we can always restrict the associated Hilbert space $\mathsf{A}$ to the support of $\tau_A$.

Throughout this paper, the following quantities in terms of the trace distance are considered as an error criterion for convex splitting:
\begin{align}
\begin{split} \label{eq:covering}
\Delta^\texttt{c}_M\left( \rho_{AB} \, \Vert\, \tau_A\right) &:=
\frac12\left\| \omega_{A_1\ldots A_M B} - \tau_{A}^{\otimes M} \otimes \rho_B \right\|_1; \\
\Delta^\texttt{c}_M\left(A:B\right)_{\rho} &:=
\Delta^\texttt{c}_M\left( \rho_{AB} \, \Vert\, \rho_A\right).
\end{split}
\end{align}
Here, the superscript `$\texttt{c}$' stands for ``covering''.

% Regarding \eqref{eq:Q1}, we study whether the error, $\Delta^\texttt{c}_M\left( \rho_{AB} \, \Vert\, \tau_A\right)$, can be upper bounded by an exponential decaying quantity and when is the associated (one-shot) error exponent positive?
% We term this as \emph{achievable one-shot error exponent} for convex splitting.
% On the other hand, we call an \emph{one-shot strong converse} if the error can be lower bounded by one minus an exponential decaying quantity.

% Regarding \eqref{eq:Q2}, the smallest integer $M$ for achieving $\Delta^\texttt{c}_M\left( \rho_{AB} \, \Vert\, \tau_A\right)\leq \eps$ is called the \emph{sample complexity} of an $\eps$-convex-splitting.
% We derive a one-shot characterization such sample complexity using so-called hypothesis-testing information. The established characterization is \emph{near optimal} as the matched second-order rate in the i.i.d. asymptotic scenario is also obtained, wherein the underlying states are $\rho_{AB}^{\otimes n}$ and $\tau_A^{\otimes n}$.

Our approach is to formulate the error, $\Delta^\texttt{c}_M\left( \rho_{AB} \, \Vert\, \tau_A\right)$, as Kosaki's \emph{weighted noncommutative  $L_p$ norm} \cite{Kos84}, and then adopt proper functional analytic tools for convex splitting.
We start with rewriting the quantities in \eqref{eq:covering} using the noncommutative $L_p$-norm.
For every density operator $\tau_A$, we define the following maps for  $m\in[M]$:
%and positive semi-definite operators $X_{AB} \geq 0$ and:
\begin{align}
	\begin{split} \label{eq:Theta}
	&\pi_m:  X_{AB} \mapsto \mathds{1}_{A_1}\otimes \cdots \otimes \mathds{1}_{A_{m-1}} \otimes  X_{A_m B} \otimes \mathds{1}_{A_{m+1}} \otimes \cdots \otimes \mathds{1}_{A_M};\\
	&\mathrm{E}_m : X_{AB} \mapsto \mathds{1}_{A_1}\otimes \cdots \otimes \mathds{1}_{A_{m-1}} \otimes \Tr_{A_m}\left[ X_{A_m B} \tau_{A_m} \right] \otimes  \mathds{1}_{A_m} \otimes \mathds{1}_{A_{m+1}} \otimes \cdots \otimes \mathds{1}_{A_M};\\
	&\Theta:=
	\frac{1}{M}\sum_{m=1}^M \left( \pi_m - \mathrm{E}_m \right).
	\end{split}
\end{align}
Given $p\geq 1$ and $\gamma\in[0,1]$, we introduce the following \emph{asymmetric noncommutative weighted $L_p$-norm} with respective to a density operator $\sigma$ and the associated \emph{noncommutative weighted $L_p$-space} \cite{Kos84, JX03} as
\begin{align}
	\begin{split} \label{eq:norm}
	\left\|X\right\|_{p,\gamma,\sigma} &:= \left( \Tr\left[ \left| \sigma^{\frac{1-\gamma}{p}} X \sigma^{\frac{\gamma}{p}} \right|^p \right] \right)^{\sfrac1p};\\
	L_{p, \gamma} (\sigma)&:= \left\{ X: \|X\|_{p, \gamma, \sigma} < \infty  \right\}.
	\end{split}
\end{align}
%and the following \emph{noncommutative weighted $L_p$-space} with respective to $\rho_A\otimes \sigma_B$ as the set of operators with finite noncommutative weighted $L_p$-norm, i.e.~
%\begin{align}
%	L_{p, \rho_A\otimes \sigma_B} := \left\{ X: \|X\|_{p, \rho_A\otimes \sigma_B} < \infty  \right\}.
%\end{align}
For two positive operators $X$ and $Y$ %with $\mathtt{supp}(X)\subseteq\mathtt{supp}(Y)$,
we introduce the notation of \emph{asymmetric noncommutative quotient for $X$ over $Y$} as\footnote{If $Y$ is not invertible, we then take the Moore--Penrose pseudo-inverse of $Y$ instead.}
\begin{align} \label{eq:quotient}
	\frac{X}{Y} := Y^{\gamma-1} X Y^{-\gamma}.
\end{align}
This notation of noncommutative quotient depends on the parameter $\gamma \in[0,1]$. To ease the burden of notation, the parameter $\gamma$ will be omitted whenever there is no confusion.

Given $M\in\mathds{N}$ and density operators $\rho_{AB}$ and $\tau_A > 0$,
our key observation is the following formulation: 
 for \emph{any} density operator $\sigma_B \in \mathcal{S}(\mathsf{B})$:
\begin{align} \label{eq:formulation}
\boxed{
	\;\;\left\| \omega_{A_1\ldots A_M B} - \tau_A^{\otimes M}\otimes \rho_{B} \right\|_1
	= \left\| \Theta\left( \frac{\rho_{AB}}{\tau_A\otimes \sigma_B} \right) \right\|_{1,\gamma, \tau_A^{\otimes M}\otimes \sigma_B}.\;\;
 }
\end{align}
Note that the above identity is for any $\gamma \in [0,1]$ with the corresponding asymmetric noncommutative quotient $\frac{\rho_{AB}}{\tau_A\ten \sigma_B}$ (which also depends on $\gamma$), and for any $\sigma_B\in\mathcal{S}(\mathsf{B})$ in the weight.

\begin{remark}
For a classical-quantum state $\rho_{XB} = \sum_{x\in\mathcal{X}} p_X(x) |x\rangle \langle x| \otimes \rho_B^x$, by substituting $A \leftarrow X$, i.e.~a classical system with probability distribution $p_X$, and $\tau_A \leftarrow \rho_X$ in \eqref{eq:formulation},
	then the convex splitting reduces to the \emph{quantum soft covering} studied in \cite{CG22}:
	\begin{align}
	\Delta^\texttt{c}_M\left(X:B\right)_{\rho} = \frac12\left\| \omega_{X_1\ldots X_M B} - \rho_X^{\otimes M} \otimes \rho_{B} \right\|_1 = \mathds{E}_{x_m \sim p_X} \frac12 \left\| \frac1M \sum_{m\in[M]} \rho_B^{x_m} - \rho_B \right\|_1.
	\end{align}
\end{remark}

Our first result is the following map norm of convex splitting.
\begin{lemm}[Map norm of convex splitting] \label{lemm:key}
	For any $M\in \mathds{N}$ and density operator $\tau_A$, let $\Theta$ be defined as in \eqref{eq:Theta}. Then, for every density operator $\sigma_B$ and $\gamma\in[0,1]$, we have
	\begin{align}
		\left\| \Theta: L_{p, \gamma}(\tau_A \otimes \sigma_B) \to L_{p, \gamma}(\tau_A^{\otimes M}\otimes \sigma_B) \right\| \leq 2^{\frac{2}{p}-1} M^{ \frac{1-p}{p} }, \quad \forall \, p\in[1,2].
	\end{align}
\end{lemm}
Given the formulation \eqref{eq:formulation}, Lemma~\ref{lemm:key} provides an upper bound for the error $\Delta^\texttt{c}_M\left( \rho_{AB} \, \Vert\, \tau_A\right)$, which leads to achievability bounds (or called direct bound) in Sections~\ref{sec:direct_bound} and \ref{sec:error_exponent}.
The proof of Lemma~\ref{lemm:key} employs Kosaki's complex interpolation relation of the weighted $L_p$ space \cite{Kos84}.
We defer the proof to Section~\ref{sec:proofs_one-shot}.

The second result is a one-shot converse bound, namely, a lower bound to the error $\Delta^\texttt{c}_M\left( \rho_{AB} \, \Vert\, \tau_A\right)$.
\begin{lemm}
	[One-shot converse] \label{lemm:converse}
	For any density operators $\rho_{AB}$, $\tau_A$, and $\omega_{A_1\ldots A_M B}$ defined in \eqref{eq:total_state},
	we have, for every $M \in \mathds{N}$ and $c>0$,
	\begin{align} % \label{eq:converse0}
		\frac12\left\| \omega_{A_1\ldots A_M B} - \tau_A^{\otimes M} \otimes \rho_{B} \right\|_1 \geq
		1 - (1+c) \Tr\left[ \rho_{AB} \wedge (1+c^{-1}) M \tau_A \otimes \rho_B \right].
	\end{align}
	Here, ``$ A \wedge B = \frac12(A+B-|A-B|)$" denotes the noncommutative minimal between $A$ and $B$ (see \eqref{eq:minimal} in Fact~\ref{fact:minimal} of Section~\ref{sec:Preliminaries}).
\end{lemm}
The proof of Lemma~\ref{lemm:converse} is presented in Section~\ref{sec:proofs_one-shot}.

% \medskip
% \noindent\textit{Organization.} The rest of the paper is arranged as follows.
% Section~\ref{sec:Preliminaries} includes preliminaries facts that will be used in the analysis.
% Section~\ref{sec:proofs_one-shot} provides proofs of Lemmas~\ref{lemm:key} and \ref{lemm:converse}.
% Section~\ref{sec:sample} devotes to a near-optimal characterization regarding Question~\eqref{eq:Q2}.
% Section~\ref{sec:exponet} contains a one-shot error exponent and one-shot strong converse regarding Question~\refeq{eq:Q1}.
% In Section~\ref{sec:app}, we list applications of the established refined analysis of convex splitting.

\subsection{Preliminaries} \label{sec:Preliminaries}
We collect the preliminary facts that will be used in our discussion
\begin{fact}[Complex interpolation of Noncommutative $L_p$-space \cite{Kos84}] \label{fact:norm}
	Fix a $\gamma\in[0,1]$ and an invertible density operator $\sigma$. The noncommutative weighted $L_p$ spaces satisfies the follow interpolation relation
\[ L_{p_\theta,\gamma}(\sigma)=[L_{p_0,\gamma}(\sigma), L_{p_1,\gamma}(\sigma) ]_\theta\]
for any $1\le p,p_0,p_1\le \infty$ , $\theta\in [0,1]$ and $\frac{1}{p}=\frac{1-\theta}{p_0}+\frac{1}{p_1}$.
In particular, the following interpolation inequality holds:
	\begin{align}
		\left\|X\right\|_{p_\theta,\gamma,\sigma} \leq
		\left\|X\right\|_{p_0,\gamma,\sigma}^{1-\theta} \left\|X\right\|_{p_1,\gamma,\sigma}^{\theta}. %, \quad \forall \theta \in[0,1], \frac{1-\theta}{p_0} + \frac{\theta}{p_1} = \frac{1}{p_\theta}.
	\end{align}
As a consequence, for every operator $X$, the map
	\begin{align}
		p\mapsto \left\|X\right\|_{p,\gamma,\sigma}
	\end{align}
	is non-decreasing over $p\in [1,\infty]$.
\end{fact}

\begin{fact}[Properties of noncommutative quotient] \label{fact:quotient}
	Consider noncommutative quotient defined in \eqref{eq:quotient} with $\gamma = \sfrac12$. The following hold.
\begin{enumerate}[(i)]
	\item\label{item:quotient_linear} $\displaystyle\frac{A+B}{C} = \frac{A}{C} + \frac{B}{C}$ for all $A,B\geq 0$ and $C>0$;

	\item\label{item:quotient_positive} $\displaystyle \frac{A}{B} \geq 0$ for all $A\geq 0$ and $B>0$;

	% \begin{align}
	% \frac{A+B}{C} = \frac{A}{C} + \frac{B}{C}
	% \end{align}
	\item\label{item:quotient_upper} $\displaystyle\frac{A}{A+B} \leq \mathds{1}$ for all $A\geq 0$ and $B>0$;
	%\item\label{item:quotient_cyclic} $\displaystyle\Tr\left[ A \frac{B}{C} \right] = \Tr\left[ B \frac{A}{C}\right]$ for all $A,B\geq 0$ and $C>0$.
\end{enumerate}
\end{fact}

\begin{fact}[Properties of noncommutative minimal {\cite{Hel67, Hol72, Hol78, ANS+08, AM14, Cheng2022}}] \label{fact:minimal}
	For any positive semi-definite operators $A$ and $B$, we define the noncommutative minimal of $A$ and $B$ as
	\begin{align}\label{eq:minimal}
		A \wedge B := \argmax\left\{ \Tr[X] : X\leq A, \, X\leq B  \right\}.
%		= \frac12\left(A + B-|A-B|\right).
	\end{align}
	It satisfies the following properties.
	\begin{enumerate}[(i)]
		\item\label{item:infimum} (Infimum representation) $\displaystyle \Tr\left[ A\wedge B \right] = \inf_{0\leq T\leq \mathds{1}} \Tr\left[ A(\mathds{1}-T) + BT \right]$.
		
		\item\label{item:closed-form} (Closed-form expression) $\displaystyle A\wedge B=  \tfrac{1}{2}\left(A+B - |A-B|\right)$.%}{2}$.
		
%		\vspace{1pt}
		
		\item\label{item:order} (Monotone increase in the Loewner ordering) $\displaystyle\Tr[A\wedge B] \leq \Tr[A'\wedge B']$ for $A\leq A'$ and $B\leq B'$.
		
%		\vspace{1pt}
		
		\item\label{item:monotone} (Monotone increase under positive trace-preserving maps) $\displaystyle\Tr[A\wedge B] \leq \Tr[\mathcal{N}(A)\wedge \mathcal{N}(B)]$ for any positive trace-preserving map $\mathcal{N}$.
		
%		\vspace{1pt}
		
		\item\label{item:concave} (Trace concavity) The map $(A,B)\mapsto \Tr[A\wedge B]$ is jointly concave.
		
%		\vspace{1pt}
		
		\item\label{item:direct} (Direct sum) $\displaystyle(A\oplus A')\wedge (B\oplus B') = (A\wedge B) \oplus (A'\wedge B')$ for any self-adjoint $A'$ and $B'$.
		
%		\vspace{1pt}
		
		\item\label{item:upper} (Upper bound)  $\displaystyle\Tr[A\wedge B] \leq \Tr[ A^{1-s} B^s]$ for any $A,B\geq 0$ and $s\in(0,1)$.
		
%		\vspace{1pt}		
		\item\label{item:lower} (Lower bound\footnote{In case that $A+B$ is not invertible, one just use the Moore--Penrose pseudo-inverse of $A+B$ in the definition of the noncommutative quotient \eqref{eq:quotient}.})
		$\displaystyle\Tr[A\wedge B] \geq \Tr\left[ A \frac{B}{A+B} \right]$ for any $A,B\geq 0$.
	\end{enumerate}
\end{fact}

\begin{fact}[Variational formula of the trace distance \cite{Hel67, Hol72}, {\cite[\S 9]{NC09}}] \label{fact:1norm}
    The following properties hold
	for any density operators $\rho$ and $\sigma$. % with $\tr(\rho)=\tr(\sigma)$,
	\begin{enumerate}[(i)]
	    \item\label{item:infimum} (Infimum representation)  
	    \begin{align}
		\frac{1}{2}\left\|\rho- \sigma\right\|_1 = \sup_{0\leq \Pi \leq \mathds{1}} \Tr\left[\Pi(\rho-\sigma)\right].
	    \end{align}
	
	\item\label{item:change} (Change of measure) For every test $0\leq \Pi\leq \mathds{1}$,
	\begin{align}
    \Tr[\rho\Pi] \leq \Tr[\sigma\Pi] + \frac{1}{2}\left\|\rho- \sigma\right\|_1.
\end{align}
	
	\end{enumerate}
	
\end{fact}

For a self-adjoint operator $X$, we have the spectrum decomposition $X = \sum_i \lambda_i e_i$, where $\lambda_i\neq \lambda_j$ are distinct eigenvalues and $e_i$ is the spectrum projection onto $i$th distinct eigenvalues.
We define the set $\textnormal{\texttt{spec}}(X):= \{\lambda_i\}$ to be the eigenvalues of $H$, and $|\textnormal{\texttt{spec}}(X)|$ to be the number of distinct eigenvalues of $H$. Recall that the \emph{pinching map} with respect to $X$ is defined as
\begin{align} \label{eq:pinching}
	\mathscr{P}_X : L\mapsto \sum_{i=1} e_i L e_i\, .
\end{align}

% Hayashi's \emph{pinching inequality} \cite{Hay02} is that
\begin{fact}[{Pinching inequality} {\cite{Hay02}}] \label{fact:pinching}
	For every positive semi-definite operator $L$,
	\begin{align} \label{eq:pinching_inequality}
		\mathscr{P}_H[L] \geq \frac{1}{|\textnormal{\texttt{spec}}(H)|} L\,.
	\end{align}
	Moreover, it holds that for every $d$-dimensional self-adjoint operator and $H$ and every $n\in\mathds{N}$,
	\begin{align}
		\left|\textnormal{\texttt{spec}}\left(H^{\otimes n}\right)\right| \leq (n+1)^{d-1}\,.
	\end{align}
\end{fact}

\begin{fact}[Uhlmann's theorem \cite{Uhl76}, \cite{FG99}, {\cite[Lemma 2.2]{DHW08}}] \label{fact:Uhlmann}
	Let $|\psi\rangle_{AB}$ and $|\varphi\rangle_{AC}$ be two pure quantum states.
	Then, there exists an isometry $\mathcal{V}_{B\to C}$ satisfying the following:
	\begin{enumerate}[(i)]
		\item\label{item:purified}	$\displaystyle \mathrm{P}(\psi_{A}, \varphi_{A}) =
		\mathrm{P}(\mathcal{V}_{B\to C}( \psi_{AB} ), \varphi_{AC} )$.
		
		\item\label{item:trace} $\displaystyle \frac12\| \psi_{A} - \varphi_{A} \|_1 \leq \varepsilon \Rightarrow
		\frac12\| \mathcal{V}_{B\to C}( \psi_{AB} ) - \varphi_{AC} \|_1 =
		\mathrm{P}(\mathcal{V}_{B\to C}( \psi_{AB} ), \varphi_{AC} ) \leq \sqrt{ 2\varepsilon - \varepsilon^2 }$.
	\end{enumerate}
	Here, $\mathrm{P}(\rho, \sigma) := \sqrt{ 1 - \|\sqrt{\rho}\sqrt{\sigma}\|_1}$ is the purified distance \cite{TCR10}.
\end{fact}

%\begin{fact}[Invariance of trace distance under c-q channel] \label{fact:invariance}
%	Let $x\mapsto \rho_B^x$ be a classical-quantum channel.
%	For any two probability distributions $p_X$ and $q_X$, we have
%	\begin{align}
%		\left\| \sum_{x\in\mathcal{X}} p_X(x) |x\rangle \langle x| - q_X(x) |x\rangle \langle x | \right\|_1
%		= \left\| \rho_{XB} - \sigma_{XB} \right\|_1,
%	\end{align}
%	where $\rho_{XB} := \sum_{x\in\mathcal{X}} p_X(x) |x\rangle \langle x|\otimes \rho_B^x$ and $\sigma_{XB} := \sum_{x\in\mathcal{X}} q_X(x) |x\rangle \langle x|\otimes \rho_B^x$.	
%\end{fact}
\begin{fact}[Invariance of trace distance] \label{fact:invariance}
	Let $\rho_{XB} := \sum_{x\in\mathcal{X}} p_X(x) |x\rangle \langle x|\otimes \rho_B^x$ be a classical-quantum state with marginals $\rho_{X} := \sum_{x\in\mathcal{X}} p_X(x) |x\rangle \langle x|$ and $\rho_{B} := \sum_{x\in\mathcal{X}} p_X(x) \rho_B^x$.
	Then,
	\begin{align}
		\left\| \rho_{XB} - |x\rangle\langle x|\otimes \rho_B^x \right\|_1
		= \left\|\rho_X - |x\rangle\langle x| \right\|_1, \quad \forall\, x\in\mathcal{X}.
	\end{align}
\end{fact}
\begin{proof}
	This can be proved by inspection and the direct-sum structure of trace-norm:
	\begin{align}
			\left\| \rho_{XB} - |x\rangle\langle x|\otimes \rho_B^x \right\|_1
			&= \left\| p(x) \rho_B^x - \rho_B \right\|_1+\sum_{\bar{x}\neq x} p_X(\bar{x})  \\
			&= \left\| \sum_{\bar{x}\neq x} p_X(\bar{x}) \rho_B^{\bar{x}} \right\|_1+\sum_{\bar{x}\neq x} p_X(\bar{x})  \\
			&= 1 - p_X(x)+\sum_{\bar{x}\neq x} p_X(\bar{x})  \\
			&=\left\|\rho_X - |x\rangle\langle x| \right\|_1.
 	\end{align}
 	We note that this property might have been known in literature; see e.g.~\cite[\S 4.2]{Wil17b}.
\end{proof}

% \begin{fact}[Change of measure\footnote{We remark that Fact~\ref{lemm:1norm} is equivalent to the variational formula of the trace distance \cite{Hel67, Hol72}, {\cite[\S 9]{NC09}}.} \cite{Hel67, Hol72}, {\cite[\S 9]{NC09}}]\label{lemm:1norm}
% For arbitrary density operators $\rho$ and $\sigma$, and a test $0\leq \Pi\leq \mathds{1}$, the following holds:
% \begin{align}
%     \Tr[\rho\Pi] \leq \Tr[\sigma\Pi] + \frac{1}{2}\left\|\rho- \sigma\right\|_1.
% \end{align}
% % \begin{align}
% % \frac{1}{2}\left\|\rho- \sigma\right\|_1 = \sup_{0\leq \Pi \leq \mathds{1}} \Tr\left[\Pi(\rho-\sigma)\right].
% % \end{align}
% \end{fact}

\subsection{Proofs of One-Shot Bounds for Convex Splitting} \label{sec:proofs_one-shot}
\begin{lemm1}[Map norm of convex splitting] %\label{lemm:key}
	For any $M\in \mathds{N}$ and density operator $\tau_A$, let $\Theta$ be defined as in \eqref{eq:Theta}. Then, for every density operator $\sigma_B$ and $\gamma\in[0,1]$, we have
	\begin{align}
		\left\| \Theta: L_{p, \gamma}(\tau_A \otimes \sigma_B) \to L_{p, \gamma}(\tau_A^{\otimes M}\otimes \sigma_B) \right\| \leq 2^{\frac{2}{p}-1} M^{ \frac{1-p}{p} }, \quad \forall \, p\in[1,2].
	\end{align}
\end{lemm1}
\begin{proof}[Proof of Lemma~\ref{lemm:key}]
	Fix an arbitrary $\gamma\in[0,1]$ throughout this proof.
For $p=1$,
	using triangle inequality we have
	\begin{align}
		&\left\| \Theta: L_{1, \gamma}(\tau_A\otimes \sigma_B) \to L_{1,\gamma}(\tau_A^{\otimes M}\otimes \sigma_B) \right\| \notag\\
		&\leq \frac1M\sum_{m=1}^M \| \pi_m: L_{1, \gamma}(\tau_A\otimes \sigma_B) \to L_{1,\gamma}(\tau_A^{\otimes M}\otimes \sigma_B)  \|
			+ \| \mathrm{E}_m: L_{1, \gamma}(\tau_A\otimes \sigma_B) \to L_{1,\gamma}(\tau_A^{\otimes M}\otimes \sigma_B) \| \\
		&\leq 2, \label{eq:L1}
	\end{align}
	since for every $m\in [M]$,
	\begin{align}
		\left\| \pi_m: L_{1, \gamma}(\tau_A\otimes \sigma_B) \to L_{1,\gamma}(\tau_A^{\otimes M}\otimes \sigma_B) \right\|
		&= \sup_{X_{AB} \neq 0} \frac{\norm{X_{A_m B} \otimes \mathds{1}_A^{[M]\backslash \{m\}}}{L_{1,\gamma}(\tau_A^{\otimes M}\otimes \sigma_B)} }{ \norm{ X_{AB} }{L_{1, \gamma}(\tau_A\otimes \sigma_B)} } = 1;\\
		\left\| E_m: L_{1, \gamma}(\tau_A\otimes \sigma_B) \to L_{1,\gamma}(\tau_A^{\otimes M}\otimes \sigma_B) \right\|
		&= \sup_{X_{AB} \neq= 0} \frac{\norm{ \mathds{1}_{A_m}\ten \Tr_{A_m}\left[ \tau_{A_m}X_{A_m B}   \right]}{L_{1,\gamma}(\tau_A^{\otimes M}\otimes \sigma_B)}  }{ \norm{ X_{AB} }{L_{1, \gamma}(\tau_A\otimes \sigma_B)} } \le  1.
	\end{align}
For $p=2$,
	define the inner product \[\langle X, Y\rangle_{\gamma,\sigma} := \Tr\left[ X^\dagger \sigma^{1-\gamma} Y \sigma^\gamma\right]\
\, .\]
	For every $m\in [M]$, write $\mathring{\pi}_m := \pi_m - \mathrm{E}_m$.
	Then, for any $X_{AB},Y_{AB}$ and denote $\mathring{X}_{AB} := X_{AB} - \Tr_A[X_{AB}\tau_A]\otimes \mathds{1}_B$ and $\mathring{Y}_{AB} := Y_{AB} - \Tr_A[Y_{AB}\tau_A]\otimes \mathds{1}_A$,
	we calculate that for $m\neq \bar{m}$,
	\begin{align}
		&\langle \mathring{\pi}_m(X_{AB}), \mathring{\pi}_{\bar m}(Y_{AB}) \rangle_{\gamma, \tau_A^{\otimes M} \otimes \sigma_B} \\
		&= \Tr\left[ \mathds{1}_{A}^{\otimes (m-1)} \otimes \mathring{X}_{A_m B}^\dagger \otimes \mathds{1}_{A}^{\otimes (M-m)}
		\left(\tau_A^{\otimes M} \otimes \sigma_B \right)^{1-\gamma}
		\mathds{1}_{A}^{\otimes (\bar{m}-1)} \otimes \mathring{Y}_{A_{\bar m} B}\otimes \mathds{1}_{A}^{\otimes (M-\bar{m})}
		\left( \tau_A^{\otimes M} \otimes \sigma_B \right)^\gamma
		\right]\\
		&\overset{\textnormal{(a)}}{=} \Tr \left[ \mathds{1}_{A_{\bar m}} \otimes \mathring{X}_{A_m B}^\dagger \left(\tau_{A_{\bar m}} \otimes {\sigma_B} \otimes {\tau_{A_m}} \right)^{1-\gamma}
		\mathring{Y}_{A_{\bar m} B} \otimes \mathds{1}_{A_m} \left(\tau_{A_{\bar m}} \otimes {\sigma_B} \otimes {\tau_{A_m}} \right)^{\gamma} \right] \\
		&= \Tr \left[ \mathring{X}_{A_m B}^\dagger \left({\sigma_B} \otimes {\tau_{A_m}} \right)^{1-\gamma} \Tr_{A_{\bar m}}\left[\mathring{Y}_{A_{\bar m} B}\tau_{A_{\bar m}} \right] \otimes \mathds{1}_{A_{m}} \left({\sigma_B} \otimes {\tau_{A_m}} \right)^{\gamma} \right] \\
		&\overset{\textnormal{(b)}}{=} 0, \label{eq:Lp_1}
	\end{align}
	where in (a) we trace out all the $\tau_{A_{\tilde m}}$ for $\tilde{m}\neq m, \bar{m}$;
	and in (b) we compute the term $\Tr_{A_{\bar m}}\left[\mathring{Y}_{A_{\bar m} B}\tau_{A_{\bar m}} \right]$ as:
	\begin{align}
		\Tr_{A_{\bar m}}\left[\mathring{Y}_{A_{\bar m} B}\tau_{A_{\bar m}} \right]
		&=  \Tr_{A}\left[Y_{AB}\tau_{A} - \Tr_A [Y_{AB} \tau_A] \otimes \tau_A\right]
		= 0.
	\end{align}
	
	Then, for any $X_{AB}$,
	\begin{align}
		\left\| \frac1M \sum_{m=1}^M \mathring{\pi}_m \left( X_{AB} \right) \right\|_{2, \gamma, \tau_A^{\otimes M} \otimes \sigma_B}^2
		&= \frac{1}{M^2} \sum_{m, \bar{m}=1}^M  \left\langle \mathring{\pi}_m \left( X_{AB} \right), \mathring{\pi}_{\bar m}\left( X_{AB} \right) \right\rangle_{\gamma, \tau_A^{\otimes M} \otimes \sigma_B} \\
		&\overset{\textnormal{(a)}}{=} \frac{1}{M^2} \sum_{m=1}^M \left\| \mathring{\pi}_m \left( X_{AB} \right) \right\|_{2, \gamma, \tau_A^{\otimes M} \otimes \sigma_B}^2 \\
		&\overset{\textnormal{(b)}}{=} \frac{1}{M^2} \sum_{m=1}^M \left\|  X_{AB} - \Tr_A\left[X_{AB} \tau_A \right]\otimes \mathds{1}_A \right\|_{2, \gamma, \tau_A\otimes \sigma_B}^2 \\
		&= \frac{1}{M} \left\|  X_{AB} - \Tr_A\left[X_{AB} \tau_A \right]\otimes \mathds{1}_A \right\|_{2, \gamma, \tau_A\otimes \sigma_B}^2 \\
		&\overset{\textnormal{(c)}}{\leq}  \frac{1}{M} \left\|  X_{AB}\right\|_{2, \gamma, \tau_A\otimes \sigma_B}^2,
	\end{align}
	where (a) follows from \eqref{eq:Lp_1};
	(b) follows from the definitions of $\mathring{\pi}_m$, ${\pi}_m$, and the noncommutative weighted norm;
	and the last inequality (c) follows from the fact that  $E_A(X_{AB})= \Tr_A\left[X_{AB} \tau_A \right]\otimes \mathds{1}_A$ is a projection for the $L_{2,\gamma}$ norm for any $\gamma\in[0,1]$.
	Hence, we obtain
	\begin{align}
		\left\| \Theta: L_{2, \gamma}(\tau_A\otimes \sigma_B) \to L_{2, \gamma}(\tau_A^{\otimes M} \otimes \sigma_B) \right\|
		\leq \frac{1}{\sqrt{M}}. \label{eq:L2}
	\end{align}
The case of general $p\in[1,2]$ follows from complex interpolation with $\theta = \frac{2(p-1)}{p} \in [0,1]$ together with \eqref{eq:L1} and \eqref{eq:L2},
	\begin{align}
		&\left\| \Theta: L_{p, \gamma}(\tau_A\otimes \sigma_B) \to L_{p, \gamma} (\tau_A^{\otimes M} \otimes \sigma_B) \right\|\\
		&\leq \left\| \Theta: L_{1, \gamma}( \rho_A\otimes \sigma_B) \to L_{1, \gamma}( \tau_A^{\otimes M} \otimes \sigma_B) \right\|^{1-\theta} \left\| \Theta: L_{2, \gamma}(\tau_A\otimes \sigma_B) \to L_{2, \gamma}( \rho_A^{\otimes M} \otimes \sigma_B) \right\|^\theta, \\
		&= 2^{1- \frac{2(p-1)}{p} } M^{ \frac{1-p}{p} },
	\end{align}
	which completes the proof.
\end{proof}

\begin{lemm2}
	[One-shot converse] %\label{lemm:converse}
	For any density operators $\rho_{AB}$, $\tau_A$, and $\omega_{A_1\ldots A_M B}$ defined in \eqref{eq:total_state},
	we have, for every $M \in \mathds{N}$ and $c>0$,
	\begin{align} \label{eq:converse0}
		\frac12\left\| \omega_{A_1\ldots A_M B} - \tau_A^{\otimes M} \otimes \rho_{B} \right\|_1 \geq
		1 - (1+c) \Tr\left[ \rho_{AB} \wedge (1+c^{-1}) M \tau_A \otimes \rho_B \right].
	\end{align}
	Here, `$ A \wedge B = \frac12(A+B-|A-B|)$' denotes the noncommutative minimal between $A$ and $B$ (see \eqref{eq:minimal} in Fact~\ref{fact:minimal} of Section~\ref{sec:Preliminaries}).
\end{lemm2}

%\begin{align}
%	\frac12 \mathds{E}_{x\sim p_X }\left\| \frac1M \sum\nolimits_{x\in \mathcal{C}} \rho_B^x -  \rho_{B} \right\|_1 \geq
%	1 - (1+c) \Tr\left[ \rho_{XB} \wedge (1+c^{-1}) M \rho_{X} \otimes \rho_B \right], \quad \forall c>0.
%\end{align}

\begin{proof}[Proof of Lemma~\ref{lemm:converse}]
	We fix $\gamma = \sfrac12$ in the definition of the noncommutative quotient in this proof.
	Let
	\begin{align}
		\Pi := \frac{\omega_{A_1\ldots A_M B}}{ \omega_{A_1\ldots A_M B} + c^{-1} \tau_A^{\otimes n} \otimes \rho_B },
	\end{align}
	and for every $m\in[M]$, we denote a density operator:
	\begin{align}
		\omega^{(m)} :=
		%		\pi_m\left( \frac{\rho_{AB}}{\tau_A\otimes \rho_B } \right) =
		\rho_{A_m B} \otimes \tau_{A_1} \otimes \tau_{A_2} \otimes  \cdots \otimes \tau_{A_{m-1}} \otimes \tau_{A_{m+1}} \otimes \cdots \otimes \tau_{A_M}.
	\end{align}
	It is clear that, by tracing out all the registers of system $A$ except of the $m$-th one, we obtain
	\begin{align} \label{eq:partial_trace}
		\Tr_{ A_1\ldots A_{m-1} A_{m+1} \ldots A_M  } \left[ \omega^{(\bar{m})} \right] =
		\begin{dcases}
			\rho_{A_m B} & \bar{m} = m \\
			\tau_{A_m}\otimes \rho_B & \bar{m} \neq m
		\end{dcases}.
	\end{align}
	Using the variational formula of the Schatten $1$-norm (Fact~\ref{fact:1norm}), we have
	\begin{align}
		\frac12\left\| \omega_{A_1\ldots A_M B} - \tau_A^{\otimes M} \otimes \rho_{B} \right\|_1
		&\geq \Tr\left[ \omega_{A_1\ldots A_M B} \Pi \right] - \Tr\left[ \tau_A^{\otimes M} \otimes \rho_{B} \Pi \right]. \label{eq:converse_12}
	\end{align}
	
	To lower bound the first term on the right-hand side of \eqref{eq:converse_12}, we make the following calculation:
	%	the calculation using the notation introduced in \eqref{eq:Theta}:
	\begin{align}
		\Tr\left[ \omega_{A_1\ldots A_M B} \Pi \right]
		&= \frac{1}{M}\sum_{m=1}^M \Tr\left[ \omega^{(m)} \frac{ \frac1M\sum_{\bar{m}=1}^M \omega^{(\bar{m})} }{ \omega_{A_1\ldots A_M B} + c^{-1} \tau_A^{\otimes M} \otimes \rho_{B} } \right] \\
		&\overset{\text{(a)}}{\geq} \frac{1}{M}\sum_{m=1}^M \Tr\left[ \omega^{(m)} \frac{ \frac1M\omega^{(m)} }{ \omega_{A_1\ldots A_M B} + c^{-1} \tau_A^{\otimes M} \otimes \rho_{B} } \right] \\
		&= \frac{1}{M}\sum_{m=1}^M \Tr\left[ \omega^{(m)} \frac{  \omega^{(m)} }{ \omega^{(m)} + \sum_{\bar{m}\neq m} \omega^{(\bar{m})} + c^{-1} M \tau_A^{\otimes M} \otimes \rho_{B} } \right] \\
		&\overset{\text{(b)}}{\geq} 1 - \frac{1}{M}\sum_{m=1}^M \Tr\left[ \omega^{(m)} \wedge \left( \sum_{\bar{m}\neq m} \omega^{(\bar{m})} + c^{-1} M \tau_A^{\otimes M} \otimes \rho_{B} \right) \right] \\
		&\overset{\text{(c)}}{\geq} 1 - \frac{1}{M}\sum_{m=1}^M \Tr\left[ \rho_{A_m B} \wedge \left( (M-1) \tau_{A_m} \otimes \rho_B + c^{-1} M \tau_{A_m} \otimes \rho_{B} \right) \right], \\
		&\overset{\text{(d)}}{\geq}  1 - \Tr\left[ \rho_{A B} \wedge \left( (1+c^{-1})M \tau_{A} \otimes \rho_B  \right) \right],
		\label{eq:converse1_0}
	\end{align}
	where (a) follows from linearity (Fact~\ref{fact:quotient}-\ref{item:quotient_linear}) and the positivity (Fact~\ref{fact:quotient}-\ref{item:quotient_positive});
	(b) follows from the relation to noncommutative minimal (Fact~\ref{fact:minimal}-\ref{item:lower});
	and in (c) we applied partial trace in \eqref{eq:partial_trace} and (d) used monotonicity of noncommutative minimal (Fact~\ref{fact:minimal}-\ref{item:order}).
	
Using identical argument, we bound the second term on the right-hand side of \eqref{eq:converse_12} as follows:
	\begin{align}
		- \Tr\left[ \tau_A^{\otimes M}\otimes \rho_{B} \Pi \right]
		&= - \frac{1}{M} \sum_{m=1}^M \Tr\left[ \omega^{(m)} \frac{\tau_A^{\otimes M}\otimes \rho_{B}}{ \omega_{A_1\ldots A_M B} + c^{-1}\tau_A^{\otimes M}\otimes \rho_{B}  }\right] \\
		&= - \frac{c}{M}\sum_{m=1}^M \Tr\left[\omega^{(m)} \frac{ c^{-1} M \tau_A^{\otimes M}\otimes \rho_{B} }{ \omega^{(m)} + \sum_{\bar{m}\neq m} \omega^{(\bar{m})} + c^{-1} M \tau_A^{\otimes M}\otimes \rho_{B} } \right] \\
		&\overset{\text{(a)}}{\geq} 	-\frac{c}{M}\sum_{m=1}^M \Tr\left[\omega^{(m)} \frac{\sum_{\bar{m}\neq m} \omega^{(\bar{m})} + c^{-1} M \tau_A^{\otimes M}\otimes \rho_{B} }{ \omega^{(m)} + \sum_{\bar{m}\neq m} \omega^{(\bar{m})} + c^{-1} M \tau_A^{\otimes M}\otimes \rho_{B} } \right] \\
		&\overset{\text{(b)}}{\geq} -\frac{c}{M}\sum_{m=1}^M \Tr\left[\omega^{(m)} \wedge \left( \sum_{\bar{m}\neq m} \omega^{(\bar{m})} + c^{-1} M \tau_A^{\otimes M}\otimes \rho_{B} \right) \right] \\
		&\overset{\text{(c)}}{\geq} -\frac{c}{M}\sum_{m=1}^M \Tr\left[ \rho_{A_m B} \wedge \left( (M-1) \tau_{A_m} \otimes \rho_B + c^{-1} M \tau_{A_m} \otimes \rho_{B} \right) \right], \\
		&\overset{\text{(d)}}{\geq} - c \Tr\left[ \rho_{A B} \wedge \left( (1+c^{-1})M \tau_{A} \otimes \rho_B  \right) \right],
		\label{eq:converse_2}\pl.
	\end{align}
	
	Combining \eqref{eq:converse_12}, \eqref{eq:converse1_0} and \eqref{eq:converse_2} yields the proof.
\end{proof}

\section{Sample Complexity (Bounds on Rate)} \label{sec:sample}
Given density operators $\rho_{AB}$, $\tau_A$, and $\omega_{A_1\ldots A_M B}$ defined in \eqref{eq:total_state},
we define the \emph{$\varepsilon$-sample complexity} as
\begin{align}
	M_\varepsilon^{\textnormal{\texttt{c}}}\left(\rho_{AB} \,\Vert\, \tau_A\right) := \inf\left\{  M \in \mathds{N}:
	\Delta_M^{\textnormal{\texttt{c}}}\left(\rho_{AB} \,\Vert\, \tau_A\right)  \leq \varepsilon \right\}.
\end{align}
We provide a direct bound (i.e.~upper bound) on $M_\varepsilon^{\textnormal{\texttt{c}}}\left(\rho_{AB} \,\Vert\, \tau_A\right)$ in Section~\ref{sec:direct_bound}, and a converse bound (i.e.~lower bound) in Section~\ref{sec:converse_bound}.
\subsection{Direct Bound} \label{sec:direct_bound}
\begin{theo} \label{theo:direct}
	Let $\rho_{AB}$ and $\tau_A$ be density operators, and let $\omega_{A_1\ldots A_M B}$ be defined in \eqref{eq:total_state}.
	Then, for every $\tau_A>0$, $\varepsilon\in(0,1)$ and $\delta \in (0, \sfrac{\varepsilon}{3})$, we have
	\begin{align} \label{eq:size}
		\log M_\varepsilon^{\textnormal{\texttt{c}}}\left(\rho_{AB} \,\Vert\, \tau_A\right) \leq
		I_\textnormal{h}^{1-\varepsilon + 3\delta}\left(\rho_{AB} \,\Vert\, \tau_A\right) + \log \tfrac{\nu^2}{\delta^4},
	\end{align}
	where $\nu := |\textnormal{\texttt{spec}}(\tau_A)||\mathsf{B}|$, and the generalized $\eps$-hypothesis-testing information and the $\eps$-hypothesis-testing divergence \cite{WR13, TH13, Li14} are defined as
	\begin{align}
		I_\textnormal{h}^{\varepsilon }\left(\rho_{AB} \,\Vert\, \tau_A\right) &:=	\inf_{\sigma_B \in \mathcal{S}(\mathsf{B})}	D_\textnormal{h}^{\varepsilon }\left(\rho_{AB}\,\|\,\tau_A\otimes \sigma_B\right); \label{eq:I_h}\\
		D_\textnormal{h}^{\varepsilon }(\rho\,\Vert\, \sigma) &:= \sup_{0\leq T\leq \mathds{1} } \left\{ -\log \Tr[\sigma T] : \Tr[\rho T] \geq 1 - \varepsilon \right\}. \label{eq:D_h}
	\end{align}
%	we have
%	\begin{align} \label{eq:one-norm}
%		\frac12\left\| \tau_{AB_1\ldots B_M} - \tau_A\otimes \sigma_{B}^{\otimes M} \right\|_1 \leq \varepsilon.
%	\end{align}
%	(Here, $\nu$ is a constant incurred from pinching.)
	
%	the logarithmic term $\log \frac{1}{\delta}$ could be something like $\log \frac{|\mathtt{spec}(\rho_B)|}{\delta^4}$.)
\end{theo}
\begin{proof}
We first claim that, for any $c>0$ and density operator $\sigma_B \in \mathcal{S}(\mathsf{B})$,
\begin{align} \label{eq:direct}
%	\frac12 \left\| \omega_{AB_1\ldots B_M} - \tau_A^{\otimes M} \otimes \rho_{B} \right\|_1
		\Delta_M^{\textnormal{\texttt{c}}}\left(\rho_{AB} \,\Vert\, \tau_A\right) \leq \Tr\left[\rho_{AB}\left\{ \mathscr{P}_{\tau_A\otimes \sigma_B}[\rho_{AB}] > c \tau_A\otimes \sigma_B \right\} \right] + \sqrt{ \frac{ c \nu }{M} },
\end{align}
where where $\{A>B\}$ denotes an projection onto the eigenspaces corresponding to the positive part of $A-B$;
$\mathscr{P}_{\tau_A\otimes \sigma_B}$ is the pinching map of $\tau_A\otimes \sigma_B$.
Then, for every $\delta \in (0,\varepsilon)$ and we choose
\begin{align}
	c = \exp\left\{D_\text{s}^{1-\varepsilon + \delta} \left( \mathscr{P}_{\tau_A\otimes\sigma_B}[\rho_{AB}]  \,\|\, \tau_A\otimes \sigma_B \right) + \xi \right\}
\end{align}
for some small $\xi>0$. Recall that the \emph{$\varepsilon$-information spectrum divergence} \cite{HN03, NH07} is defined as
\begin{align}
	D_\text{s}^{\varepsilon}(\rho \parallel \sigma) &:= \sup_{c\in\mathds{R}} \left\{ \log c : \Tr\left[ \rho \left\{ \rho \leq c \sigma \right\} \right] \leq \varepsilon  \right\}. \label{eq:Ds}
% 	\\
% 	&{\color{Red}{=} }
% 	\inf_{c\in\mathds{R}} \left\{ \log c : \Tr\left[ \rho \left\{ \rho > c \sigma \right\} \right] \leq 1- \varepsilon  \right\}. 
\end{align}
By definition of $D_\text{s}^{\varepsilon}$, we have
\begin{align} \label{eq:upper_1}
	\Tr\left[\rho_{AB}\left\{ \mathscr{P}_{\tau_A\otimes\sigma_B}[\rho_{AB}] > c \tau_A\otimes \sigma_B \right\} \right] =
	\Tr\left[\mathscr{P}_{\tau_A\otimes\sigma_B}[\rho_{AB}]\left\{ \mathscr{P}_{\tau_A\otimes\sigma_B}[\rho_{AB}] > c \tau_A\otimes \sigma_B \right\} \right]
	< \varepsilon - \delta.
\end{align}
Letting
\begin{align}
	M = \ceil[\Big]{ {c \nu }{\delta^{-2}} }, \label{eq:M}
\end{align}
we obtain
\begin{align}
%	\frac12 \left\| \omega_{AB_1\ldots B_M} - \rho_A\otimes \sigma_{B}^{\otimes M} \right\|_1
	\Delta_M^{\textnormal{\texttt{c}}}\left(\rho_{AB} \,\Vert\, \tau_A\right) \leq \varepsilon,
\end{align}
and then by \eqref{eq:M} and invoking the relation of $D_\text{s}^{\varepsilon}$ and $D_\text{h}^{\varepsilon}$ \cite{TH13}, we have the following upper bound on $\log M_\varepsilon^{\textnormal{\texttt{c}}}\left(\rho_{AB} \,\Vert\, \tau_A\right) $:
\begin{align}
	\log M_\varepsilon^{\textnormal{\texttt{c}}}\left(\rho_{AB} \,\Vert\, \tau_A\right) &\leq D_\text{s}^{1-\varepsilon + \delta} \left( \mathscr{P}_{\tau_A\otimes \sigma_{B}}[\rho_{AB}]  \,\|\, \tau_A\otimes \sigma_B \right) + \xi + \log \nu - 2 \log \delta \\
	&\leq D_\text{h}^{1-\varepsilon + 3\delta} \left( \rho_{AB}  \,\|\, \tau_A\otimes \sigma_B \right) + \xi + 2 \log \nu -  4 \log \delta,
\end{align}
where in the last line we apply the following relation to remove the pinching operator $\mathscr{P}_{\tau_A\otimes \sigma_{B}}$ and translate $D_\text{s}^\eps$ back to $D_\text{h}^\eps$ \cite[Proposition13 \& Theorem 14]{TH13}:
\begin{align}
    D_\textnormal{s}^{\varepsilon - \delta}\left( \mathscr{P}_{\sigma}[\rho] \,\|\,\sigma \right) &\leq D_\textnormal{h}^{\varepsilon + \delta}(\rho \,\|\,\sigma) + \log |\textnormal{\texttt{spec}}(\sigma)| - 2 \log \delta.
\end{align}
%where we have used Lemma~\ref{lemm:relation} in the last inequality.
Since $\xi>0$ is arbitrary, we take $\xi\to 0$ and optimize over all $\sigma_B \in \mathcal{S}(\mathsf{B})$ to arrive at our claim in the theorem.

In the following, we prove the claim in \eqref{eq:direct}.
Define:	
\begin{align}
	\Pi &:= \left\{\mathscr{P}_{\tau_A\otimes\sigma_B}[\rho_{AB}]\leq c\tau_A\otimes\sigma_B\right\}\pl, \pl\Pi^\mathrm{c} &:= \mathds{1} - \Pi.
\end{align}
For the rest proof, we fix $\gamma = 1$ for the noncommutative $L_p$ norm \eqref{eq:norm} and the noncommutative quotient \eqref{eq:quotient}.
Using linearity of the noncommutative quotient and
the map $\Theta$ defined in \eqref{eq:Theta}, we apply the triangle inequality of the noncommutative weighted $L_1$-norm to calculate the formulation given in \eqref{eq:formulation}:
\begin{align}
	\frac12\left\| \Theta\left( \frac{\rho_{AB}}{\tau_A\otimes \sigma_B} \right) \right\|_{1,\gamma,\tau_A^{\otimes M} \otimes \sigma_B}
	&= \frac12\left\| \Theta\left( \frac{\Pi^\mathrm{c}\rho_{AB}\Pi^\mathrm{c}}{\tau_A\otimes \sigma_B} \right) +
	\Theta\left( \frac{\Pi \rho_{AB}\Pi^\mathrm{c}}{\tau_A\otimes \sigma_B} \right) + \Theta\left(  \frac{\rho_{AB}\Pi}{\tau_A\otimes \sigma_B} \right)  \right\|_{1,\gamma,\tau_A^{\otimes M} \otimes \sigma_B} \\
	\begin{split} \label{eq:12}
	&\leq \frac12\left\| \Theta\left( \frac{\Pi^\mathrm{c}\rho_{AB}\Pi^\mathrm{c}}{\tau_A\otimes \sigma_B} \right) \right\|_{1,\gamma,\tau_A^{\otimes M} \otimes \sigma_B} \\
	&\quad + \frac12\left\| \Theta\left( \frac{ \Pi \rho_{AB}\Pi^\mathrm{c}}{\tau_A\otimes \sigma_B} \right) \right\|_{1,\gamma,\tau_A^{\otimes M} \otimes \sigma_B}  + \frac12\left\| \Theta\left(  \frac{\rho_{AB}\Pi}{\tau_A\otimes \sigma_B} \right)  \right\|_{1,\gamma,\tau_A^{\otimes M} \otimes \sigma_B}.
	\end{split}
\end{align}
We bound the three terms in \eqref{eq:12} respectively.
To bound the first term of \eqref{eq:12}, we apply the operator norm of convex splitting proved in Lemma~\ref{lemm:key} with $p=1$ to obtain
\begin{align}
	\frac12\left\| \Theta\left( \frac{\Pi^\mathrm{c}\rho_{AB}\Pi^\mathrm{c}}{\tau_A\otimes \sigma_B} \right) \right\|_{1,\gamma,\tau_A^{\otimes M} \otimes \sigma_B}
	&\leq \frac12\left\| \Theta: L_{1, \gamma, \tau_A\otimes \sigma_B} \to L_{1, \gamma, \tau_A\otimes \sigma_B^{\otimes M}} \right\| \cdot \left\|  \frac{\Pi^\mathrm{c}\rho_{AB}\Pi^\mathrm{c}}{\tau_A\otimes \sigma_B} \right\|_{1,\gamma,\tau_A\otimes \sigma_B} \\
	&\overset{\textnormal{(a)}}{\leq} \left\| \frac{\Pi^\mathrm{c}\rho_{AB}\Pi^\mathrm{c}}{\tau_A\otimes \sigma_B} \right\|_{1,\gamma,\tau_A\otimes \sigma_B} \\
%	&\overset{\textnormal{(b)}}{=} \left\| \rho_{AB} (\tau_A\otimes \sigma_B)^{-1} \Pi^\mathrm{c} \tau_A\otimes \sigma_B \right\|_1 \\
	&\overset{\textnormal{(b)}}{=} \left\| \Pi^\mathrm{c} \rho_{AB} \Pi^\mathrm{c} \right\|_1 \\
%	&\overset{\textnormal{(c)}}{\leq} \left\| \rho_{AB}^{\sfrac12} \right\|_2 \left\| \rho_{AB}^{\sfrac12} \Pi^\mathrm{c} \right\|_2 \\	
%	&= 1\cdot \Tr\left[ \Pi^\mathrm{c} \rho_{AB} \Pi^\mathrm{c} \right] \\
	&= \Tr\left[ \rho_{AB} \left\{\mathscr{P}_{\tau_A\otimes \sigma_B}[\rho_{AB}] > c\tau_A\otimes\sigma_B\right\} \right], \label{eq:12_1}
\end{align}
where (a) follows from Lemma~\ref{lemm:key} with $p=1$;
(b) follows from the definitions of noncommutative norm and the noncommutative quotient given in \eqref{eq:norm} and \eqref{eq:quotient} with $\gamma = 1$.
%(b) follow from the fact that $\Pi^\mathrm{c}$ commutes with $\tau_A\otimes \sigma_B$;
%and (c) follows from H\"older's inequality of the Schatten norm.

Next, we bound the second term of \eqref{eq:12}.
%{\color{red}
By writing $X_{AB} = \Pi\rho_{AB}\Pi^\mathrm{c}$,
note that
	\begin{align}
		\left\| \Theta\left( \frac{\Pi\rho_{AB}\Pi^\mathrm{c}}{\tau_A\otimes \sigma_B} \right) \right\|_{1,\gamma,\tau_A^{\otimes M} \otimes \sigma_B}
	&= 	\left\| \frac1M\sum_{m\in[M]} X_{A_m B} \otimes \tau_{A}^{ \otimes [M]\backslash\{m\} } - \tau_A^{\otimes [M]} \otimes \Tr_{A_m}\left[ X_{A_m B} \right]   	\right\|_{1}  \\
	&= 	\left\| \frac1M\sum_{m\in[M]} X_{A_m B}^\dagger \otimes \tau_{A}^{\otimes [M]\backslash\{m\} } - \tau_A^{\otimes [M]} \otimes X_{A}^{\dagger} \right\|_{1}  \\
	&= 	\left\| \Theta\left( \frac{ \Pi^\mathrm{c}\rho_{AB}\Pi }{\tau_A\otimes \sigma_B} \right) \right\|_{1,\gamma,\tau_A^{\otimes M} \otimes \sigma_B}
	\end{align}	
because the Schatten $1$-norm is invariant under complex conjugate and transpose `$\dagger$'.
%(Hao-Chung: Does `$\dagger$' commute with partial trace?)
%}

Then, using the monotone increase of $p\mapsto \|\,\cdot\,\|_{p, \gamma,\tau_A^{\otimes M} \otimes \sigma_B}$ (Fact~\ref{fact:norm}), we have
\begin{align}
	\left\| \Theta\left( \frac{ \Pi^\mathrm{c}\rho_{AB}\Pi}{\tau_A\otimes \sigma_B} \right) \right\|_{1,\gamma,\tau_A^{\otimes M} \otimes \sigma_B}
	&\leq\left\| \Theta\left( \frac{ \Pi^\mathrm{c}\rho_{AB}\Pi}{\tau_A\otimes \sigma_B} \right) \right\|_{2,\gamma,\tau_A^{\otimes M} \otimes \sigma_B} \\
	&\leq \left\| \Theta: L_{2, \gamma, \tau_A\otimes \sigma_B} \to L_{2, \gamma, \tau_A^{\otimes M} \otimes \sigma_B} \right\| \cdot \left\|  \frac{ \Pi^\mathrm{c}\rho_{AB}\Pi}{\tau_A\otimes \sigma_B} \right\|_{2,\gamma,\tau_A\otimes \sigma_B} \\
	&\overset{\textnormal{(a)}}{\leq} \frac{1}{\sqrt{M}} \left\|  \frac{ \Pi^\mathrm{c}\rho_{AB}\Pi}{\tau_A\otimes \sigma_B} \right\|_{2,\gamma,\tau_A\otimes \sigma_B} \\	
	&\overset{\textnormal{(b)}}{=} \frac{1}{\sqrt{M}} \sqrt{ \Tr\left[ \left| \Pi^\mathrm{c} \rho_{AB} \Pi\left( \tau_{A}\otimes \sigma_B \right)^{-\sfrac12} \right|^2\right] }\\
	&{=} \frac{1}{\sqrt{M}}\sqrt{\Tr\left[\Pi \rho_{AB} \Pi^\mathrm{c} \rho_{AB} \Pi (\tau_A\otimes \sigma_B)^{-1} \right]},  \\
	&\overset{\textnormal{(c)}}{\leq} \frac{1}{\sqrt{M}}\sqrt{\Tr\left[\Pi \rho_{AB}^2 \Pi (\tau_A\otimes \sigma_B)^{-1} \right]},
	\label{eq:12_2}
\end{align}
where (a) follows from Lemma~\ref{lemm:key} with $p=2$;
and (b) follows from the definitions given in \eqref{eq:norm} and \eqref{eq:quotient} with $\gamma = 1$.
and (c) follows from the operator inequality $\Pi^\mathrm{c} \leq \mathds{1}$.

Following the same reasoning, one can upper bound the third term of \eqref{eq:12} by the right-hand side of \eqref{eq:12_2} as well.
It remains to calculate the term $\Tr\left[ \rho_{AB}^2 \Pi (\tau_A\otimes \sigma_B)^{-1} \Pi \right]$.
Since
\begin{align}
	\Pi = \left\{ \mathscr{P}_{\tau_A\otimes\sigma_B}[\rho_{AB}] \leq c \tau_A\otimes \sigma_B \right\} = \left\{ \left( \tau_A \otimes\sigma_B\right) ^{-1} \leq c \left(\mathscr{P}_{\tau_A\otimes\sigma_B}[\rho_{AB}]\right)^{-1} \right\},
\end{align}
we obtain
\begin{align}
	\Pi (\tau_A\otimes \sigma_B)^{-1} \Pi &\leq c \Pi \left(\mathscr{P}_{\tau_A\otimes\sigma_B}[\rho_{AB}]\right)^{-1} \Pi \\
	&= c\left(\mathscr{P}_{\tau_A\otimes\sigma_B}[\rho_{AB}]\right)^{-1},
\end{align}
where we used the fact that $\Pi$ commutes with $\mathscr{P}_{\tau_A\otimes\sigma_B}[\rho_{AB}]$. Then,
\begin{align}
	\Tr\left[\rho_{AB}^2 \Pi (\tau_A\otimes \sigma_B)^{-1} \Pi\right] &\le c\Tr\left[\rho_{AB}^2 \left(\mathscr{P}_{\tau_A\otimes\sigma_B}[\rho_{AB}]\right)^{-1}\right] \\
	&\overset{\textnormal{(a)}}{\leq} c|\textnormal{\texttt{spec}}(\tau_A\otimes\sigma_B)|\Tr\left[ \rho_{AB} \rho_{AB}^{-1} \right], \\
%	&=c|\textnormal{\texttt{spec}}(\tau_A)||\textnormal{\texttt{spec}}(\sigma_B)|
	&\leq c \nu \label{eq:12_0},
\end{align}
where in (a) we used the the pinching inequality (Fact~\ref{fact:pinching}), i.e.~
\[\rho_{AB}\le |\textnormal{\texttt{spec}}(\tau_A \otimes \sigma_B )| \mathscr{P}_{\tau_A\otimes \sigma_B}[\rho_{AB}],\] and the operator monotonicity of inversion. Then, combining \eqref{eq:12}, \eqref{eq:12_1}, \eqref{eq:12_2}, and \eqref{eq:12_0} leads to our claim in \eqref{eq:direct}, which completes the prove.
\end{proof}

\subsection{Converse Bound} \label{sec:converse_bound}
\begin{theo}[Converse Bound] \label{theo:converse}
Let $\rho_{AB}$ and $\tau_A$ be density operators, and let $\omega_{A_1\ldots A_M B}$ be defined in \eqref{eq:total_state}.
Then, for any $\varepsilon\in(0,1)$, $\delta < 1-\eps$,
and $c>0$, we have
\begin{align}
	\log M_\varepsilon^{\textnormal{\texttt{c}}}\left(\rho_{AB} \,\Vert\, \tau_A\right)  \geq D_\textnormal{h}^{1-\varepsilon - \delta }\left(\rho_{AB}\,\|\,\tau_A\otimes \rho_B \right)  - \log \frac{ 2 + c + c^{-1} }{ c\eps - c + (1+c)\delta }.
	\end{align}
	Here, the $\eps$-hypothesis-testing divergence $D_\textnormal{h}^{\eps}$ is defined in \eqref{eq:D_h}.
\end{theo}

\begin{proof}
	We apply Lemma~\ref{lemm:converse} to obtain
	\begin{align}
		&\frac12\left\| \omega_{A_1\ldots A_M B} - \tau_A^{\otimes M} \otimes \rho_{B} \right\|_1\\
		&\geq 1 - (1+c) \Tr\left[ \rho_{AB} \wedge (1+c^{-1}) M \tau_A \otimes \rho_B \right] \\
		&\overset{\text{(a)}}{=} 1 - (1+c) \inf_{0\leq T_{AB} \leq \mathds{1}_{AB}} \tr\left[ \rho_{AB}(\mathds{1}_{AB}-T_{AB}) +  (1+c^{-1}) M \tau_A \otimes \rho_B \cdot T_{AB} \right]\\
		&\overset{\text{(b)}}{\geq} 1 - (1+c) \left[ (1-\eps - \delta) + (1+c^{-1}) M \mathrm{e}^{- D_\text{h}^{1-\eps - \delta }\left(\rho_{AB}\,\|\,\tau_A\otimes \rho_B \right) }  \right],
	\end{align}
	where (a) follows from the infimum representation of the noncommutative minimal (Fact~\ref{fact:minimal}-\ref{item:infimum});
	and in (b) we choose a test $T_{AB}$ to attain the definition of the hypothesis-testing divergence $D_\text{h}^{1-\eps - \delta }\left(\rho_{AB}\,\|\,\tau_A\otimes \rho_B \right)$ introduced in \eqref{eq:D_h}. By choosing the largest $M$ such that $	\frac12\left\| \omega_{A_1\ldots A_M B} - \tau_A^{\otimes M} \otimes \rho_{B} \right\|_1 \leq \eps $, we complete the proof.
\end{proof}

\section{Exponents (Bounds on Error)} \label{sec:exponet}
We provide a one-shot error exponent bound (i.e.~an upper bound) on $\Delta^\texttt{c}_M(A:B)_\rho$ in Section~\ref{sec:error_exponent}, and a one-shot strong converse bound (i.e.~a lower bound) in Section~\ref{sec:strong_converse}.

\subsection{One-Shot Error Exponent} \label{sec:error_exponent}
\begin{theo} \label{theo:exponent}
	For any density operators $\rho_{AB}$, $\tau_A$, and $\omega_{A_1\ldots A_M B}$ defined in \eqref{eq:total_state},
	we have, for every $M \in \mathds{N}$,
	\begin{align} \label{eq:exponent}
		\frac12\left\| \omega_{A_1\ldots A_M B} - \tau_A^{\otimes M} \otimes \rho_{B} \right\|_1 \leq 2^{\frac{2}{\alpha} - 2} \mathrm{e}^{ - \frac{\alpha-1}{\alpha} \left( \log M -  I_\alpha^*\left(\rho_{AB}\,\|\,\tau_A\right) \right)}, \quad \forall \alpha \in [1,2].
	\end{align}
	Here, the generalized sandwiched R\'enyi information \cite{HT14} and the order-$\alpha$ sandwiched R\'enyi divergence \cite{MDS+13, WWY14} are defined as:
	\begin{align}
	I_\alpha^*\left(\rho_{AB}\,\|\,\tau_A\right) &:= \inf_{ \sigma_B \in \mathcal{S}(\mathsf{B}) } D_\alpha^*\left(\rho_{AB}\,\|\,\tau_A \otimes \sigma_B \right); \label{eq:I_tau}\\
	D_\alpha^*(\rho\,\Vert\,\sigma) &:= \frac{1}{\alpha-1} \log \Tr\left[ \left( \sigma^{\frac{1-\alpha}{2\alpha}} \rho \sigma^{\frac{1-\alpha}{2\alpha}} \right)^\alpha \right].
	\end{align}
Moreover, the exponent $\sup_{\alpha\in [\sfrac12,1] } \frac{\alpha-1}{\alpha} \left( \log M - I_\alpha^*\left(\rho_{AB}\,\|\,\tau_A\right)  \right)$ is positive if and only if
	\begin{align}
	\log M > I\left(\rho_{AB}\,\|\,\tau_A\right)
	= D(\rho_{AB} \,\Vert\, \tau_A \otimes \rho_B).
	\end{align}
\end{theo}

\begin{proof}
	Throughout this proof, we fix $\gamma = \sfrac12$ as the parameter in the weighted norm defined in \eqref{eq:norm}.
	Using the formulation given in \eqref{eq:formulation} and recalling that the noncommutative weighted $L_p$-norm is monotonically increasing in $p\geq 1$ (Fact~\ref{fact:norm}), we have, for every $\alpha \in [1,2]$ and any density operator $\sigma_B \in \mathcal{S}(\mathsf{B})$,
	\begin{align}
	\left\| \Theta\left( \frac{\rho_{AB}}{\tau_A\otimes \sigma_B} \right) \right\|_{1, \gamma, \tau_A^{\otimes M} \otimes \sigma_B}
	&\leq \left\| \Theta\left( \frac{\rho_{AB}}{\tau_A\otimes \sigma_B} \right) \right\|_{\alpha, \gamma ,\tau_A^{\otimes M} \otimes \sigma_B} \\
	&\leq \left\| \Theta: L_{\alpha,\gamma}(\tau_A\otimes \sigma_B) \to L_{\alpha,\gamma} (\tau_A^{\otimes M} \otimes \sigma_B) \right\| \cdot
	\left\|\frac{\rho_{AB}}{\tau_A\otimes \sigma_B}\right\|_{\alpha, \gamma, \tau_A\otimes \sigma_B}	 \\
	&\leq 2^{\frac{2}{\alpha} - 1} M^{\frac{1-\alpha}{\alpha}} \cdot
	\left\|\frac{\rho_{AB}}{\tau_A\otimes \sigma_B}\right\|_{\alpha, \gamma, \tau_A\otimes \sigma_B},
\end{align}
	where in the last inequality we have invoked the map norm of $\Theta$ given in Lemma~\ref{lemm:key}. On the other hand, by definition of the noncommutative $L_{\al,\frac{1}{2}}$ norm \eqref{eq:norm}, we have
\begin{align}
	\left\|\frac{\rho_{AB}}{\tau_A\otimes \sigma_B}\right\|_{\alpha, \sfrac12, \tau_A\otimes \sigma_B}
	&= \left\| (\tau_A\otimes \sigma_B)^{\frac{1-\alpha}{2\alpha}} \rho_{AB} (\tau_A\otimes \sigma_B)^{\frac{1-\alpha}{2\alpha}} \right\|_\alpha \\
	&= \left( \Tr\left[ \left( (\tau_A\otimes \sigma_B)^{\frac{1-\alpha}{2\alpha}} \rho_{AB} (\tau_A\otimes \sigma_B)^{\frac{1-\alpha}{2\alpha}} \right)^\alpha \right] \right)^{\sfrac{1}{\alpha} } \\
	&= \mathrm{e}^{ \frac{\alpha-1}{\alpha} D_\alpha^*\left(\rho_{AB}{\,\|\,}\tau_A\otimes \sigma_B\right) }.
\end{align}
Since this holds for every density operators $\sigma_B \in \mathcal{S}(\mathsf{B})$, optimizing it completes the proof.
\end{proof}

\subsection{One-Shot Exponential Strong Converse} \label{sec:strong_converse}
\begin{theo} \label{theo:sc}
	For any density operators $\rho_{AB}$, $\tau_A$, and $\omega_{A_1\ldots A_M B}$ defined in \eqref{eq:total_state},
	we have, for every $M \in \mathds{N}$,
	\begin{align} \label{eq:sc}
		\frac12\left\| \omega_{A_1\ldots A_M B} - \tau_A^{\otimes M} \otimes \rho_{B} \right\|_1 \geq 1 - 4\,\mathrm{e}^{ - \frac{1-\alpha}{\alpha} \left( D_{2 - \sfrac{1}{\alpha}}\left(\rho_{AB}\,\|\,\tau_A\otimes \rho_B\right) - \log M \right)}, \quad \forall \alpha \in [\sfrac12,1],
	\end{align}
	where the Petz--R\'enyi divergence \cite{Pet86} is defined as
	\begin{align}
		D_\alpha\left( \rho\,\Vert\, \sigma\right) &= \frac{1}{\alpha-1}\log \Tr\left[ \rho^\alpha \sigma^{1-\alpha} \right]. \label{eq:Petz}
	\end{align}
	Moreover, the exponent $\sup_{\alpha\in [\sfrac12,1]} \frac{1-\alpha}{\alpha} \left( I_{2 - \sfrac{1}{\alpha}}\left(\rho_{AB}\,\|\,\tau_A\otimes \rho_B\right) - \log M \right)$ is positive if and only if $\log M < D\left(\rho_{AB}\,\|\, \tau_A\otimes \rho_B \right)$.
\end{theo}

\begin{proof}
	We apply Lemma~\ref{lemm:converse} with $c=1$ to have
	\begin{align}
	\frac12\left\| \omega_{A_1\ldots A_M B} - \tau_A^{\otimes M} \otimes \rho_{B} \right\|_1 &\geq 1 - 2 \Tr\left[ \rho_{AB} \wedge \left(2 M \tau_A \otimes \rho_B\right) \right] \\
	&\overset{\textnormal{(a)}}{\geq}
	1 - 2\cdot (2M)^s \Tr\left[ \left(\rho_{AB}\right)^{1-s} \left(\tau_A\otimes \rho_B \right)^s \right], \quad \forall\, s\in (0,1) \\
	&\overset{\textnormal{(b)}}{=} 1 - 2^{\frac{1}{\alpha} } \mathrm{e}^{ \frac{\alpha-1}{\alpha} \left( D_{2 - \sfrac{1}{\alpha}}\left(\rho_{AB}\,\|\,\tau_A\otimes \rho_B \right) - \log M \right)}, \quad \forall \alpha \in (\sfrac12,1),
\end{align}
where in (a) we invoked the upper bound to the noncommutative minimal (Fact~\ref{fact:minimal}-\ref{item:upper}); and in (b) we used the substitution $\alpha = \frac{1}{1+s} \in (\sfrac12,1)$.
\end{proof}

\section{Applications in Quantum Information Theory} \label{sec:app}
The goal of this section is to express certain operational quantities as ``covering" and ``packing" trace distances defined below. Then, the established one-shot characterizations of convex splitting in Sections~\ref{sec:sample} and \ref{sec:exponet} apply to:
\begin{enumerate}
\item[(i)] private communication over quantum wiretap channel (Section~\ref{sec:private});
\item[(ii)] secret key distillation (Section~\ref{sec:secrete-key});
\item[(iii)] one-way quantum message compression (Section~\ref{sec:message});
\item[(iv)] quantum measurement simulation (Section~\ref{sec:measurement});
\item[(v)] communication with channel state information available at encoder (Section~\ref{sec:state}).
\end{enumerate}
For density operator $\rho_{AB}$ and $\tau_{A}$, we recall the ``\emph{covering error}'' in \eqref{eq:covering} : for any $M\in\mathds{N}$,
\begin{align}
\begin{split}
	\Delta^\texttt{c}_M\left( \rho_{AB} \, \Vert\, \tau_A\right)&:=\frac12\left\| \frac{1}{M}\sum_{m=1}^M \rho_{A_m B} \otimes \tau_{A_1}\otimes \cdots \otimes \tau_{A_{m-1}} \otimes \tau_{A_{m+1}} \otimes \cdots \otimes \tau_{A_M}  - \tau_{A}^{\otimes M} \otimes \rho_B \right\|_1; \\
\Delta^\texttt{c}_M\left(A:B\right)_{\rho} &:=
	\Delta^\texttt{c}_M\left( \rho_{AB} \, \Vert\, \rho_A\right)
\end{split}
\end{align}
We also introduce a ``\emph{packing error}'': for any density operator $\rho_{AB}$ and $M\in\mathds{N}$,
\begin{align}
\begin{split}\label{eq:packing}
\Delta^\texttt{p}_M\left(A:B\right)_{\rho} &:= \inf_{ \mathcal{M}_{A_1\ldots A_M B\to \hat{M}}} \frac{1}{M}\sum_{m=1}^M \frac12\left\|\mathcal{M}_{A_1\ldots A_M B\to\hat{M}} \left( \rho_{A_1\ldots A_M B}^{(m)}  \right) - |m\rangle\langle m|_{\hat{M}} \right\|_1 \\
&= \inf_{ \text{POVM } \{ \Pi_{ A_1\ldots A_M B }^{m} \}_{m\in[M]} }  \frac{1}{M}\sum_{m=1}^M\Tr\left[ \rho_{A_1\ldots A_M B}^{(m)}  \cdot \left( \mathds{1} - \Pi_{ A_1\ldots A_M B }^{m} \right) \right].
\end{split}\\
\rho_{A_1\ldots A_M B}^{(m)} &:= \rho_{A_m B} \otimes \rho_{A_1}\otimes \cdots \otimes \rho_{A_{m-1}} \otimes \rho_{A_{m+1}} \otimes \cdots \otimes \rho_{A_M}, \quad \forall\, m\in[M],
\end{align}
where the minimization is over all positive operator-valued measures (POVM).
%In other words, it is the \emph{minimum error probability using random codebook} (with codewords drawn from distribution $p_X$).
Here, the superscript `$\texttt{c}$' designates for the ``covering error", while
superscript `$\texttt{p}$' designates for the ``packing error". We recall the following lemma about one-shot packing
\begin{lemm}[One-shot packing lemma {\cite{Cheng2022}}] \label{lemm:packing}
	For any quantum state $\rho_{AB}$, the following hold.
	Consider the quantity $\Delta^\textnormal{\texttt{p}}_M\left(A:B\right)_{\rho}$ defined in \eqref{eq:packing}.
	\begin{enumerate}[(a)]
		\item\label{lemma:packing-(a)} ($\varepsilon$-one-shot capacity)
		For every $\varepsilon\in(0,1)$ and $\Delta^\textnormal{\texttt{p}}_M\left(A:B\right)_{\rho} \leq \varepsilon$, one has
		\begin{align}\label{eq:packing}
			M \geq I_\textnormal{h}^{\varepsilon - \delta} (X:B)_\rho - \log \frac{1}{\delta}, \quad \forall \delta > 0.
		\end{align}
		
		\item\label{lemma:packing-(b)} (Achievable error exponent)
		For every $M \in \mathds{N}$,
		\begin{align}
			\Delta^\textnormal{\texttt{p}}_M\left(A:B\right)_{\rho} \leq \mathrm{e}^{ -\sup_{\alpha\in [ \sfrac12,1]} \frac{1-\alpha}{\alpha} \left( I^\uparrow_{2 - \sfrac{1}{\alpha}}\left(A:B\right)_\rho - \log M \right) }.
%			^{- E_{\log M}^{\textnormal{\texttt{p}}}(X:B)_{\rho}}.
		\end{align}
		Here, $I^\uparrow_\alpha(A:B)_\rho := D_\alpha(\rho_{AB}\|\rho_A\otimes \rho_B)$, and
		$D_\alpha\left( \rho\,\Vert\, \sigma\right) = \frac{1}{\alpha-1}\log \Tr\left[ \rho^\alpha \sigma^{1-\alpha} \right]$ (see \cite{Pet86}).
	\end{enumerate}
\end{lemm}

We define
$I(A: B)_{\rho} := D( \rho_{AB}\,\Vert\, \rho_A \otimes \rho_{B})$ as the quantum mutual information and by $V( \rho_{AB}\,\Vert\, \rho_A \otimes \rho_{B})$ the quantum information variance, where $V(\rho\,\Vert\,\rho) := \Tr\left[ \rho(\log \rho - \log \sigma ) - D(\rho\,\Vert\,\sigma)^2 \right]$.
We also denote $ \Phi^{-1}({\eps}) := \sup \{ \gamma\in \mathds{R} : \frac{1}{\sqrt{2\pi}} \int_{-\infty}^{\gamma} \mathrm{e}^{-u^2/2} \, \mathrm{d} u \leq \eps  \}$ as the inverse function of distribution function of normal distribution.

%We also introduce the exponents for ``covering" and ``packing:, respectively, as
%\begin{align}
%E_{\log K}^{\texttt{c}}(X:B)_{\rho} &:= \sup\limits_{\alpha\in(1,2)} \frac{1-\alpha}{\alpha} \left( I^*_{\alpha}\left(X:B\right)_\rho - \log K \right); \\
%E_{\log M}^{\texttt{p}}(X:B)_{\rho}
%&:=\sup\limits_{\alpha\in(\sfrac12,1)} \frac{1-\alpha}{\alpha} \left( I^\uparrow_{2 - \sfrac{1}{\alpha}}\left(X:B\right)_\rho - \log M \right);\\
%\end{align}

\subsection{Private Communication over Quantum Wiretap Channels} \label{sec:private}

\hfill \break
The private communication protocol \cite{CWY04, Dev05, Hay15, WTB17, Wil17b, KKG+19} is described as follows.
\begin{defn}[Private communication] \label{defn:private}
%	Let $\rho_{XBE} = \sum_{x\in\mathcal{X}} p_X(x) |x\rangle \langle x |\otimes \rho_{BE}^x$ be a classical-quantum-quantum state.
	Let $\mathcal{N}_{X\to BE}: x\to \rho_{BE}^x$ be a classical-quantum wiretap channel.
	\begin{enumerate}[1.]
		\item Alice holds a classical register $X$ at Alice; Bob and Eve hold quantum registers $B$ and $E$, respectively.
		
%		\item A resource of uniform distribution on $[K]$ as a private randomness at Alice.
		
%		\item A random encoding $(m,k) \mapsto x_{m,k}$ according to $p_X$ and $k$ is drawn equally probably.
		
		\item Alice performs an encoding $m\mapsto x_{m} \in \mathcal{X}$ to encode each message $m\in[M]$ and send the codeword $x$ via the channel $\mathcal{N}_{X\to BE}: x\mapsto \rho_{BE}^x$.
		
%		\item The channel
		
		\item A decoding measurement $\mathcal{M}_{B\to\hat{M}}$ on system $B$ at Bob to decode the message.
	\end{enumerate}
%	A  $(\log M, \log K, \eps)$-one-shot private communication protocol for $x\mapsto \rho_{BE}^x$ using random coding (with codewords drawn from distribution $p_X$) satisfies
%	\begin{align}
%	\frac{1}{M} \sum_{m\in[M]}	\mathds{E}_{x_{m,k}\sim p_X}\frac12 \left\| \mathcal{M}_{B\to\hat{M}} \left(\frac{1}{K}\sum_{k\in[K]} \rho_{BE}^{x_{m,k}}\right) - |m\rangle\langle m|_{\hat{M}} \otimes \rho_E \right\|_1 \leq \varepsilon.
%	\end{align}
%	Moreover,
We call it a $(\log M, \eps)$ private communication protocol for $\mathcal{N}_{X\to BE}$ if the \emph{privacy error} \cite{HHH+09, WTB17, Wil17b} satisfies
	\begin{align} \label{eq:single_error}
	\frac1M \sum_{m=1}^M \frac12 \left\| \mathcal{M}_{B\to\hat{M}} \left(\rho_{BE}^m\right) - |m\rangle\langle m|_{\hat{M}} \otimes \rho_E \right\|_1 \leq \eps.
	\end{align}
\end{defn}

In studying private communication over quantum wiretap channels, one often consider two separate error criteria \cite{Hay15}, namely, $\eps_1$ for characterizing the error probability of decoding at Bob, and $\eps_2$ for describing the security issue.
In Definition~\ref{defn:private}, we consider the single error criterion, which follows from the approach taken in \cite{HHH+09, WTB17, Wil17b}.
As pointed out in \cite{WTB17, Wil17b}, the two separate criteria are both fulfilled once the single error criterion is satisfied by monotonicity of trace norm under partial trace.
Precisely, suppose \eqref{eq:single_error} holds by some $\eps$, then tracing out the $E$ system gives
\begin{align}
	\frac1M \sum_{m=1}^M \frac12 \left\| \mathcal{M}_{B\to\hat{M}} \left(\rho_{B}^m\right) - |m\rangle\langle m|_{\hat{M}} \right\|_1 = \Pr\left\{\hat{M}\neq M \right\}\leq \eps,
\end{align}
which guarantees the error probability $\eps_1$ of decoding at Bob being no larger than $\eps$.
On the other hand, tracing out the $\hat{M}$ system gives
\begin{align}
	\frac1M \sum_{m=1}^M \frac12 \left\| \rho_{E}^m -  \rho_E \right\|_1 \leq \eps,
\end{align}
which guarantees the secrecy error $\eps_2$ being at most $\eps$.

%\begin{align}
%	\Delta_\textnormal{\texttt{private}}(M,K)_\rho &:=
%	\mathds{E}_{x_{m,k}\sim p_X}\frac12 \left\| \mathcal{M}_{B\to\hat{M}} \left(\frac{1}{K}\sum_{k\in[K]} \rho_{BE}^{x_{m,k}}\right) - |m\rangle\langle m|_{\hat{M}} \otimes \rho_E \right\|_1 \\
%	\rho_{XBE} &:= \sum_{x\in\mathcal{X}} p_X(x) |x\rangle \langle x |\otimes \rho_{BE}^x. \label{eq:rho_XBE}
%\end{align}

%\begin{lemm}[One-shot achievability for private communication] \label{lemm:private}
%%	For every $\eps_1, \eps_2>0$ with $\eps_1+\eps_2<1$, there exists  $(\log M, \log K, \eps)$-one-shot private communication protocol satisfying
%	For each $\rho_{XBE} = \sum_{x\in\mathcal{X}} p_X(x) |x\rangle \langle x |\otimes \rho_{BE}^x$ and $(M,K)\in[M]\times [K]$, there exists a measurement $\mathcal{M}_{B\to\hat{M}}$ satisfying
%	\begin{align}
%	\frac{1}{M} \sum_{m\in[M]}	\mathds{E}_{x_{m,k}\sim p_X}\frac12 \left\| \mathcal{M}_{B\to\hat{M}} \left(\frac{1}{K}\sum_{k\in[K]} \rho_{BE}^{x_{m,k}}\right) - |m\rangle\langle m|_{\hat{M}} \otimes \rho_E \right\|_1
%	\leq
%	\Delta^{\textnormal{\texttt{p}}}_{MK}\left(X:B\right)_{\rho} + \Delta^{\textnormal{\texttt{c}}}_{K}\left(X:E\right)_{\rho}.
%	\end{align}
%\end{lemm}

\begin{theo}[One-shot error exponent for private communication] \label{theo:private_exponent}
	Let $\mathcal{N}_{X\to BE}$ be a classical-quantum wiretap channel.
	There exists a $\left(\log M, \eps \right)$ private communication protocol for $\mathcal{N}_{X\to BE}$ such that for any $K\in\mathds{N}$ and probability distribution $p_X$,
%	using $\log|K|$-nat private randomness such that
	\begin{align}\label{eq:oneshot}
		\eps &\leq \Delta^{\textnormal{\texttt{p}}}_{MK}\left(X:B\right)_{\rho} + \Delta^{\textnormal{\texttt{c}}}_{K}\left(X:E\right)_{\rho} \\
		&\leq \mathrm{e}^{- \sup_{\alpha\in [ \sfrac12,1]} \frac{1-\alpha}{\alpha} \left( I^\uparrow_{2 - \sfrac{1}{\alpha}}\left(X:B\right)_\rho - \log(MK) \right)  } + \mathrm{e}^{- \sup_{\alpha\in[1,2]} \frac{\alpha-1}{\alpha} \left( \log K -  I^*_{\alpha}\left(X:E\right)_\rho  \right) }.
	\end{align}
	Here, $\rho_{XBE} = \sum_{x\in\mathcal{X}} p_X(x)|x\rangle \langle x|\otimes \rho_{BE}^x$.
	
	Moreover, the overall error exponent is positive if and only if there exists $K$ such that $\log (MK) < I(X:B)_\rho$ and $\log K > I(X:E)_\rho$, namely,
    \begin{align}
    \log M < I(X:B)_\rho - I(X:E)_\rho.
    \end{align}
	%\begin{align}
	%	\Delta_\textnormal{\texttt{private}}(M,K)_\rho
	%	&\leq \mathrm{e}^{- E_\textnormal{\texttt{public}}(MK)_\rho }
	%	+ \mathrm{e}^{- E_\textnormal{\texttt{covering}}(K)_\rho}, \\
	%	E_\textnormal{\texttt{public}}(MK)_\rho
	%	&:=\sup\limits_{\alpha\in(\sfrac12,1)} \frac{1-\alpha}{\alpha} \left( I^\uparrow_{2 - \sfrac{1}{\alpha}}\left(X:B\right)_\rho - \log MK \right);\\
	%	E_\textnormal{\texttt{covering}}(K)_\rho &:= \sup\limits_{\alpha\in(1,2)} \frac{1-\alpha}{\alpha} \left( I^*_{\alpha}\left(X:E\right)_\rho - \log K \right).
	%\end{align}
\end{theo}

\begin{remark}
	Theorem~\ref{theo:private_exponent} demonstrates that the privacy error (using random coding) decomposes to the ``\emph{packing error}" $\Delta^{\textnormal{\texttt{p}}}_{MK}\left(X:B\right)_{\rho}$ that describes the \emph{error probability of decoding} for communication from Alice to Bob
	and the ``\emph{covering error}" $\Delta^{\textnormal{\texttt{c}}}_{K}\left(X:E\right)_{\rho}$ that describes the \emph{secrecy error} (i.e.~the information leakage) from Alice to Eve.
\end{remark}

\begin{remark}
	Theorem~\ref{theo:private_exponent} for bounding the secrecy part (i.e.~$\Delta^{\textnormal{\texttt{c}}}_{K}\left(X:E\right)_{\rho}$) has the following advantages compared to the existing works (e.g.~\cite{Hay15}).
	Firstly, the upper bound on the term~$\Delta^{\textnormal{\texttt{c}}}_{K}\left(X:E\right)_{\rho}$ is expressed in terms of the sandwiched R\'enyi divergence $D_\alpha^*$ for $\alpha \in [1,2]$, which gives the tightest error exponent among all additive quantum R\'enyi divergences \cite{MDS+13, WWY14, Tom16}.
	Secondly, some of the previous secrecy bounds were obtained under other error criteria such as relative entropy or purified distance \cite{Hay15}.
	By translating those criteria back to trace distance (e.g.~via Pinsker's inequality) will result in a looser bound, e.g.~with one-half of the error exponent.
	Since Theorem~\ref{theo:private_exponent} directly analyzes trace distance as the secrecy criterion, our bound does not suffer from such a one-half factor.
\end{remark}

\begin{remark}
	If shared randomness is given between Alice and Bob, Theorem~\ref{theo:private_exponent} can be strengthened to hold against compound wiretap channel by following the argument \cite[Theorem 4]{KKG+19} (see also the discussion at \cite[Remark 5]{KKG+19} therein).
\end{remark}

%Then, applying Lemma~\ref{lemm:packing}-\ref{lemma:packing-(a)} and Theorem~\ref{theo:direct} to Lemma~\ref{lemm:private}, we obtain:
Applying Theorem~\ref{theo:direct} to
Theorem~\ref{theo:private_exponent} yields the following one-shot private capacity for $\mathcal{N}_{X\to BE}$ and the corresponding second-order coding rate for $\mathcal{N}_{X\to BE}^{\otimes n}$.
\begin{theo}[One-shot private capacity and second-order coding rate] \label{theo:private_rate}
	Let $\mathcal{N}_{X\to BE}$ be a classical-quantum wiretap channel.
	For every $\eps_1, \eps_2>0$ with $\eps_1+\eps_2<1$, there exists a $\left(\log M, \eps_1 + \eps_2 \right)$-one-shot private communication protocol for $\mathcal{N}_{X\to BE}:x\to\rho_{BE}^x$ such that for any probability distribution $p_X$,
	\begin{align}
		\log M\geq I_\textnormal{h}^{\eps_1-\delta_1}\left( X:B \right)_\rho - I_\textnormal{h}^{1 - \eps_2 + 3\delta_2}\left( X: E\right)_\rho - \log \tfrac{1}{\delta_1} - \log\tfrac{\nu^2}{\delta_2^4},
	\end{align}
	where $\rho_{XBE} = \sum_{x\in\mathcal{X}} p_X(x)|x\rangle \langle x|\otimes \rho_{BE}^x$ and  $\nu = |E|$.
	
	Moreover, for sufficiently large $n\in\mathds{N}$, there exists a $(\log M, \eps_1+\eps_2)$ private communication protocol for $\mathcal{N}_{X\to BE}^{\otimes n}$ such that the following holds for any probability distribution $p_X$,
	\begin{align}
		\log M \geq n\left(I(X:B)_\rho - I(X:E)_\rho \right) + \sqrt{ n V(X:B)_\rho}  \Phi^{-1}(\eps_1) + \sqrt{ n V(X:E)_\rho}  \Phi^{-1}(\eps_2) - O(\log n).
	\end{align}
%	where $\Phi^{-1}({\eps}) := \sup \{ \gamma\in \mathds{R} : \frac{1}{\sqrt{2\pi}} \int_{-\infty}^{\gamma} \mathrm{e}^{-u^2/2} \, \mathrm{d} u \}$.	
\end{theo}
\begin{remark}
	Theorem~\ref{theo:private_rate} follows immediately combining \eqref{eq:oneshot} with  \eqref{eq:size} from Theorem \ref{theo:direct} and \eqref{eq:packing} from Theorem \ref{lemm:packing}. It improves the information leakage to Eve (i.e.~$\eps_2$) from \cite{Wil17b} quadratically, and hence, gives a better second-order achievable coding rate.
\end{remark}

\begin{proof}[Proof of Theorem~\ref{theo:private_exponent}]
	We employ a random encoding $(m,k) \mapsto x_{m,k}$ according to some distribution $p_X$, where $k \in [K]$ is drawn equally probably according to some private randomness at Alice (unknown by Eve).
	For every POVM $\left\{\Pi_{B}^{x_{m,k}}\right\}_{m,k}$, we define the associated quantum measurement operation:
	\begin{align}
	\mathcal{M}_{B\to\hat{M}\hat{K}}(\omega_B) := \sum_{\hat{m}, \hat{k}} \Tr\left[ \omega_B \Pi_B^{x_{m,k}} \right] |m\rangle\langle m|_{\hat{M}} \otimes |k\rangle\langle k|_{\hat{K}},
	\end{align}
	and $\mathcal{M}_{B\to\hat{M}} := \Tr_{\hat{K}} \circ \mathcal{M}_{B\to\hat{M}\hat{K}}$. For each $m\in[M]$, triangle inequality implies that
	\begin{align}
	&\frac12 \left\| \mathcal{M}_{B\to\hat{M}} \left(\frac{1}{K}\sum_{k\in[K]} \rho_{BE}^{x_{m,k}}\right) - |m\rangle\langle m|_{\hat{M}} \otimes \rho_E \right\|_1 \notag \\
	\begin{split} \label{eq:private12}
	&\leq \frac12\left\|\mathcal{M}_{B\to\hat{M}} \left(\frac{1}{K}\sum_{k\in[K]} \rho_{BE}^{x_{m,k}}\right) - |m\rangle\langle m|_{\hat{M}} \otimes \frac{1}{K}\sum_{k\in[K]} \rho_{E}^{x_{m,k}} \right\|_1 \\
	&\quad + \frac12\left\| |m\rangle\langle m|_{\hat{M}} \otimes \frac{1}{K}\sum_{k\in[K]} \rho_{E}^{x_{m,k}} - |m\rangle\langle m|_{\hat{M}} \otimes \rho_E \right\|_1.
	\end{split}
	\end{align}
We will upper bound the two terms on the right-hand side of \eqref{eq:private12} separately.
For the first term,
	\begin{align}
	&\frac12\left\|\mathcal{M}_{B\to\hat{M}} \left(\frac{1}{K}\sum_{k\in[K]} \rho_{BE}^{x_{m,k}}\right) - |m\rangle\langle m|_{\hat{M}} \otimes \frac{1}{K}\sum_{k\in[K]} \rho_{E}^{x_{m,k}} \right\|_1 \\
	&\overset{\textnormal{(a)}}{=}  \frac12\left\|\mathcal{M}_{B\to\hat{M}} \left(\frac{1}{K}\sum_{k\in[K]} \rho_{B}^{x_{m,k}}\right) - |m\rangle\langle m|_{\hat{M}} \right\|_1 \\
	&\overset{\textnormal{(b)}}{\leq} \frac{1}{K}\sum_{k\in[K]} \frac12\left\|\mathcal{M}_{B\to\hat{M}} \left( \rho_{B}^{x_{m,k}}\right) - |m\rangle\langle m|_{\hat{M}} \right\|_1 \\
	&\overset{\textnormal{(c)}}{\leq} \frac{1}{K}\sum_{k\in[K]} \frac12\left\|\mathcal{M}_{B\to\hat{M}\hat{K}} \left( \rho_{B}^{x_{m,k}}\right) - |m\rangle\langle m|_{\hat{M}}\otimes |k\rangle\langle k|_{\hat{K}} \right\|_1 \\
	&\overset{ }{=} \frac{1}{K}\sum_{k\in[K]}\Tr\left[ \left( \mathds{1}_B - \Pi_B^{x_{m,k}} \right) \rho_B^{x_{m,k}} \right]. \label{eq:private1_1}
	\end{align}
	Here, (a) follows from Fact~\ref{fact:invariance};
	(b) follows from convexity of trace norm;
	and (c)	relies on the fact that trace distance is increasing by appending additional systems $\hat{K}$.
	Applying expectation via random coding and choosing the best measurement $\mathcal{M}_{B\to \hat{M}\hat{K}}$, we obtain
	\begin{align}
	 \frac{1}{MK}\sum_{(m,k)\in[M]\times [K]} \mathds{E}_{x_{m,k}\sim p_X} \Tr\left[ \left( \mathds{1}_B - \Pi_B^{x_{m,k}} \right) \rho_B^{x_{m,k}} \right]
	&= \Delta^{\textnormal{\texttt{p}}}_{MK}\left(X:B\right)_{\rho}. \label{eq:private1_2}
	\end{align}
%	\begin{align}
%	\mathds{E}_{x_{m,k}\sim p_X} \frac{1}{K}\sum_{k\in[K]} \Tr\left[ \left( \mathds{1}_B - \Pi_B^{x_{m,k}} \right) \rho_B^{x_{m,k}} \right] \leq \eps_1 \label{eq:private1_2}
%	\end{align}
%	as long as we choose
%	\begin{align}
%	\log MK =  I_\text{h}^{\eps_1-\delta_1}\left( X:B \right)_\rho - \log \frac{1}{\delta_1}. \label{eq:private1_0}
%	\end{align}
	Next, we bound the second term on the right-hand side of \eqref{eq:private12}:
	\begin{align}
	\mathds{E}_{x_{m,k}\sim p_X} \frac12\left\| |m\rangle\langle m|_{\hat{M}} \otimes \frac{1}{K}\sum_{k\in[K]} \rho_{E}^{x_{m,k}} - |m\rangle\langle m|_{\hat{M}} \otimes \rho_E \right\|_1
	&\overset{\textnormal{(a)}}{=} \mathds{E}_{x_{m,k}\sim p_X} \frac12\left\| \frac{1}{K}\sum_{k\in[K]} \rho_{E}^{x_{m,k}} -  \rho_E \right\|_1 \\
	&= \frac12 \left\| \rho_{X^{MK}E}^m - \rho_{X^{MK}}\otimes \rho_E \right\|_1 \\
	&\overset{\textnormal{(b)}}{=} \frac12 \left\| \rho_{X^{K}E}^m - \rho_{X^{K}}\otimes \rho_E \right\|_1 \\
	&= \Delta^{\textnormal{\texttt{c}}}_{K}\left(X:E\right)_{\rho}, \label{eq:private2_1}
%	&\overset{\textnormal{(b)}}{\leq} \eps_2 \label{eq:private2_1}
	\end{align}
		where (a) and (b) follow from the fact that the trace-norm is invariant with respect to tensor-product states.
%	, and (b) holds if we choose
%	\begin{align}
%	\log K = I_\text{h}^{1 - \eps_2 + \delta_2}\left( X: E\right)_\rho + \log\frac{\nu}{4\delta_2^4}. \label{eq:private2_0}
%	\end{align}
	Combining \eqref{eq:private12}, \eqref{eq:private1_2}, and  \eqref{eq:private2_1},
%	\eqref{eq:private1_0}, and \eqref{eq:private2_0},
we complete the proof of the first claim. 
The second claim follows from Theorem~\ref{theo:exponent}.
%	we obtain
%	\begin{align}
%	\Delta_\textnormal{\texttt{private}}(M,K)_\rho \leq \eps_1 + \eps_2
%	\end{align}
\end{proof}

\subsection{Secret Key Distillation} \label{sec:secrete-key}
\hfill \break
A \emph{secret-key distillation} protocol aims to distill independent and uniform  secret key $$\vartheta_{M \hat{M}} := \frac{1}{M} \sum_{m\in[M]} |m\rangle\langle m|_M \otimes |m\rangle\langle m|_{\hat{M}}$$ shared between Alice (holding register $M$) and Bob (holding register $\hat{M}$) against Eve \cite{RR12, TLG+12, PR14, KKG+19, KW20}.
In this section, we consider a model \cite{KKG+19, KW20} described as follows.
\begin{defn}[Secret key distillation] \label{defn:secret-key}
		Let $\rho_{XBE} = \sum_{x\in\mathcal{X}} p_X(x) |x\rangle \langle x |\otimes \rho_{BE}^x$ be a classical-quantum-quantum state.
%	Let $\mathcal{N}_{X\to BE}: x\to \rho_{BE}^x$ be a classical-quantum wiretap channel.
	\begin{enumerate}[1.]
		\item Alice holds a classical register $X$ at Alice; Bob and Eve hold quantum registers $B$ and $E$, respectively.
		
		\item Alice performs an encoding operation $\mathcal{E}_{X\to M Z_A}$, where $M$ is the classical register representing the key possessed at Alice, and $Z_A$ is a classical register to be sent to later Bob.
		
		\item Alice announces her register $Z_A$ via free and public classical communication.
				We denote by $\mathcal{I}_{Z_A \to Z_B Z_E}: |z\rangle\langle z|_{Z_A} \mapsto |z\rangle\langle z|_{Z_B} \otimes |z\rangle\langle z|_{Z_E}$ as an isometry such that Bob and Eve hold identical registers $Z_B$ and $Z_E$, respectively.
		
		\item Bob performs a decoding operation $\mathcal{D}_{Z_B B\to \hat{M}}$, where $\hat{M}$ is the classical register representing the resulting key possessed at Bob.
		
	\end{enumerate}
	
	A $(\log M, \eps)$ secret-key distillation protocol for $\rho_{XBE}$ exists if the \emph{security error} \cite{TLG+12, PR14, KKG+19, KW20} satisfies
	\begin{align}
		\frac12 \left\| \sigma_{M \hat{M} Z_E E} - \vartheta_{M \hat{M}} \otimes \sigma_{Z_E E} \right\|_1 \leq \eps,
	\end{align}
	where the overall joint is $\sigma_{M \hat{M} Z_E E} := \mathcal{D}_{Z_B B \to \hat{M}} \circ \mathcal{I}_{Z_A \to Z_B Z_E} \circ \mathcal{E}_{X\to M Z_A} (\rho_{XBE})$.
\end{defn}

Note that a canonical tripartite secret key distillation starts with a pure state $\psi_{ABE}$ that system $E$ at Eve purifies the original bipartite state $\rho_{AB}$ shared between Alice and Bob.
As pointed out in \cite[\S 15.1]{KW20}, Alice and Bob allow to perform \emph{operations and public classical communication} to obtain a classical-quantum-quantum state as described in our secret-key distillation model in Definition~\ref{defn:secret-key} (with Eve holding the classical register during communication).
In the following, we demonstrate how the established analysis for convex splitting can improve the distillable key for Definition~\ref{defn:secret-key}; other extensions would immediately apply.
Lastly, we remark that a close relation between secret key distillation and quantum wiretap channel coding was observed in \cite[Section IV]{KKG+19}. That is the reason why we can easily adapt the analysis given in Section~\ref{sec:private} to this situation.

\begin{theo}[One-shot error exponent for secrete key distillation] \label{theo:secret_exponent}
	Let $\rho_{XBE} = \sum_{x\in\mathcal{X}} p_X(x) |x\rangle \langle x |\otimes \rho_{BE}^x$ be a classical-quantum-quantum state.
	There exists a $\left(\log M, \eps \right)$ (one-way) secret-key distillation protocol for $\rho_{XBE}$ such that for any $K\in\mathds{N}$,
	\begin{align}
		\eps &\leq \Delta^{\textnormal{\texttt{p}}}_{MK}\left(X:B\right)_{\rho} + \Delta^{\textnormal{\texttt{c}}}_{K}\left(X:E\right)_{\rho} \\
		&\leq \mathrm{e}^{- \sup_{\alpha\in [ \sfrac12,1]} \frac{1-\alpha}{\alpha} \left( I^\uparrow_{2 - \sfrac{1}{\alpha}}\left(X:B\right)_\rho - \log(MK) \right)  } + \mathrm{e}^{- \sup_{\alpha\in[1,2]} \frac{\alpha-1}{\alpha} \left( \log K -  I^*_{\alpha}\left(X:E\right)_\rho  \right) }.
	\end{align}

	Moreover, the overall error exponent is positive if and only if there exists $K$ such that $\log (MK) < I(X:B)_\rho$ and $\log K > I(X:E)_\rho$, namely, $\log M < I(X:B)_\rho - I(X:E)_\rho$.
\end{theo}

%\begin{remark}
%	Theorem~\ref{theo:private_exponent} demonstrates that the privacy error (using random coding) decomposes to the ``\emph{packing error}" $\Delta^{\textnormal{\texttt{p}}}_{MK}\left(X:B\right)_{\rho}$ that describes the \emph{error probability of decoding} for communication from Alice to Bob
%	and the ``\emph{covering error}" $\Delta^{\textnormal{\texttt{c}}}_{K}\left(X:E\right)_{\rho}$ that describes the \emph{secrecy error} (i.e.~the information leakage) from Alice to Eve.
%\end{remark}

%\begin{remark}
%	Theorem~\ref{theo:private_exponent} for bounding the secrecy part (i.e.~$\Delta^{\textnormal{\texttt{c}}}_{K}\left(X:E\right)_{\rho}$) has the following advantages compared to the existing works (e.g.~\cite{Hay15}).
%	Firstly, the upper bound on the term~$\Delta^{\textnormal{\texttt{c}}}_{K}\left(X:E\right)_{\rho}$ is expressed in terms of the sandwiched R\'enyi divergence $D_\alpha^*$ for $\alpha \in [1,2]$, which gives the tightest error exponent among all additive quantum R\'enyi divergences \cite{MDS+13, WWY14, Tom16}.
%	Secondly, some of the previous secrecy bounds were obtained under error criteria such as relative entropy or purified distance \cite{Hay15}.
%	By translating those criteria back to trace distance (e.g.~by Pinsker's inequality) will result in a looser bound, e.g.~with one-half of the error exponent.
%	Since Theorem~\ref{theo:private_exponent} directly analyzes trace distance as the secrecy criterion, our bound does not suffer from such a one-half factor.
%\end{remark}

\begin{remark}
	Again, following \cite[\S 4]{KKG+19}, our result in Theorem~\ref{theo:secret_exponent} straightforwardly extends to compound wiretap source state $\{\rho_{XBE}^s\}_{s\in\mathcal{S}}$ for some set $\mathcal{S}$.
\end{remark}

%Then, applying Lemma~\ref{lemm:packing}-\ref{lemma:packing-(a)} and Theorem~\ref{theo:direct} to Lemma~\ref{lemm:private}, we obtain:
Applying Theorem~\ref{theo:direct} to
Theorem~\ref{theo:secret_exponent} yields the following one-shot distillable key length for the source $\rho_{XBE}$ and the corresponding second-order coding rate for $\rho_{XBE}^{\otimes n}$.
\begin{theo}[One-shot distillable secret key and second-order coding rate] \label{theo:secrete_rate}
	Let $\rho_{XBE} = \sum_{x\in\mathcal{X}} p_X(x) |x\rangle \langle x |\otimes \rho_{BE}^x$ be a classical-quantum-quantum state.
	For every $\eps_1, \eps_2>0$ with $\eps_1+\eps_2<1$, there exists a $\left(\log M, \eps_1 + \eps_2 \right)$ (one-way) secret-key distillation protocol for $\rho_{XBE}$ such that,
	\begin{align}
		\log M\geq I_\textnormal{h}^{\eps_1-\delta_1}\left( X:B \right)_\rho - I_\textnormal{h}^{1 - \eps_2 + 3\delta_2}\left( X: E\right)_\rho - \log \tfrac{1}{\delta_1} - \log\tfrac{\nu^2}{\delta_2^4},
	\end{align}
	where $\rho_{XBE} = \sum_{x\in\mathcal{X}} p_X(x)|x\rangle \langle x|\otimes \rho_{BE}^x$ and  $\nu = |\mathsf{E}|$.
	
	Moreover, for sufficiently large $n\in\mathds{N}$, there exists a $(\log M, \eps_1+\eps_2)$ private communication protocol for $\rho_{XBE}^{\otimes n}$ such that,
	\begin{align}
		\log M \geq n\left(I(X:B)_\rho - I(X:E)_\rho \right) + \sqrt{ n V(X:B)_\rho}  \Phi^{-1}(\eps_1) + \sqrt{ n V(X:E)_\rho}  \Phi^{-1}(\eps_2) - O(\log n).
	\end{align}
	%	where $\Phi^{-1}({\eps}) := \sup \{ \gamma\in \mathds{R} : \frac{1}{\sqrt{2\pi}} \int_{-\infty}^{\gamma} \mathrm{e}^{-u^2/2} \, \mathrm{d} u \}$.	
\end{theo}
\begin{remark}
	Theorem~\ref{theo:secrete_rate} quadratically improves the information leakage to Eve (i.e.~$\eps_2$) upon \cite[\S (251)]{KKG+19}, and hence, we obtain a better second-order achievable coding rate.
\end{remark}

\begin{proof}[Proof of Theorem~\ref{theo:secret_exponent}]

We follow the protocol given in \cite[\S 15.1.4]{KW20} and \cite[\S 4]{KKG+19} and prove the claimed error bound. For completeness, we briefly describe it here.

Alice chooses indices $(m,k) \in [M]\times [K]$ uniformly at random, for which $m$ means the realization of the key to be distilled and $k$ is a random index for the security purpose.
Then, Alice prepares $MK-1$ copies of her system $\rho_X = \sum_{x\in\mathcal{X}} p_X(x) |x\rangle\langle x|_X$ and announce it via public communication.
Namely, the resulting state shared between Alice (holding systems labelled $X$, $M$, and $K$), Bob (holding systems labelled $X'$ and $B$), and Eve (holding systems labelled $X''$ and $E$) is
\begin{align}
	\begin{split} \label{eq:secret_state}
	\rho_{MK X^{MK} X'^{MK} X''^{MK} BE} &:= \frac{1}{MK} \sum_{(m,k)\in[M]\times [K]} |m\rangle\langle m|_M \otimes |k\rangle\langle k|_K \otimes \rho_{X^{MK} X'^{MK} X''^{MK} BE}^{m,k},  \\
	\rho_{X^{MK} X'^{MK} X''^{MK} BE}^{m,k} &:= \rho_{X_{m,k} X_{m,k}' X_{m,k}'' BE} \bigotimes_{(\bar{m},\bar{k}) \neq (m,k)} \rho_{X_{\bar{m},\bar{k}} X_{\bar{m},\bar{k}}' X_{\bar{m},\bar{k}}''},
	\end{split}
\end{align}
where $\rho_{X^{MK} X'^{MK} X''^{MK} BE}$ denotes the state conditioned on the realization of $(m,k)$.

Our goal is to show that there is a measurement $\mathcal{M}_{X'^{MK} B\to\hat{M}}$ at Bob such that the associated secrecy error can be well estimated:
\begin{align}
	\eps &:=\frac12\left\| \mathcal{M}_{X'^{MK} B\to\hat{M}}\left(  \rho_{M X'^{MK} X''^{MK} BE} \right) - \vartheta_{M\hat{M}} \otimes \rho_{X''^{MK}} \otimes \rho_E \right\|_1 \notag
	\\
	\begin{split} \label{eq:secret0}
	&\leq
	\frac12\left\| \mathcal{M}_{X'^{MK} B\to\hat{M}}\left(  \rho_{M X'^{MK} X''^{MK} BE} \right) - \sum_{ \in[M]}\frac{1}{M} |m\rangle\langle m|_{M} \otimes |m\rangle\langle m|_{\hat{M}} \otimes \frac{1}{K} \sum_{k\in[K]}\rho_{X''^{MK} E}^{m,k} \right\|_1 \\
	&\quad
	+ \frac12\left\| \sum_{m\in[M]}\frac{1}{M} |m\rangle\langle m|_{M} \otimes |m\rangle\langle m|_{\hat{M}} \otimes \rho_{X''^{MK} E}^{m,k} - \vartheta_{M\hat{M}} \otimes \rho_{X''^{MK}} \otimes \rho_E \right\|_1.
	\end{split}
\end{align}
We then upper bound the two terms on the right-hand side of \eqref{eq:secret0}, separately.

Using the direct-sum structure of trace norm, we bound the first term of \eqref{eq:secret0} as:
\begin{align}
	2\eps_1 &:= \left\| \mathcal{M}_{X'^{MK} B\to\hat{M}}\left(  \rho_{M X'^{MK} X''^{MK} BE} \right) - \sum_{m\in[M]}\frac{1}{M} |m\rangle\langle m|_{M} \otimes |m\rangle\langle m|_{\hat{M}} \otimes \frac{1}{K} \sum_{k\in[K]} \rho_{X''^{MK} E}^{m,k} \right\|_1 \notag \\
	&\leq \frac1M \sum_{m\in[M]} \left\| \mathcal{M}_{X'^{MK} B\to\hat{M}}\left( \frac{1}{K} \sum_{k\in[K]} \rho_{X'^{MK} X''^{MK} BE}^{m,k} \right) - |m\rangle\langle m|_{\hat{M}} \otimes \frac{1}{K} \sum_{k\in[K]} \rho_{X''^{MK} E}^{m,k} \right\|_1 \\
	&\overset{\textnormal{(a)}}{=} \frac1M \sum_{m\in[M]} \left\| \mathcal{M}_{X'^{MK} B\to\hat{M}}\left( \frac{1}{K} \sum_{k\in[K]} \rho_{X'^{MK} B}^{m,k} \right) - |m\rangle\langle m|_{\hat{M}}  \right\|_1 \\
	&\overset{\textnormal{(b)}}{=} \frac{1}{MK} \sum_{(m,k)\in[M]\times [K]} \left\| \mathcal{M}_{X'^{MK} B\to\hat{M}}\left( \rho_{X'^{MK} B}^{m,k} \right) - |m\rangle\langle m|_{\hat{M}} \right\|_1 \\
	&\overset{\textnormal{(c)}}{\le } \frac{1}{MK} \sum_{(m,k)\in[M]\times [K]} \left\| \mathcal{M}_{X'^{MK} B\to\hat{M} \hat{K}}\left( \rho_{X'^{MK} B}^{m,k} \right) - |m\rangle\langle m|_{\hat{M}} \otimes |k\rangle\langle k|_{\hat{K}} \right\|_1. \label{eq:secret10}
\end{align}
Here, (a) follows from Fact~\ref{fact:invariance};
(b) follows from convexity of trace norm;
and (c)	relies on the fact that trace distance is increasing by appending additional systems $\hat{K}$.

Bob is allowed to choose the best measurement $\mathcal{M}_{X'^{MK} B\to\hat{M} \hat{K}}$ in \eqref{eq:secret10}. Hence, by relating it to the packing error given in \eqref{eq:packing}, we obtain
\begin{align}
	\eps_1
	\leq \Delta^{\textnormal{\texttt{p}}}_{MK}\left(X:B\right)_{\rho}. \label{eq:secret1}
\end{align}

Next, we upper bound the second term of \eqref{eq:secret0}. Again, using the direct-sum structure of trace norm, we obtain:
\begin{align}
	\eps_2 &:=
	\frac12\left\| \sum_{m\in[M]}\frac{1}{M} |m\rangle\langle m|_{M} \otimes |m\rangle\langle m|_{\hat{M}} \otimes \frac{1}{K} \sum_{k\in[K]} \rho_{X''^{MK} E}^{m,k} - \vartheta_{M\hat{M}} \otimes \rho_{X''^{MK}} \otimes \rho_E \right\|_1 \\
	&= \frac{1}{M}\sum_{m\in[M]} \frac12\left\| \frac{1}{K} \sum_{k\in[K]} \rho_{X''^{MK} E}^{m,k} -  \rho_{X''^{MK}} \otimes \rho_E \right\|_1 \\
	&= \frac{1}{M} \sum_{m\in[M]}\Delta^{\textnormal{\texttt{c}}}_{K}\left(X:E\right)_{\rho}\\
	&= \Delta^{\textnormal{\texttt{c}}}_{K}\left(X:E\right)_{\rho}, \label{eq:secret2}
\end{align}
where we invoke the covering error given in \eqref{eq:covering}.

Combining \eqref{eq:secret0}, \eqref{eq:secret1}, and  \eqref{eq:secret2},
we obtain the first claim of the proof.
The second claim immediately follows from the one-shot error-exponent of convex splitting proved in Theorem~\ref{theo:exponent} and the one-shot packing Lemma~\ref{lemm:packing}-\ref{lemma:packing-(b)}.
\end{proof}

\subsection{One-Way Quantum Message Compression} \label{sec:message}
We first define an information quantity.
For any bipartite density operator $\rho_{AB} \in \mathcal{S}(\mathsf{A}\otimes \mathsf{B})$ and $\alpha \in (1,\infty)$, we define:
\begin{align} \label{eq:double}
	I^{\downarrow \downarrow}(A:B)_\rho := \inf_{(\tau_A, \sigma_B)\in \mathcal{S}(\mathsf{A}) \times \mathcal{S}(\mathsf{B})}  D_\alpha^*(\rho_{AB}\,\Vert\, \tau_A \otimes \sigma_B).
\end{align}
The above quantity for the classical case was introduced independently by Tomamichel and Hayashi \cite{TH18} and Lapidoth and Pfister \cite{LP19}.

We follow the  one-way quantum message compression protocol studied by Anshu, Devabathini, and Jain\cite{ADJ17
}.
\begin{defn}[One-way quantum message compression] \label{def:one-way}
	Let $\Psi_{RACB}$ be a pure state as input of the protocol.
	\begin{enumerate}[1.]
		\item Alice holds quantum registers $AC$; Bob holds quantum register $B$; and $R$ is an inaccessible Reference.
		
		\item A resource of perfect unlimited entanglement are shared between Alice and Bob, and Alice send messages to Bob via  limited noiseless one-way classical communication.
		
		\item Alice applies an encoding on her system and the shared entanglement to obtain a
		compressed message set of size $M$.
		%	quantum system $M_\text{q}$.
		
		\item Alice sends the above message to Bob via the noiseless classical communication. % {\color{red}\cite{BW92}}.
		%	quantum communication.
		
		\item Bob applies a decoding on the received system and his shared entanglement to obtain an overall state $\hat{\Psi}_{RACB}$ where system  $C$ is at Bob now.
		\end{enumerate}
		A $(\log M, \eps)$ one-way quantum message compression protocol for $\Psi_{RACB}$ satisfies
		\begin{align}
			\frac12\left\| \Psi_{RACB} - \widehat{\Psi}_{RACB} \right\|_1 \leq \eps.
		\end{align}
		Here, $\log M$ is the \emph{classical communication cost}.
	
\end{defn}

Applying Theorem~\ref{theo:exponent} in Section~\ref{sec:exponet} gives us the following one-shot and hence finite blocklength error-exponent bounds.
\begin{theo}[One-way communication cost]
	For any $\Psi_{RACB}$ defined in Definition~\ref{def:one-way}, there exists a $(\log M, \eps)$ one-way quantum message compression protocol satisfying.
	\begin{align}
		\frac12\left\| \Psi_{RACB} - \widehat{\Psi}_{RACB} \right\|_1 \leq f\left( \Delta^{\textnormal{\texttt{c}}}_{M} \left( \Psi_{CRB} \, \Vert \, \tau_C \right) \right), \quad \forall\, \tau_C \in \mathcal{S}(\mathsf{C}).
	\end{align}
	where $f(u):= \sqrt{2u-u^2}$. Moreover, for any block length $n \in \mathds{N}$, by letting $r:= \frac1n \log M$, we have
    \begin{align}
		\frac12\left\| \Psi_{RACB}^{\otimes n} - \widehat{\Psi}_{R^n A^n C^n B^n }\right\|_1 \leq \sqrt{2} \mathrm{e}^{ - n \sup_{\alpha \in [1,2]} \frac{\alpha-1}{2\alpha} \left( r -  I_\alpha^{\downarrow\downarrow} \left( C:RB \right)_{\psi} \right) }, \quad \forall\, n \in \mathds{N}.
	\end{align}	
% 	\begin{align}
% 		\frac12\left\| \Psi_{RACB}^{\otimes n} - \widehat{\Psi}_{R^n A^n C^n B^n }\right\|_1 \leq \sqrt{2} \mathrm{e}^{ - n \sup_{\alpha \in [1,2]} \frac{\alpha-1}{2\alpha} \left( R - \inf_{\tau_C \in \mathcal{S}(\mathsf{C})} I_\alpha^* \left( \Psi_{CRB} \, \Vert \, \tau_C \right) \right) }, \quad \forall\, n \in \mathds{N}.
% 	\end{align}
	The error exponent $\sup_{\alpha \in [1,2]} \frac{\alpha-1}{2\alpha} \left( r - I_\alpha^{\downarrow\downarrow} \left( C:RB \right)_{\psi} \right)$ is positive if and only if
	\begin{align}
		r > \lim_{\alpha\downarrow 1} I_\alpha^{\downarrow\downarrow} \left( C:RB \right)_{\psi} = I(C: RB)_{\Psi}.
	\end{align}
	%	where $I(C: RB)_{\Psi} := D(\Psi_{CRB}\,\Vert\, \Psi_{C}\otimes \Psi_{RB})$ is the quantum mutual information between system $C$ and $RB$ with respect to state $\Psi$.
\end{theo}
\begin{proof}
	The first claim of relating the trace distance to convex splitting follows the protocol by Anshu Devabathini and Jain \cite[Theorem 1]{ADJ17}, and Uhlmann's theorem (Fact~\ref{fact:Uhlmann}).
	To obtain the second claim, we apply the one-shot error exponent of convex splitting (Theorem~\ref{theo:exponent} in Section~\ref{sec:exponet}) and an additivity property stated in Proposition~\ref{prop:properties}-(\ref{item:additivity}) below.
\end{proof}

\begin{prop}\label{prop:properties}
	Let $\rho_{AB} \in \mathcal{S}(\mathsf{A}\otimes \mathsf{B})$ be a bipartite density operator.
	Then, the following hold.
	\begin{enumerate}
		\item\label{item:convexity} (Convexity) The map $(\tau_A, \sigma_B) \mapsto D_\alpha^*(\rho_{AB}\,\Vert\, \tau_A \otimes \sigma_B)$ is convex on $\mathcal{S}(\mathsf{A}) \times \mathcal{S}(\mathsf{B})$ for $\alpha\in (1,\infty)$.
		
		\item\label{item:fixed-point} (Fixed-point property) 
		Suppose the underlying Hilbert spaces are all finite-dimensional.
		For $\alpha \in (1,\infty)$, any pair  $(\tau_A^\star, \sigma_B^\star)$ attaining $D_\alpha^*(\rho_{AB}\,\Vert\, \tau_A^\star \otimes \sigma_B^\star) = I^{\downarrow \downarrow}(A:B)_\rho$
		satisfy
		\begin{align}
			\tau_A^\star &= \frac{ \Tr_B \left[ \left( (\tau_A^\star\otimes \sigma_B^\star)^{\frac{1-\alpha}{2\alpha}} \rho_{AB} (\tau_A^\star\otimes \sigma_B^\star)^{\frac{1-\alpha}{2\alpha}} \right)^\alpha \right] }{ \Tr \left[ \left( (\tau_A^\star\otimes \sigma_B^\star)^{\frac{1-\alpha}{2\alpha}} \rho_{AB} (\tau_A^\star\otimes \sigma_B^\star)^{\frac{1-\alpha}{2\alpha}} \right)^\alpha \right]  }, \label{eq:fixed_tau} \\
			\sigma_B^\star &= \frac{ \Tr_A \left[ \left( (\tau_A^\star\otimes \sigma_B^\star)^{\frac{1-\alpha}{2\alpha}} \rho_{AB} (\tau_A^\star\otimes \sigma_B^\star)^{\frac{1-\alpha}{2\alpha}} \right)^\alpha \right] }{ \Tr \left[ \left( (\tau_A^\star\otimes \sigma_B^\star)^{\frac{1-\alpha}{2\alpha}} \rho_{AB} (\tau_A^\star\otimes \sigma_B^\star)^{\frac{1-\alpha}{2\alpha}} \right)^\alpha \right]  }. \label{eq:fixed_sigma}
		\end{align}
	
		\item\label{item:additivity} (Additivity) 
		Suppose the underlying Hilbert spaces are all finite-dimensional.		
		For any density operators $\rho_{A_1 B_1} \in \mathcal{S}(\mathsf{A}_1 \otimes \mathsf{B}_1 )$ and $\omega_{A_2 B_2} \in \mathcal{S}(\mathsf{A}_2 \otimes \mathsf{B}_2 )$, then for any $\alpha\in (1,\infty)$, the information quantity defined in \eqref{eq:double} satisfies:
		\begin{align} \label{eq:additivity}
			I^{\downarrow \downarrow}(A_1 A_2:B_1 B_2)_{\rho\otimes \omega} 
			= I^{\downarrow \downarrow}(A_1:B_1)_{\rho} + I^{\downarrow \downarrow}(A_2:B_2)_{\omega} . 
		\end{align}
%		\begin{align}
%			\inf_{ \tau_{A_1 A_2} \in \mathcal{S}(\mathsf{A}_1 \otimes  \mathsf{A}_2 ) }	I_\alpha^*\left(\rho_{A_1 B_1} \otimes \rho_{A_2 B_2} \,\|\,\tau_{A_1 A_2}\right)
%			=  \inf_{ \tau_{A_1} \in \mathcal{S}(\mathsf{A}_1) }	I_\alpha^*\left(\rho_{A_1 B_1} \,\|\,\tau_{A_1}\right)
%			+ \inf_{ \tau_{A_2} \in \mathcal{S}(\mathsf{A}_2) }	I_\alpha^*\left(\rho_{A_2 B_2} \,\|\,\tau_{A_2}\right).
%		\end{align}
	\end{enumerate}
	
\end{prop}
\begin{remark}
	Lapidoth and Pfister \cite{LP19} first established thorough properties of the quantity $I^{\downarrow \downarrow }_\alpha$ such as the convexity, fixed-point characterization, and the additivity in the classical setting.
	Proposition~\ref{prop:properties} is thus a generalization of part of the classical results in Ref.~\cite{LP19}.
	The convexity (Proposition~\ref{prop:properties}-(\ref{item:convexity})) relies on a non-trivial log-convexity property given in Lemma~\ref{lemm:log-convexity} of Appendix~\ref{sec:additivity}.
	The fixed-point property (Proposition~\ref{prop:properties}-(\ref{item:fixed-point})) follows from the convexity and the standard Fr\'echet derivative approach as in \cite{HT14, CHT19, MO18, RT22}.
	The additivity (Proposition~\ref{prop:properties}-(\ref{item:additivity})) is proved based on the fixed-point property.
	Similar idea for proving the additivity was used in the work by Nakibo\u{g}lu  for studying the classical Augustin information \cite{nakiboglu_augustin_2018}, the work by Mosonyi and Ogawa for studying the quantum Augustin information \cite{MO18}, and the work by Rubboli and  Tomamichel for studying the $\alpha$-$z$ R\'enyi divergence \cite{RT22}.
\end{remark}

\noindent The proof of Proposition~\ref{prop:properties} is deferred to Appendix~\ref{sec:additivity}.

% \begin{lemm}[An additivity property] \label{lemm:additivity}
% 	For any density operators $\rho_{A_1 B_1} \in \mathcal{S}(\mathsf{A}_1 \otimes \mathsf{B}_1 )$ and $\rho_{A_2 B_2} \in \mathcal{S}(\mathsf{A}_2 \otimes \mathsf{B}_2 )$, then for any $\alpha> 1$, the information quantity defined in \eqref{eq:I_tau} satisfies:
% 	\begin{align}
% 		\inf_{ \tau_{A_1 A_2} \in \mathcal{S}(\mathsf{A}_1 \otimes  \mathsf{A}_2 ) }	I_\alpha^*\left(\rho_{A_1 B_1} \otimes \rho_{A_2 B_2} \,\|\,\tau_{A_1 A_2}\right)
% 		=  \inf_{ \tau_{A_1} \in \mathcal{S}(\mathsf{A}_1) }	I_\alpha^*\left(\rho_{A_1 B_1} \,\|\,\tau_{A_1}\right)
% 		+ \inf_{ \tau_{A_2} \in \mathcal{S}(\mathsf{A}_2) }	I_\alpha^*\left(\rho_{A_2 B_2} \,\|\,\tau_{A_2}\right).
% 	\end{align}
% \end{lemm}

\subsection{Quantum Measurement Simulation} \label{sec:measurement}
We describe the protocol of quantum measurement simulation below \cite{MP00, WM01, Win04, WHB+12, AJW19b}.
Throughout this section, for any POVM, say $\{\Pi_A^u \}_u$, we define the associated quantum measurement:
\begin{align}
    \mathcal{M}_{A\to U}^{\Pi}(\rho_A) := \sum_u \Tr\left[ \rho_A \Pi_A^u \right] |u\rangle \langle u|_U.
\end{align}
\begin{defn}[Non-feedback measurement compression] \label{defn:measurement}
	Let $\rho_A \in \mathcal{S}(\mathsf{A})$, $\left\{ \Lambda_A^x \right\}_{x\in\mathcal{X}}$ be a POVM on $\mathsf{A}$, and $\eps\in[0,1]$.
	A one-shot non-feedback quantum measurement compression protocol consists of the following.
	\begin{enumerate}[1.]
		\item Alice holds a quantum  register $A$ and
		classical registers $M$ and $L$; Bob holds  classical registers $M' \cong M$ and $X$.
		
		\item A resource $\vartheta_{MM'} $ of uniform and independent distribution on $[M]$ representing the uniform randomness shared at $M$ and $M'$.
		
		\item Alice applies an encoding operation $\mathcal{E}_{AM\to L}$ on her system $A$ and her part of the shared randomness $M$, and then send the resulting system $L$ to Bob via a noiseless classical channel.
		
		\item Bob applies the a decoding operation $\mathcal{D}_{LM'\to X}$ on the received system $L$ and his part of the shared randomness $M'$ to simulate the output classical system $X$.
	\end{enumerate}

	A $(\log M, \log L, \eps)$ (non-feedback) quantum measurement compression protocol for $\rho_A$ and $\{\Lambda_A^x\}_{x\in\mathcal{X}}$ satisfies
	\begin{align}
		\frac12 \left\| (\mathcal{D}\circ\mathcal{E})\left( \phi^{\rho}_{RA} \otimes \vartheta_{MM} \right) - \mathcal{M}^{\Lambda}_{A\to X}\left( \phi^{\rho}_{RA} \right) \right\|_1 \leq \eps
	\end{align}
	for any purification $\phi^{\rho}_{RA}$ of $\rho_A$, and $\mathcal{M}^{\Lambda}$ is a quantum measurement with respect to the POVM $\left\{ \Lambda_A^x \right\}_{x\in\mathcal{X}}$.
	The quantities $\log M$ denotes the consumption of shared randomness and $\log L$ represents the classical communication cost.
\end{defn}

With Theorem~\ref{theo:exponent}, we have the following one-shot error-exponent result.
\begin{theo}[One-shot error exponent for non-feedback measurement compression] \label{theo:measurement_one-shot}
	Let $\rho_A \in \mathcal{S}(A)$, $\left\{ \Lambda_A^x \right\}_{x\in\mathcal{X}}$ be a POVM on $\mathsf{A}$.
	For every $M,L \in \mathds{N}$, there exists $(\log M, \log L, \eps)$ non-feedback quantum measurement compression protocol 
	for $\rho_A$ and $\left\{ \Lambda_A^x \right\}_{x\in\mathcal{X}}$
	satisfying
	\begin{align}
	\eps &\leq
 \Delta^{\textnormal{\texttt{c}}}_{L}\left(U:RX\right)_{\theta} +
	 	f\left( \Delta^{\textnormal{\texttt{c}}}_{ML}\left(U:R\right)_{\theta} \right) \label{eq:measurement_one-shot} \\
	&\leq \mathrm{e}^{- \sup_{\alpha\in[1,2]} \frac{1-\alpha}{\alpha} \left( I^*_{\alpha}\left(U:RX\right)_\theta - \log L \right) }
%	\notag \\ 	&\quad
	+ \sqrt{2} \cdot \mathrm{e}^{- \sup_{\alpha\in[1,2]} \frac{\alpha-1}{2\alpha} \left( \log (ML) - I^*_{\alpha}\left(U:R\right)_\theta  \right) },
	\end{align}
	for all Markov states $R$-$U$-$X$ of the form:
	\begin{align}  \label{eq:measurement_state}
	\theta_{RUX} := \sum_{x,u} p_{X|U}(x|u) |x\rangle \langle x|_X \otimes |u\rangle \langle u|_U \otimes \Tr_A\left[  \mathds{1}_R \otimes \Pi_A^u \left( \phi^{\rho}_{RA} \right) \right],
	\end{align}
	with $\phi^{\rho}_{RA}$ being a purification of $\rho_A$, and %{\color{red}the union} is
	with respect to all decompositions of the form:
	\begin{align}
	\mathcal{M}^{\Lambda}_{A\to X}(\rho_A) = \sum_{x,u} p_{X|U}(x|u) \Tr[\rho_A \Pi_A^u] \otimes |x\rangle \langle x|_X
	\end{align}
	for some POVM $\{\Pi_A^u\}_{u}$.
	Here, $f(u):= \sqrt{2u-u^2}$.
\end{theo}

%{\color{red}
%\begin{remark}
%	
%	At this moment we can only show
%	\begin{align}
%	\eps \leq
%	\sqrt{ \Delta^{\textnormal{\texttt{c}}}_{L}\left(U:RX\right)_{\theta} } +
%	\Delta^{\textnormal{\texttt{c}}}_{ML}\left(U:R\right)_{\theta},
%	\end{align}
%	where the sqrt-root comes from switching between the trace distance and the purified distance (since Uhlmann's Theorem is used as a standard approach of showing the existence of the encoder).
%\end{remark}
%}

Following~\eqref{eq:measurement_one-shot} and application of Theorem~\ref{theo:direct}, we obtain the following one-shot achievable rates.
\begin{theo}[One-shot achievable rates for non-feedback measurement compression] \label{theo:measurement_rate}
	Let $\rho_A \in \mathcal{S}(A)$, $\left\{ \Lambda_A^x \right\}_{x\in\mathcal{X}}$ be a POVM on $\mathsf{A}$, and ${\eps_1}, {\eps_2}>0$ with ${\eps_1} + \sqrt{2\eps_2} \leq 1$.
	Then, there exists a $(\log M, \log L, \eps_1 + \sqrt{2\eps_2})$-one-shot non-feedback quantum measurement compression protocol for $\rho_A$ and $\left\{ \Lambda_A^x \right\}_{x\in\mathcal{X}}$
% 	(c.f.~Definition~\ref{defn:measurement}) 
	if $\log L$ and $\log M$ are in the following union of regions: for all $\delta_1 \in (0,\sfrac{\sqrt{2\eps_1}}{3})$ and $\delta_2\in (0,\sfrac{{\eps_2}}{3})$,
	\begin{align}
	\begin{dcases}
	\log L \geq I^{1- \sqrt{2\eps_1} + 3\delta_1}_\textnormal{h}(U:R)_\theta + \log \tfrac{|\mathsf{A}|^2}{\delta_1^4}, \\
	\log ML \geq I^{1- \eps_2 + 3\delta_2}_\textnormal{h}(U:RX)_\theta + \log \tfrac{|\mathsf{A}|^2}{\delta_2^4},
	\end{dcases}
	\end{align}
	where the information quantities are evaluated with respect for all Markov states $R$-$U$-$X$ of the form:
	\begin{align}  \label{eq:measurement_state}
	\theta_{RUX} := \sum_{x,u} p_{X|U}(x|u) |x\rangle \langle x|_X \otimes |u\rangle \langle u|_U \otimes \Tr_A\left[  \mathds{1}_R \otimes \Pi_A^u \left( \phi^{\rho}_{RA} \right) \right],
	\end{align}
	with $\phi^{\rho}_{RA}$ being a purification of $\rho_A$, and %{\color{red}the union} is
	with respect to all decompositions of the form:
	\begin{align}
	\mathcal{M}^{\Lambda}_{A\to X}(\rho_A) = \sum_{x,u} p_{X|U}(x|u) \Tr[\rho_A \Pi_A^u] \otimes |x\rangle \langle x|_X
	\end{align}
	for some POVM $\{\Pi_A^u\}_{u}$.
\end{theo}
\begin{proof}[Proof of Theorem~\ref{theo:measurement_one-shot}]%[Unfinished proof]
	For any state $\theta_{RUX}$ in \eqref{eq:measurement_state} and letter $u$, we define
	\begin{align}
		\theta_R^u &:= \frac{1}{p_U(u)}\Tr_A\left[  \mathds{1}_R \otimes \Pi_A^u \phi^{\rho}_{RA}  \right];\\
% 		= \frac{1}{p_U(u)} \left( \sqrt{\rho} (\Pi^u)^T \sqrt{\rho} \right)_R; \\
		\theta_X^u &:= \sum_{x} p_{X|U}(x|u) |x\rangle \langle x|_X; \\
		p_U(u) &= \Tr\left[ \Pi_A^u \rho_A \right].
	\end{align}
	
	We consider an random codebook $\mathcal{C} := \left\{ u_{m,l}\right\}_{(m,l)\in[M]\times [L]}$ whose codewords are independently and identically distributed according to distribution $p_U$.
	We adopt the standard random codebook argument by considering a realization of the random codebook $\mathcal{C}$ and taking expectation to bound the trace distance in the end.
	
	The proposed procedure of the encoding and decoding operations are defined as follows:
	\begin{align}
		\mathcal{E}_{AM\to L}(\,\cdot\,) &:= \sum_{m\in[M]} \langle m|\cdot |m\rangle_M  \otimes  \mathcal{M}^{\Upsilon^{(m)}}(\cdot);\\
		\mathcal{D}_{LM\to X}(\,\cdot\,) &:= \sum_{(m,l)\in[M]\times [L]} \langle m,l| \cdot |m,l\rangle_{ML} \otimes \theta_X^{u_{m,l}}.
	\end{align}
	That is, upon sampling a realization $m\in[M]$ of the shared randomness, Alice performs a measurement $\{\Upsilon_{l}^{(m)} \}_{l\in[L]}$ that will be determined later,
%	 according to the POVM:
%	\begin{align}
%		\Upsilon_l^{(m)} &:= \frac{ \theta_A^{u_{m,l}} }{ \sum_{\bar{l}\in[L]} \theta_A^{u_{m,\bar{l}}} }, \\
%		\theta_A^u &:= \frac{1}{p_U(u)} \left( \sqrt{\rho} (M^u)^T \sqrt{\rho} \right)_A,
%	\end{align}
	and send the outcome $l\in[L]$ through a noiseless classical channel to Bob. Based on the received classical information $l$ together with the shared randomness $m$, Bob then prepares the state $\theta_X^{u_{m,l}}$ at his proposal, aiming to simulate the output distribution $\mathcal{M}^{\Lambda}(\rho_A)$.

	Now, we evaluate the performance of the protocol.
	We introduce an artificial state (which depends on $\mathcal{C}$):
	\begin{align}
		\Theta_{RX}^{\mathcal{C}} := \frac{1}{ML}\sum_{(m,l)\in[M]\times [L]} \theta_R^{u_{m,l}} \otimes \theta_X^{u_{m,l}}.
	\end{align}
	Then, triangle inequality implies that
	\begin{align}
		&\frac12 \left\| (\mathcal{D}\circ\mathcal{E})\left( \phi^{\rho}_{RA} \otimes \vartheta_{MM} \right) - \mathcal{M}^{\Lambda}_{A\to X}\left( \phi^{\rho}_{RA} \right) \right\|_1 \\
		&= \frac12\left\| \frac1M\sum_{m,l} \Tr_A\left[ \left( \mathds{1}_R\otimes \Upsilon^{(m)}_l \right)(\phi^\rho_{RA})\right] \otimes \theta_X^{u_{m,l}} - \sum_x \Tr_A\left[ \left(  \mathds{1}_R\otimes \Lambda^x_A\right)(\phi^\rho_{RA})\right] \otimes |x\rangle \langle x| \right\|_1 \\
		\begin{split} \label{eq:measurement_12}
			&\leq \frac12\left\| \frac1M\sum_{m,l} \Tr_A\left[ \left( \mathds{1}_R\otimes \Upsilon^{(m)}_l \right)(\phi^\rho_{RA})\right] \otimes \theta_X^{u_{m,l}} - \Theta_{RX}^{\mathcal{C}} \right\|_1 \\
			&\quad +\frac12\left\| \Theta_{RX}^{\mathcal{C}} - \sum_x \Tr_A\left[ \left(  \mathds{1}_R\otimes \Lambda^x_A\right)(\phi^\rho_{RA})\right] \otimes |x\rangle \langle x| \right\|_1.
		\end{split}
	\end{align}
	Note that
	\begin{align}
	\mathds{E}_{\mathcal{C}} \left[ \Theta_{RX}^{\mathcal{C}} \right]
	= \theta_{RX}
	= \sum_x \Tr_A\left[ \left(  \mathds{1}\otimes \Lambda^x_A\right)(\phi^\rho_{RA})\right] \otimes |x\rangle \langle x|.
	\end{align}	
	Then, after taking expectation over $u_{m,l} \sim p_U$ the second term of \eqref{eq:measurement_12} is
	\begin{align}
	\mathds{E}_{u_{m,l} \sim p_U}	\frac12\left\| \Theta_{RX}^{\mathcal{C}} - \sum_x \Tr_A\left[ \left(  \mathds{1}\otimes \Lambda^x_A\right)(\phi^\rho_{RA})\right] \otimes |x\rangle \langle x| \right\|_1
	= \Delta^{\textnormal{\texttt{c}}}_{ML}\left(U:RX\right)_{\theta}
	\end{align}
	by the direct-sum structure of trace norm.
	
%	The convex-splitting lemma directly implies that
%	\begin{align}
%		\mathds{E}_{\mathcal{C}} \frac12\left\| \Theta_{RX}^{\mathcal{C}} - \sum_x \Tr_A\left[ \left( \id\otimes \Lambda^x_A\right)(\phi^\rho_{RA})\right] \otimes |x\rangle \langle x| \right\|_1 \leq \varepsilon,
%	\end{align}
%	provided that
%	\begin{align}
%		\log ML \geq I^{1-\eps}_\text{h}(U:RX)_\theta.
%	\end{align}
	
%	To compute the first term of \eqref{eq:measurement_12}, we start with
%	\begin{align}
%		\frac12\left\| \rho_R - \frac1L\sum_{l\in[L]} \theta_R^{u_{m,l}} \right\|_1, \quad \forall,\, m\in[M].
%	\end{align}
%	Note that since we take $u_{m,l} \sim p_U$ for each $(m,l)\in[M]\times [L]$ independently, for each $m\in[M]$, we have
%	\begin{align}
%		\Delta^{\textnormal{\texttt{c}}}_{L}\left(U:R\right)_{\theta}
%		&= \mathds{E}_{u_{m,l}\sim p_U }\frac12\left\| \rho_R - \frac1L\sum_{l\in[L]} \theta_R^{u_{m,l}} \right\|_1, \quad \forall,\, m\in[M] \\
%		&= \frac1M \sum_{m\in[M]} \mathds{E}_{u_{m,l}\sim p_U }\frac12\left\| \rho_R - \frac1L\sum_{l\in[L]} \theta_R^{u_{m,l}} \right\|_1, \quad \forall,\, m\in[M]
%	\end{align}
%	Recall Uhlmann's theorem (Fact~\ref{fact:Uhlmann}), we then have%, for each $m\in[M]$,
%	\begin{align}
%		f
%	\end{align}
	The first term of \eqref{eq:measurement_12} is computed as follows.
	Using convexity of trace-norm, 	we calculate: 
	% of trace norm and Uhlmann's theorem (Fact~\ref{fact:Uhlmann}), we calculate
\begin{align}
	&\frac12\left\| \frac1M\sum_{m,l} \Tr_A\left[ \left( \mathds{1}\otimes \Upsilon^{(m)}_l \right)(\phi^\rho_{RA})\right] \otimes \theta_X^{u_{m,l}} - \Theta_{RX}^{\mathcal{C}} \right\|_1 \notag \\
	&\leq \frac1M \sum_{m\in[M]} \frac12\left\| \sum_{l\in[L]} \Tr_A \left[ \left( \mathds{1}\otimes \Upsilon^{(m)}_l \right)(\phi^\rho_{RA})\right] \otimes \theta_X^{u_{m,l}} - \frac{1}{L} \sum_{l\in[L]}\theta_R^{u_{m,l}} \otimes \theta_X^{u_{m,l}} \right\|_1 \\
	&\overset{\text{(a)}}{\leq} \frac1M \sum_{m\in[M]}
    f\left(  \frac12\left\| \sum_{l\in[L]} \Tr_A \left[ \left( \mathds{1}\otimes \Upsilon^{(m)}_l \right)(\phi^\rho_{RA})\right]  - \frac{1}{L} \sum_{l\in[L]}\theta_R^{u_{m,l}} \right\|_1 \right)\\
    &=\frac1M \sum_{m\in[M]} f\left( \frac12\left\| \rho_R - \frac1L\sum_{l\in[L]} \theta_R^{u_{m,l}} \right\|_1 \right),
	\end{align}
	where in (a) we recall Uhlmann's theorem (Fact~\ref{fact:Uhlmann}); namely, for each $m\in[M]$, we choose the measurement $\{\Upsilon^{(m)}_l \}_l$ such that its Stinespring dilation attains the equality in the Uhlamnn's theorem.

\begin{comment}
{\color{red}
What we can control is:
\begin{align} \label{eq:original}
\eps:= 
\frac12\left\| \rho_R  - \frac{1}{L} \theta_R^{u_{m,l}}  \right\|_1
&=
    \frac12\left\| \sum_{l\in[L]} \Tr_A \left[ \left( \mathds{1}\otimes \Upsilon^{(m)}_l \right)(\phi^\rho_{RA})\right]  - \frac{1}{L} \sum_{l\in[L]}\theta_R^{u_{m,l}}  \right\|_1.
\end{align}
What we start with are the extension:
\begin{align} \label{eq:goal}
    \frac12\left\| \sum_{l\in[L]} \Tr_A \left[ \left( \mathds{1}\otimes \Upsilon^{(m)}_l \right)(\phi^\rho_{RA})\right] \otimes \theta_X^{u_{m,l}} - \frac{1}{L} \sum_{l\in[L]} \theta_R^{u_{m,l}} \otimes \theta_X^{u_{m,l}} \right\|_1.
\end{align}
(Note that tracing out $X$ in \eqref{eq:goal} is \eqref{eq:original}.)

We need to show that for each $m\in[M]$, there exists a POVM $\{ \Upsilon^{(m)}_l \}_{l\in[L]}$ such that \eqref{eq:goal} is upper bounded by $\sqrt{2 \eps}$ via Ulhmann's theorem.

Define the following pure states
\begin{align}
|\psi\rangle_{MLXAR} &:=
\sum_{m,l,x} \sqrt{\frac{\theta_{X}^{u_{m,l}}(x) }{M}} |m,l,x\rangle_{MLX} |\phi^\rho\rangle_{AR}; \\
|\theta\rangle_{MLXAR} &:= \sum_{m,l,x} \sqrt{\frac{\theta_{X}^{u_{m,l}}(x) }{ML}} |m,l,x\rangle_{MLX}|\theta^{u_{m,l}}\rangle_{AR}.
\end{align}
}
\end{comment}

	Using concavity of $f(\cdot)$, we take expectation over $u_{m,l} \sim p_U$ and note that for each $m\in[M]$, $u_{m,l}$ is already independent for all $l\in[L]$, i.e.
	\begin{align}
		&\mathds{E}_{u_{m,l} \sim p_U} \frac12\left\| \frac1M\sum_{m,l} \Tr_A\left[ \left( \mathds{1}_R\otimes \Upsilon^{(m)}_l \right)(\phi^\rho_{RA})\right] \otimes \theta_X^{u_{m,l}} - \theta_{RX}^{\mathcal{C}} \right\|_1 \notag \\
		&\leq \mathds{E}_{u_{m,l}}  \frac1M \sum_{m\in[M]} f\left( \frac12\left\| \rho_R - \frac1L\sum_{l\in[L]} \theta_R^{u_{m,l}} \right\|_1 \right) \\
		&\leq  \frac1M \sum_{m\in[M]} f\left(  \mathds{E}_{u_{m,l}} \frac12\left\| \rho_R - \frac1L\sum_{l\in[L]} \theta_R^{u_{m,l}} \right\|_1 \right) \\
		&= f \left( \Delta^{\textnormal{\texttt{c}}}_{L}\left(U:R\right)_{\theta} \right),
	\end{align}
	which completes our claim in \eqref{eq:measurement_one-shot}. The second claim immediately follows from the one-shot error-exponent of convex splitting proved in Theorem~\ref{theo:exponent}.
\end{proof}

%\newpage
\subsection{Communication with Casual State Information at Encoder} \label{sec:state}
We describe the entanglement-assisted classical communication over quantum channel with casual state information at encoder below \cite{GK11, Dup10, AJW19a, AJW19c}.
\begin{defn}[Classical Communication over quantum channel coding with causal state information]
	Let $\mathcal{N}_{AS \to B}$ be a quantum channel with a pure entangled state $\vartheta_{S_A S}$ shared between encoder (holding $S_A$) and Channel (holding $S$).
	\begin{enumerate}[1.]
		\item Alice holds quantum registers $A$, $A'$, and $S_A$;
		Bob holds quantum registers $B$ and $R'$;
		and an additional quantum register $S$ is at the input of channel.
		
% 		A quantum register ${S}$ at Channel, quantum registers $S_A$, $A$, $A'$ at Alice,
% 		and quantum registers $B$, $R'$ at Bob.
		
		\item An arbitrary resource of entanglement state $\theta_{ A' R'}$ is shared between  Alice  (holding $A'$) and Bob (holding $R'$).
		
		\item Alice applies an encoding $\mathcal{E}_{ S_A A' \to A}^{(m)}$ on her systems $S_A$ and $A'$ for sending each message $m\in[M]$.
		
		\item The quantum channel $\mathcal{N}_{AS \to B}$ is applied on Alice's register $A$ and the channel state system $S$.
		
		\item Bob performs a decoding measurement described by a POVM $\{{\Pi}_{B R'} \}_{m\in[M]}$ on $B$ to extract the sent message $m\in [M]$.
	\end{enumerate}
	
	A $(\log M, \eps)$-code for $\mathcal{N}_{AS \to B}$ with state information $\vartheta_{S_A A}$ is a protocol such that  the average error probability satisfies
	\begin{align}
		\frac{1}{M} \sum_{m\in [M]}	\Tr\left[  \mathcal{N}_{AS \to B} \circ \mathcal{E}_{S_A A' \to A}^m\left( \vartheta_{S_A S} \otimes  \theta_{A' R' }\right) \cdot \left( \mathds{1}- {\Pi}_{B R'}^m \right)\right]
		\leq \varepsilon.
	\end{align}
	The quantity $\log M$ denotes the classical communication rate.
\end{defn}

\begin{theo}[One-shot error exponent for communication with casual state information at encoder]\label{theo:state_exponent}
	Let $\vartheta_{S_A S}$ be the channel state information shared between Alice and Channel.
	Then, for every density operator $\theta_{ARS}$ satisfying $\theta_S = \vartheta_S$ and any $K\in\mathds{N}$, there exists a $(\log M,\eps)$-code for $\mathcal{N}_{AS\to B}$ satisfying
	\begin{align}
	\eps&\leq \Delta_{MK}^\textnormal{\texttt{p}}\left(R:B\right)_{\mathcal{N}_{AS\to B}(\theta) }
		+ f\left( \Delta_{K}^\textnormal{\texttt{c}}\left( R:S\right)_\theta \right) \label{theo:state}\\
	&\leq \mathrm{e}^{- \sup_{\alpha\in[\sfrac12,1]} \frac{1-\alpha}{\alpha} \left( I_{\alpha}^{\uparrow}\left(R:B\right)_{\mathcal{N}_{AS\to B}(\theta) }  - \log (MK) \right) }
	%	\notag \\ 	&\quad
	+ \sqrt{2} \mathrm{e}^{- \sup_{\alpha\in[1,2]} \frac{\alpha-1}{2\alpha} \left( \log K - I^*_{\alpha}\left(R:S\right)_\theta \right) }.
	\end{align}
	Here, $f(u) := \sqrt{2u-u^2}$. Moreover, the overall error exponent (i.e.~minimum of the two error exponents) is positive if and only if there exists $K\in\mathds{N}$ such $\log(MK) < I(R:B)_{\mathcal{N}(\theta)}$ and $\log K > I(R:S)_\theta$, namely, $$\log M < I(R:B)_{\mathcal{N}(\theta)} - I(R:S)_\theta .$$
\end{theo}

Applying Theorem~\ref{theo:direct} to \eqref{theo:state} leads us to the following one-shot achievable rate.
\begin{theo}[One-shot achievable rate and second-order coding rate for communication with casual state information at encoder] \label{theo:state_rate}
	Let $\vartheta_{S_A S}$ be the channel state information shared between Alice and Channel.
	Then, for every density operator $\theta_{ARS}$ satisfying $\theta_S = \vartheta_S$, there exists a $(\log M,\eps_1 + \sqrt{2\eps_2})$-code for $\mathcal{N}_{AS\to B}$ satisfying
	\begin{align}
		\log M \geq
		I_\textnormal{h}^{\eps_1 - \delta_1}\left(R:B\right)_{\mathcal{N}_{AS\to B}(\theta) }
		- I_\textnormal{h}^{1 - \eps_2 + 3\delta_2}\left( R:S\right)_\theta - \log\frac{1}{\delta_1} - \log\frac{\nu^2}{\delta_2^4},
	\end{align}
	where $\nu = |\mathsf{R}||\mathsf{S}|$.
	
	Moreover, for sufficiently large $n\in\mathds{N}$, there exists a $(\log M, \eps_1+ \sqrt{2\eps_2} )$-code for $\mathcal{N}_{AS\to B}^{\otimes n}$ such that the following holds for any density operator $\theta_{ARS}$ satisfying $\theta_S = \vartheta_S$,
	\begin{align}
		\log M \geq n\left(I(R:B)_{\mathcal{N}(\theta)} - I(R:S)_\theta \right) + \sqrt{ n V(R:B)_{\mathcal{N}(\theta)} }   \Phi^{-1}(\eps_1) + \sqrt{ n V(R:S)_\theta}   \Phi^{-1}(\eps_2) - O(\log n).
	\end{align}
%	where $\Phi^{-1}({\eps}) := \sup \{ \gamma\in \mathds{R} : \frac{1}{\sqrt{2\pi}} \int_{-\infty}^{\gamma} \mathrm{e}^{-u^2/2} \, \mathrm{d} u \}$.	
\end{theo}

\begin{proof}[Proof of Theorem~\ref{theo:state_exponent}]
Given a (pure) channel state information $\vartheta_{S_A S}$ shared between Alice (holding system $S_A$) and Channel (holding system $S$),
fix a density operator $\theta_{ARS}$ satisfying $\theta_S = \vartheta_S$ as mentioned in the theorem. The coding 

\begin{itemize}	
	\item \textbf{Preparation:} Alice and Bob share $MK$ identical copies of entangled state $\theta_{R_A R}$.
% 	$\theta_{ R_{A_{1,1}} \ldots R_{A_{M,K}} R_{1,1} \ldots R_{M,K}}$ shared between .
	Hence, the starting state is
	\begin{align} \label{eq:tau_m_o}
		\mathring{\tau}_{S_ASR_{A_{1}}R_{1,1}\ldots R_{A_{M,K}}R_{M,K}} &:= \vartheta_{S_A S}  \bigotimes_{(m,k)\in[M]\times [K]}
		\theta_{R_{A_{m,k}}R_{m,k}},
	\end{align}
    where Alice holds $R_{A_{1,1}} \ldots R_{A_{M,K}}$ and Bob holds $R_{1,1} \ldots R_{M,K}$.
	
	\item \textbf{Encoding:} For each message $m\in[M]$, Alice applies an encoding map $\mathcal{E}_{ S_A R_{A_{m,1}} \ldots R_{A_{m,K}} \to A }$ to arrive at
	\begin{align}
		\mathring{\sigma}_{ASR_{m,1}\ldots R_{m,K}}^{(m)} := \mathcal{E}_{S_A R_{A_{m,1}}^{(m)}\ldots R_{A_{m,K}}\to A}\left(\mathring{\tau}_{S_ASR_{A_{{m,1}}}R_{m,1}\ldots R_{A_{{m,K}}}R_{{m,K}}}\right), \label{eq:true_sigma}
	\end{align}
	and sends her register $A$ via the channel $\mathcal{N}_{AS\to B}$.
	The encoding map  $\mathcal{E}$ will be specified later.

	\item \textbf{Decoding:}
	The channel output state (with his shared entanglement) at Bob	is denoted by
	\begin{align}
		\mathring{\rho}_{BR_{1,1}\ldots R_{M,K}}^{(m)} := \mathcal{N}_{AS\to B}\left( 	\mathring{\sigma}_{ASR_{m,1}\ldots R_{m,K}}\right) \bigotimes_{(\bar{m},k):  \bar{m}\neq m\, \&\, k\in[K] } \theta_{R_{\bar{m},k}}.
	\end{align}
	Then, Bob applies a POVM $\left\{ {\Pi}_{B R_{1,1}\ldots R_{M,K}}^{ (m,k) } \right\}_{(m,k)\in[M]\times [K]}$ to decode message $(m,k)$ and he reports $m$ as the sent message.
\end{itemize}

Before going to the error analysis, let us introduce some notation.
For each $m\in[M]$, we define
\begin{align}
%	\tau_{SR_{m,1}\ldots R_{m,K}} &:=	\frac1K \sum_{k\in[K]} \theta_{SR_{m,k}} \otimes \theta_{R_{m,1}} \otimes \cdots \otimes \theta_{R_{m,K}}; \\
%	{\sigma}_{ASR_{m,1}\ldots R_{m,K}} &:= \mathcal{E}_{S_A R_{A_{m,1}}\ldots R_{A_{m,K}}\to A}\left( {\tau}_{S_ASR_{A_{{m,1}}}R_{m,1}\ldots R_{A_{{m,K}}}R_{{m,K}}}\right); \\
	\tau_{SR_{m,1}\ldots R_{m,K}}^{(m)} &:= \frac1K \sum_{k\in[K]} \theta_{SR_{m,k}} \otimes \theta_{R_{m,1}} \otimes \cdots \otimes \theta_{R_{m,K}}; \label{eq:artifical_tau}\\
	\sigma_{ASR_{m,1}\ldots R_{m,K}}^{(m)} &:= \frac1K \sum_{k\in[K]} \theta_{ASR_{m,k}} \otimes \theta_{R_{m,1}} \otimes \cdots \otimes \theta_{R_{m,K}}; \label{eq:artifical_sigma} \\
	{\rho}_{BR_{1,1}\ldots R_{M,K}}^{(m)} &:= \mathcal{N}_{AS\to B}\left( 	{\sigma}_{ASR_{m,1}\ldots R_{m,K}}\right) \bigotimes_{(\bar{m},k):  \bar{m}\neq m\, \&\, k\in[K] } \theta_{R_{\bar{m},k}}. \label{eq:artifical_rho}
\end{align}
%\begin{align}
%	\begin{split} \label{eq:tau_m}
%		\tau_{SR_{m,1}\ldots R_{m,K}} &:=
%		\frac1K \sum_{k\in[K]} \theta_{SR_{m,k}} \otimes \theta_{R_{m,1}} \otimes \cdots \otimes \theta_{R_{m,K}};  \\
%		\mathring{\tau}_{SR_{m,1}\ldots R_{m,K}} &:= \vartheta_{S} \otimes \theta_{R_{m,1}}\otimes \ldots \otimes \theta_{R_{m,K}},
%	\end{split}
%\end{align}
where $R_{m,k} = R$ for all $(m,k)\in [M]\times[K]$.
%Then, for each $m\in[M]$, we have
%\begin{align}
%	\frac12 \left\| \tau_{SR_{m,1}\ldots R_{m,K}} - \mathring{\tau}_{SR_{m,1}\ldots R_{m,K}} \right\|_1 = \Delta_K^\texttt{c}(R:S)_\theta.
%\end{align}

Following from the coding strategy give above, the error probability for each message $m\in[M]$ is
\begin{align}
	&\Tr\left[ \mathring{\rho}_{BR_{1}\ldots R_{MK}}^{(m)} \cdot \left(\sum_{ (\bar{m},k):  \bar{m}\neq m\, \&\, k\in[K]} {\Pi}_{B R_{1}\ldots R_{M,K}}^{ (\bar{m},k) } \right) \right] \notag \\
	&\overset{\textnormal{(a)}}{\leq} \Tr\left[ {\rho}_{BR_{1}\ldots R_{MK}}^{(m)} \cdot \left(\sum_{ (\bar{m},k):  \bar{m}\neq m\, \&\, k\in[K] } {\Pi}_{B R_{1}\ldots R_{M,K}}^{ (\bar{m},k) } \right) \right]
	+ \frac12\left\| {\rho}_{BR_{1}\ldots R_{MK}}^{(m)} - \mathring{\rho}_{BR_{1}\ldots R_{MK}}^{(m)} \right\|_1, \label{eq:state_12}
\end{align}
where we invoke change of measure (Fact~\ref{fact:1norm}-\ref{item:change}) in (a).
By the definition of ${\rho}^{(m)}_{BR_{1}\ldots R_{MK}}$ in \eqref{eq:artifical_rho}, the first term on the right-hand side of \eqref{eq:state_12} can be calculated as:
\begin{align}
	&\frac{1}{K} \sum_{k\in[K]} \Tr\left[  \mathcal{N}_{AS\to B}\left( {\theta}_{ASR_{m,k}} \right) \bigotimes\limits_{(\bar{m},\bar{k})\neq (m,k)} \theta_{R_{\bar{m},\bar{k}}} \cdot \left(\sum_{ (\bar{m},k):  \bar{m}\neq m\, \&\, k\in[K] } {\Pi}_{B R_{1}\ldots R_{M,K}}^{ (\bar{m},k) } \right) \right]\\
	&\leq \frac{1}{K} \sum_{k\in[K]} \Tr\left[  \mathcal{N}_{AS\to B}\left( {\theta}_{ASR_{m,k}} \right) \bigotimes\limits_{(\bar{m},\bar{k})\neq (m,k)} \theta_{R_{\bar{m},\bar{k}}} \cdot \left(\sum_{ (\bar{m},\bar{k})\neq (\bar{m},\bar{k}) } {\Pi}_{B R_{1}\ldots R_{M,K}}^{ (\bar{m},k) } \right) \right]\\
	&\overset{\textnormal{(a)}}{=} \frac{1}{K} \sum_{k\in[K]} \Delta_{MK}^\texttt{p}\left( R: B \right)_{\mathcal{N}_{AS\to B}(\theta_{ASR})}, \label{eq:state_12_1}
\end{align}
where in (a) we recall the definition of packing error in \eqref{eq:packing}.

Next, the second term in \eqref{eq:state_12} is bounded via monotone decrease of trace distance under the channel $\mathcal{N}_{AS\to B}$, i.e.~
\begin{align}
	\frac12\left\| {\rho}_{BR_{1}\ldots R_{MK}}^{(m)} - \mathring{\rho}_{BR_{1}\ldots R_{MK}}^{(m)} \right\|_1
	&\leq \frac12 \left\| \frac1K \sum_{k\in[K]} \theta_{ASR_{m,k}} \otimes \theta_{R_{m,1}} \otimes \cdots \otimes \theta_{R_{m,K}} - \mathring{\sigma}_{ASR_{m,1}\ldots R_{m,K}} \right\|_1. \label{eq:state_12_2_1}
\end{align}
Here, note that the state $\mathring{\sigma}_{ASR_{m,1}\ldots R_{m,K}}$ is the output of the encoding map $\mathcal{E}_{ S_A R_{A_{(m,1)}} \ldots R_{A_{m,K}} \to A }$ with the product state $\mathring{\tau}_{S_A SR_{m,1}\ldots R_{m,K}}$ as input.
On the other hand, if we could trace output system $A$ at the right-hand side of \eqref{eq:state_12_2_1}, we can formulate The trace distance in terms of a covering error, i.e.~
\begin{align}
	\frac12 \left\| \tau_{SR_{m,1}\ldots R_{m,K}}^{(m)} - \mathring{\tau}_{SR_{m,1}\ldots R_{m,K}} \right\|_1 = \Delta_K^\texttt{c}(R:S)_\theta, \quad \forall\, m\in[M],
\end{align}
where $\tau_{SR_{m,1}\ldots R_{m,k}}^{(m)}$ is introduced in \eqref{eq:artifical_tau}.
% \begin{align}
% 	\tau_{SR_{m,1}\ldots R_{m,K}}^{(m)} = \Tr_{A}\left[ \frac1K \sum_{k\in[K]} \theta_{ASR_{m,k}} \otimes \theta_{R_{m,1}} \otimes \cdots \otimes \theta_{R_{m,K}} \right].
% \end{align}
Hence, we choose the encoding map $\mathcal{E}_{ S_A R_{A_{(m,1)}} \ldots R_{A_{m,K}} \to A }$ such that its Stinespring dilation attains the Uhlmann's theorem in Fact~\ref{fact:Uhlmann} to obtain
\begin{align}
	\frac12 \left\| \frac1K \sum_{k\in[K]} \theta_{ASR_{m,k}} \otimes \theta_{R_{m,1}} \otimes \cdots \otimes \theta_{R_{m,K}} - \mathring{\sigma}_{ASR_{m,1}\ldots R_{m,K}} \right\|_1
	\leq f\left(  \Delta_K^\texttt{c}(R:S)_\theta \right), \quad \forall\, m\in[M]. \label{eq:state_12_2_2}
\end{align}

Combining \eqref{eq:state_12}, \eqref{eq:state_12_1}, \eqref{eq:state_12_2_1}, and \eqref{eq:state_12_2_2} proves our first claim.
The exponential bounds follow from the one-shot error exponent of convex splitting given in Theorem~\ref{theo:exponent}.
\end{proof}

\section{Conclusions and Discussions}
\label{sec:conclusions}

In this paper,
we establish a one-shot error exponent and one-shot strong converse analysis for convex splitting beyond the i.i.d.~asymptotic scenario. The exponents are additive under product (and probably non-identical) states, and hence applies to \emph{any} blocklength. 
Moreover, the positivity of the established error exponent immediately yields a meaningful \emph{achievable rate region}: $\log M > I(\rho_{AB}\,\Vert\,\tau_A)$.
Let us take the original convex split-lemma\footnote{We remark that it has been improved recently by Li and Yao \cite{LY21a, LY21b}} \cite{ADJ17} for illustration. 
The convex-splitting error under \emph{purified distance}, denoted by $\mathrm{P}_M(\rho_{AB}\,\Vert\,\tau_A)$ , was shown to be upper bounded by\footnote{More precisely, in Ref.~\cite{ADJ17}, the state on system $B$ in the second argument of $D_\infty^*$ is not optimized as compared to the one given in \eqref{eq:conclusion_1}. However, let us ignore this subtlety since such a difference can be further estimated.}:
\begin{align} \label{eq:conclusion_1}
    \mathrm{P}_M(\rho_{AB}\,\Vert\,\tau_A) \leq \mathrm{e}^{- \frac12 \left( \log M - I_\infty^*(\rho_{AB}\,\Vert\, \tau_A) \right) }.
\end{align}
The achievable exponent given in \eqref{eq:conclusion_1} is positive only if the  rate is far from the fundamental threshold, i.e.~$\log M > I_{\infty}^*(\rho_{AB}\|\tau_A) \geq I(\rho_{AB}\,\Vert\,\tau_A)$.  We remark that such a gap could be quite large in some scenarios. 

To circumvent this issue, a standard approach is to \emph{smooth} it.
However, it will incur an additional \emph{additive error penalty} such that it will very much deteriorates the error exponent.
Namely, in \cite[Corollary 1]{AJW19a} and \cite[Lemma 2]{Wil17b}, \eqref{eq:conclusion_1} was smoothed to yield
\begin{align} \label{eq:conclusion_smoothed}
    \mathrm{P}_M(\rho_{AB}\,\Vert\,\tau_A) \leq \mathrm{e}^{- \frac12 \left( \log M - I_\infty^{\eps'}(\rho_{AB}\,\Vert\, \tau_A) \right) } + \eps', \quad \forall\,\eps'>0,
\end{align}
where $I_\infty^{\eps'}$ denotes the smoothed max-relative entropy.
The additive error term $\eps'$ is effective in small deviation analysis (by choosing $\eps' = \eps + \sfrac{1}{\sqrt{n}}$ for an $\eps$-convex-splitting). However, we will loose the one-shot error-exponent expression in \eqref{eq:conclusion_smoothed} for large deviation analysis\footnote{
The philosophy behind the analysis is the following.
A polynomial \emph{multiplicative} penalty is fine for an error-exponent analysis, but an \emph{additive} penalty is hurting.
On the other hand, a \emph{multiplicative} penalty, say $2$, is damaging for small deviation analysis (because we want an $\eps$-capacity instead of a $2\eps$-capacity), but an \emph{additive} penalty of order $o(1)$ is fine for second-order asymptotics.
}.

In contrast, our one-shot strong converse bound holds for all rates outside the achievable rate region, which provides a complete one-shot analysis for all the rate regions.
{Hence, this is the first work being able to characterize the one-shot error performance at rates near the fundamental threshold $I(\rho_{AB}\,\Vert\,\tau_A)$.}
On one hand, it is tight enough to give a \emph{moderate deviation} result for the high-error regime \cite{CH17, CG22}; namely, while the rate approaches the fundamental limit from below, the trace distance converges to 1 asymptotically. 
On the other hand, it is unclear how to improve the established strong converse exponent even for the \emph{commuting case}.
To date, we do not find tighter strong converse bounds (up to a constant) even in the classical case. We leave it as an open question.

The proposed one-shot technique are used as a key building block in analyzing the \emph{covering-type} problems and establishing \text{new achievable one-shot error exponents} in several quantum information tasks mentioned above. 
We note that the applications are one-sided because convex splitting is only used in achievability analysis of quantum information theory, to our best knowledge.
This does not mean our converse analysis given in Sections~\ref{sec:converse_bound} and \ref{sec:strong_converse} are not important.
They illustrates that the derived achievability analysis for convex splitting under trace distance in this work is tight\footnote{That is to say, if a quantum achievability result is not tight, it is not the proposed convex splitting analysis to be blamed, but probably the way of applying it.}.
%Let us point out another issue. The one-shot converse analysis for the problems presented in Section~\ref{sec:app} is hard even in the classical case because there are more than one information quantities involved in the problem.
% Still, the achievable bounds for private communication and secrete key distillation are improved upon existing results.
Here, we pose an open question: How to apply convex splitting for converse analysis in quantum information theory?

Some of the existing works on convex splitting employ relative entropy or purified distance as error criteria \cite{ADJ17, KL21, LY21a, LY21b}.
In this paper, we adopt trace distance for analysis.
We argue that both criteria have pros and cons depending on the contexts.
If the quantity of interest is already fidelity or purified distance, then probably using purified distance for convex splitting would be more appropriate because of the use of Uhlmann's theorem.
However, if the security criterion requires to be \emph{composable} \cite{Ren05, TLG+12, PR14}), then using trace distance as a figure of merit for convex splitting would be direct.
Moreover, a technique---\emph{change of measure}---is often used in classical and quantum information theory.
When applying it, the trace distance will naturally comes up (see also Fact~\ref{fact:1norm}).
As for change of measure under purified distance, a similar result was shown in \cite[Lemma 1]{AJW19a}: for any test $0\leq T\leq \mathds{1}$,
\begin{align}
    \sqrt{\Tr[\rho T ]} \leq \sqrt{\Tr[\sigma T]} + \mathrm{P}(\rho,\sigma).
\end{align}
It is not as tight as the one based on trace distance.

Lastly, we remark that the error exponent and strong converse bound derived in Section~\ref{sec:exponet} and the one-shot converse in Section~\ref{sec:converse_bound} hold for infinite-dimensional Hilbert spaces as well. However, the achievability for sample complexity in Section~\ref{sec:direct_bound} relies on pinching's inequality, where we need the finite-dimension assumption.
We leave it as a future work for removing such a technical requirement.

\section*{Acknowledgement}
We sincerely thank anonymous reviewers for giving us helpful comments, suggestions, and pointing out typos in our first version.
HC is supported by the Young Scholar Fellowship (Einstein Program) of the Ministry of Science and Technology, Taiwan (R.O.C.) under Grants No.~MOST 109-2636-E-002-001, No.~MOST 110-2636-E-002-009, No.~MOST 111-2636-E-002-001, No.~MOST 111-2119-M-007-006, and No.~MOST 111-2119-M-001-004, by the Yushan Young Scholar Program of the Ministry of Education, Taiwan (R.O.C.) under Grants No.~NTU-109V0904, No.~NTU-110V0904, and No.~NTU-111V0904 and by the research project ``Pioneering Research in Forefront Quantum Computing, Learning and Engineering'' of National Taiwan University under Grant No. NTU-CC-111L894605.''
LG is partially supported by NSF grant DMS-2154903.

\appendix
\section{An Additivity Property} \label{sec:additivity}
\begin{prop_add} %\label{prop:properties}
	Let $\rho_{AB} \in \mathcal{S}(\mathsf{A}\otimes \mathsf{B})$ be a bipartite density operator.
	Then, the following hold.
	\begin{enumerate}
		\item\label{item:convexity} (Convexity) The map $(\tau_A, \sigma_B) \mapsto D_\alpha^*(\rho_{AB}\,\Vert\, \tau_A \otimes \sigma_B)$ is convex on $\mathcal{S}(\mathsf{A}) \times \mathcal{S}(\mathsf{B})$ for $\alpha\in (1,\infty)$.
		
		\item\label{item:fixed-point} (Fixed-point property) 
		Suppose the underlying Hilbert spaces are all finite-dimensional.
		For $\alpha \in (1,\infty)$, any pair  $(\tau_A^\star, \sigma_B^\star)$ attaining $D_\alpha^*(\rho_{AB}\,\Vert\, \tau_A^\star \otimes \sigma_B^\star) = I^{\downarrow \downarrow}(A:B)_\rho$
		satisfy
		\begin{align}
			\tau_A^\star &= \frac{ \Tr_B \left[ \left( (\tau_A^\star\otimes \sigma_B^\star)^{\frac{1-\alpha}{2\alpha}} \rho_{AB} (\tau_A^\star\otimes \sigma_B^\star)^{\frac{1-\alpha}{2\alpha}} \right)^\alpha \right] }{ \Tr \left[ \left( (\tau_A^\star\otimes \sigma_B^\star)^{\frac{1-\alpha}{2\alpha}} \rho_{AB} (\tau_A^\star\otimes \sigma_B^\star)^{\frac{1-\alpha}{2\alpha}} \right)^\alpha \right]  }, \label{eq:fixed_tau} \\
			\sigma_B^\star &= \frac{ \Tr_A \left[ \left( (\tau_A^\star\otimes \sigma_B^\star)^{\frac{1-\alpha}{2\alpha}} \rho_{AB} (\tau_A^\star\otimes \sigma_B^\star)^{\frac{1-\alpha}{2\alpha}} \right)^\alpha \right] }{ \Tr \left[ \left( (\tau_A^\star\otimes \sigma_B^\star)^{\frac{1-\alpha}{2\alpha}} \rho_{AB} (\tau_A^\star\otimes \sigma_B^\star)^{\frac{1-\alpha}{2\alpha}} \right)^\alpha \right]  }. \label{eq:fixed_sigma}
		\end{align}
	
		\item\label{item:additivity} (Additivity) 
		Suppose the underlying Hilbert spaces are all finite-dimensional.		
		For any density operators $\rho_{A_1 B_1} \in \mathcal{S}(\mathsf{A}_1 \otimes \mathsf{B}_1 )$ and $\omega_{A_2 B_2} \in \mathcal{S}(\mathsf{A}_2 \otimes \mathsf{B}_2 )$, then for any $\alpha\in (1,\infty)$, the information quantity defined in \eqref{eq:double} satisfies:
		\begin{align} \label{eq:additivity}
			I^{\downarrow \downarrow}(A_1 A_2:B_1 B_2)_{\rho\otimes \omega} 
			= I^{\downarrow \downarrow}(A_1:B_1)_{\rho} + I^{\downarrow \downarrow}(A_2:B_2)_{\omega} . 
		\end{align}
%		\begin{align}
%			\inf_{ \tau_{A_1 A_2} \in \mathcal{S}(\mathsf{A}_1 \otimes  \mathsf{A}_2 ) }	I_\alpha^*\left(\rho_{A_1 B_1} \otimes \rho_{A_2 B_2} \,\|\,\tau_{A_1 A_2}\right)
%			=  \inf_{ \tau_{A_1} \in \mathcal{S}(\mathsf{A}_1) }	I_\alpha^*\left(\rho_{A_1 B_1} \,\|\,\tau_{A_1}\right)
%			+ \inf_{ \tau_{A_2} \in \mathcal{S}(\mathsf{A}_2) }	I_\alpha^*\left(\rho_{A_2 B_2} \,\|\,\tau_{A_2}\right).
%		\end{align}
	\end{enumerate}
	
\end{prop_add}

\begin{proof}
	Item~(\ref{item:convexity}):
	By duality \cite[Lemma 12]{MDS+13} \cite{FL13},
	\begin{align}
		D_\alpha\left( \rho_{AB} \,\Vert\, \tau_{A} \otimes \sigma_{B} \right)
		= \frac{\alpha}{\alpha-1} \sup_{ \omega\geq 0,  \Tr[\omega]\leq 1 } \log \left\langle \left( \tau_A\otimes \sigma_B \right)^{\frac{1-\alpha}{\alpha}} , \rho_{AB}^{\frac12} \omega^{\frac{\alpha-1}{\alpha}} \rho_{AB}^{\frac12} \right\rangle.
	\end{align}
	%Note that $(\cdot)^{\frac{1-\alpha}{\alpha}}$ is monotone decreasing for $\alpha>1$.
	%If $\tau_A$ is fixed, then the above inner product is decreasing in $\sigma_B$. Then, by Proposition 1.1 of \url{https://arxiv.org/pdf/0911.5267.pdf}, 
	%\begin{align}
	%	\sigma_B \mapsto \log \left\langle \left( \tau_A\otimes \sigma_B \right)^{\frac{1-\alpha}{\alpha}} , \rho_{AB}^{\frac12} \omega^{\frac{\alpha-1}{\alpha}} \rho_{AB}^{\frac12} \right\rangle
	%\end{align}
	%is convex. Further, supremum of convex functions is convex showing the convexity in $\sigma_B$.
	%
	%Taking $\tau_A$ into account, the map
	%\begin{align}
	%	(\tau_A, \sigma_B) \mapsto \left( \tau_A\otimes \sigma_B \right)^{\frac{1-\alpha}{\alpha}}
	%\end{align}
	%is still monotone decreasing.
	%{\color{red} Can we argue that the functional
	%\begin{align}
	%	(\tau_A, \sigma_B) \mapsto \log \left\langle \left( \tau_A\otimes \sigma_B \right)^{\frac{1-\alpha}{\alpha}} , \rho_{AB}^{\frac12} \omega^{\frac{\alpha-1}{\alpha}} \rho_{AB}^{\frac12} \right\rangle
	%\end{align}
	%is still convex?
	%}
	Then, the log-convexity provided in Lemma~\ref{lemm:log-convexity} with the linear functional 
	$f(\cdot ) = \left\langle \,\cdot\,, \rho_{AB}^{\frac12} \omega^{\frac{\alpha-1}{\alpha}} \rho_{AB}^{\frac12} \right\rangle$ with $s = \frac{1-\alpha}{\alpha} \in (-1,0)$ and the fact that pointwise supremum of convex functions is convex 
	yield the desired convexity.

	Item~(\ref{item:fixed-point}):
	The proof of this property closely follows from the derivation given by Tomamichel and Hayashi \cite[Appendix C]{HT14}.
	With the finite-dimension assumption, the minimizers exits by the lower semi-continuity of the sandwiched R\'enyi divergence in its second argument (see e.g.~\cite[Corollary 3.27]{MO17}) and the extreme-value theorem.
	
	For $\alpha\in(1,\infty)$, we define
	\begin{align}
		\chi_\alpha(\tau_A,\sigma_B) &:= \Tr \left[ \left( (\tau_A\otimes \sigma_B)^{\frac{1-\alpha}{2\alpha}} \rho_{AB} (\tau_A\otimes \sigma_B)^{\frac{1-\alpha}{2\alpha}} \right)^\alpha \right], \\
		\mathcal{M}_\alpha &:= \argmin_{ (\tau_{A}, \sigma_{B})  \in \mathcal{S}(\mathsf{A}_1  ) \times \mathcal{S}(\mathsf{B} )} D_\alpha^*\left( \rho_{AB} \,\Vert\, \tau_{A} \otimes \sigma_{B} \right), \\
		\mathcal{F}_\alpha &:= \left\{ (\tau_{A}, \sigma_{B})  \in \mathcal{S}(\mathsf{A} ) \times \mathcal{S}(\mathsf{B} ): \eqref{eq:fixed_tau} \text{ and } \eqref{eq:fixed_sigma} \text{ hold}  \right\}.
	\end{align}
%	{\color{red} We may only need to show 
	We claim that:
		\begin{align}
			\mathcal{M}_\alpha =		\mathcal{F}_\alpha.
		\end{align}
%	}
	Further, it is sufficient to consider a minimizer $(\tau_A, \sigma_B)$ that has full support.
	Since the map $\chi_\alpha(\tau_A,\sigma_B)$ is convex in $(\tau_A, \sigma_B)$ by Item~\ref{item:convexity},
	we take Fr\'echet derivative with direction of traceless operators $(\delta_A, \delta_B)$, i.e.~$\Tr[\delta_A] = \Tr[\delta_B] = 0$:
	\begin{align}
		&\partial_{\delta_A, \delta_B} \chi_\alpha(\tau_A, \sigma_B) \notag \\
		&= \alpha \Tr\left[ \left( \rho^{\frac12} \tau_A^{\frac{1-\alpha}{\alpha}} \otimes \sigma_B^{\frac{1-\alpha}{\alpha}} \rho^{\frac12} \right)^{\alpha-1} \cdot 
		\partial_{\delta_A, \delta_B} \left( \rho^{\frac12} \tau_A^{\frac{1-\alpha}{\alpha}} \otimes \sigma_B^{\frac{1-\alpha}{\alpha}} \rho^{\frac12} \right)
		\right] \\
		&= \alpha \Tr\left[ \tau_A^{\frac{1-\alpha}{2\alpha}} \otimes \sigma_B^{\frac{1-\alpha}{2\alpha}} \rho^{\frac12} \left( \rho^{\frac12} \tau_A^{\frac{1-\alpha}{\alpha}} \otimes \sigma_B^{\frac{1-\alpha}{\alpha}} \rho^{\frac12} \right)^{\alpha-1} \rho^{\frac12} \tau_A^{\frac{1-\alpha}{2\alpha}} \otimes \sigma_B^{\frac{1-\alpha}{2\alpha}} \right. \notag \\
		&\qquad \qquad \left.\cdot \tau_A^{-\frac{1-\alpha}{2\alpha}} \otimes \sigma_B^{-\frac{1-\alpha}{2\alpha}}
		\partial_{\delta_A, \delta_B} \left( \tau_A^{\frac{1-\alpha}{\alpha}} \otimes \sigma_B^{\frac{1-\alpha}{\alpha}}  \right) \tau_A^{-\frac{1-\alpha}{2\alpha}} \otimes \sigma_B^{-\frac{1-\alpha}{2\alpha}}
		\right] \\	
		&= \alpha \Tr\left[ \left( \tau_A^{\frac{1-\alpha}{2\alpha}} \otimes \sigma_B^{\frac{1-\alpha}{2\alpha}} \rho \tau_A^{\frac{1-\alpha}{2\alpha}} \otimes \sigma_B^{\frac{1-\alpha}{2\alpha}} \right)^\alpha \right. \notag \\
		&\qquad \qquad \left.\cdot \tau_A^{-\frac{1-\alpha}{2\alpha}} \otimes \sigma_B^{-\frac{1-\alpha}{2\alpha}}
		\partial_{\delta_A, \delta_B} \left( \tau_A^{\frac{1-\alpha}{\alpha}} \otimes \sigma_B^{\frac{1-\alpha}{\alpha}}  \right) \tau_A^{-\frac{1-\alpha}{2\alpha}} \otimes \sigma_B^{-\frac{1-\alpha}{2\alpha}}
		\right]. \label{eq:temp}
	\end{align}
	Now, we take the partial derivative with respect to $\sigma_B$ (with $\tau_A$ fixed), then \eqref{eq:temp} is
	\begin{align}
		\partial_{ \delta_B} \chi_\alpha(\tau_A, \sigma_B)
		&= \alpha \Tr\left[  \sigma_B^{-\frac12} \Tr_A\left[\left( \left( \tau_A \otimes \sigma_B \right)^{\frac{1-\alpha}{2\alpha}} \rho \left( \tau_A \otimes \sigma_B \right)^{\frac{1-\alpha}{2\alpha}} \right)^\alpha \right]  \sigma_B^{-\frac12}  
	\cdot    \sigma_B^{1-\frac{1}{2\alpha}} 
			\partial_{\delta_B} \left(  \sigma_B^{\frac{1-\alpha}{\alpha}}  \right) \sigma_B^{1-\frac{1}{2\alpha}}  
		\right]. \label{eq:partial_delta_B}
	\end{align}

	Note that the set $ \left\{\sigma_B^{1-\frac{1}{2\alpha}} 
	\partial_{\delta_B} \left(  \sigma_B^{\frac{1-\alpha}{\alpha}}  \right) \sigma_B^{1-\frac{1}{2\alpha}}\right\}_{\delta_A \in \mathcal{B}(\mathsf{A})}$ span the space of Hermitian traceless operators as pointed out in \cite[Appendix C]{HT14}.
	To demand the partial derivative $\partial_{ \delta_B} \chi_\alpha(\tau_A, \sigma_B)$ to be $0$ for all traceless $\delta_B$, we require that the
	\begin{align}
		\sigma_B^{-\frac12} \Tr_A\left[\left( \left( \tau_A \otimes \sigma_B \right)^{\frac{1-\alpha}{2\alpha}} \rho \left( \tau_A \otimes \sigma_B \right)^{\frac{1-\alpha}{2\alpha}} \right)^\alpha \right]  \sigma_B^{-\frac12}   \propto \mathds{1}_B,
	\end{align}
	which translates to \eqref{eq:fixed_tau}.
%	{\color{red}red} part to be proportional to $\mathds{1}_B$.
%	Then, for the $B$ part:
%	\begin{align}
%		\sigma_B \propto    \Tr_A\left[\left( \left( \tau_A \otimes \sigma_B \right)^{\frac{1-\alpha}{2\alpha}} \rho \left( \tau_A \otimes \sigma_B \right)^{\frac{1-\alpha}{2\alpha}} \right)^\alpha \right]  .
%	\end{align}
	
	On the other hand, \eqref{eq:temp} with taking partial derivative with respect to $\tau_A$ (with $\sigma_B$ fixed) can be rewritten as:
	\begin{align}
		\partial_{\delta_A} \chi_\alpha(\tau_A, \sigma_B)
		&= \alpha \Tr\left[ \tau_A^{-\frac12} \Tr_B\left[\left( \left( \tau_A \otimes \sigma_B \right)^{\frac{1-\alpha}{2\alpha}} \rho \left( \tau_A \otimes \sigma_B \right)^{\frac{1-\alpha}{2\alpha}} \right)^\alpha \right]  \tau_A^{-\frac12} \cdot   \tau_A^{1-\frac{1}{2\alpha}} 
			\partial_{\delta_A} \left( \tau_A^{\frac{1-\alpha}{\alpha}} \right)  \tau_A^{1-\frac{1}{2\alpha}} 
		\right].
	\end{align}
	Similarly, for it to be $0$ for all traceless operator $\delta_A$, we must have
	\begin{align}
		\tau_A \propto    \Tr_B\left[\left( \left( \tau_A \otimes \sigma_B \right)^{\frac{1-\alpha}{2\alpha}} \rho \left( \tau_A \otimes \sigma_B \right)^{\frac{1-\alpha}{2\alpha}} \right)^\alpha \right].
	\end{align}
	This completes the claim for Item~\ref{item:fixed-point}.

	Item~(\ref{item:additivity}): 
	Let $(\tau_{A_1}^\star, \sigma_{B_1}^\star) \in \mathcal{S}(\mathsf{A}_1) \times \mathcal{S}(\mathsf{B}_1)$ and $(\tau_{A_2}^\star, \sigma_{B_2}^\star) \in \mathcal{S}(\mathsf{A}_2) \times \mathcal{S}(\mathsf{B}_2)$ satisfy
	\begin{align}
		D_\alpha^*(\rho_{A_1 B_1}\,\Vert\, \tau_{A_1}^\star \otimes \sigma_{B_1}^\star) &= I^{\downarrow \downarrow}(A_1:B_1)_{\rho}, \\
		D_\alpha^*(\omega_{A_2 B_2}\,\Vert\, \tau_{A_2}^\star \otimes \sigma_{B_2}^\star) &= I^{\downarrow \downarrow}(A_2:B_2)_{\omega}.
	\end{align}
	Note that the ``$\leq$'' direction in \eqref{eq:additivity} is easy because by the additivity of $D_\alpha^*$,
	\begin{align}
		I^{\downarrow \downarrow}(A_1 A_2:B_1 B_2)_{\rho\otimes \omega} 
		&\leq D_\alpha^*(\rho_{A_1 B_1} \otimes \omega_{A_2 B_2} \,\Vert\, \tau_{A_1}^\star \otimes \tau_{A_2}^\star \otimes  \sigma_{B_1}^\star \otimes \sigma_{B_2}^\star) \\
		&= D_\alpha^*(\rho_{A_1 B_1}\,\Vert\, \tau_{A_1}^\star \otimes \sigma_{B_1}^\star) + D_\alpha^*(\rho_{A_2 B_2}\,\Vert\, \tau_{A_2}^\star \otimes \sigma_{B_2}^\star) \\
		&= I^{\downarrow \downarrow}(A_1:B_1)_{\rho} + I^{\downarrow \downarrow}(A_2:B_2)_{\omega}.
	\end{align}

	Next, we will show that the product states $(\tau_{A_1}^\star \otimes \tau_{A_2}^\star, \sigma_{B_1}^\star \otimes \sigma_{B_2}^\star)$ actually attain the minimum of $I^{\downarrow \downarrow}(A_1 A_2:B_1 B_2)_{\rho\otimes \omega} $.
	Indeed, by the fixed-point property given in Item~\ref{item:fixed-point}, the following four equations are satisfied:
	\begin{align}
		\tau_{A_1}^\star &= \frac{ \Tr_{B_1} \left[ \left( (\tau_{A_1}^\star\otimes \sigma_{B_1}^\star)^{\frac{1-\alpha}{2\alpha}} \rho_{A_1 B_1} (\tau_{A_1}^\star\otimes \sigma_{B_1}^\star)^{\frac{1-\alpha}{2\alpha}} \right)^\alpha \right] }{ \Tr \left[ \left( (\tau_{A_1}^\star\otimes \sigma_{B_1}^\star)^{\frac{1-\alpha}{2\alpha}} \rho_{A_1 B_1} (\tau_{A_1}^\star\otimes \sigma_{B_1}^\star)^{\frac{1-\alpha}{2\alpha}} \right)^\alpha \right]  }, \label{eq:tau_A1} \\
		\sigma_{B_1}^\star &= \frac{ \Tr_{A_1} \left[ \left( (\tau_{A_1}^\star\otimes \sigma_{B_1}^\star)^{\frac{1-\alpha}{2\alpha}} \rho_{A_1 B_1} (\tau_{A_1}^\star\otimes \sigma_{B_1}^\star)^{\frac{1-\alpha}{2\alpha}} \right)^\alpha \right] }{ \Tr \left[ \left( (\tau_{A_1}^\star\otimes \sigma_{B_1}^\star)^{\frac{1-\alpha}{2\alpha}} \rho_{A_1 B_1} (\tau_{A_1}^\star\otimes \sigma_{B_1}^\star)^{\frac{1-\alpha}{2\alpha}} \right)^\alpha \right]  }, \label{eq:sigma_B1} \\
		\tau_{A_2}^\star &= \frac{ \Tr_{B_2} \left[ \left( (\tau_{A_2}^\star\otimes \sigma_{B_2}^\star)^{\frac{1-\alpha}{2\alpha}} \omega_{A_2 B_2} (\tau_{A_2}^\star\otimes \sigma_{B_2}^\star)^{\frac{1-\alpha}{2\alpha}} \right)^\alpha \right] }{ \Tr \left[ \left( (\tau_{A_2}^\star\otimes \sigma_{B_2}^\star)^{\frac{1-\alpha}{2\alpha}} \omega_{A_2 B_2} (\tau_{A_2}^\star\otimes \sigma_{B_1}^\star)^{\frac{1-\alpha}{2\alpha}} \right)^\alpha \right]  }, \label{eq:tau_A2} \\
		\sigma_{B_2}^\star &= \frac{ \Tr_{A_2} \left[ \left( (\tau_{A_2}^\star\otimes \sigma_{B_2}^\star)^{\frac{1-\alpha}{2\alpha}} \omega_{A_2 B_2} (\tau_{A_2}^\star\otimes \sigma_{B_2}^\star)^{\frac{1-\alpha}{2\alpha}} \right)^\alpha \right] }{ \Tr \left[ \left( (\tau_{A_2}^\star\otimes \sigma_{B_2}^\star)^{\frac{1-\alpha}{2\alpha}} \omega_{A_2 B_2} (\tau_{A_2}^\star\otimes \sigma_{B_2}^\star)^{\frac{1-\alpha}{2\alpha}} \right)^\alpha \right]  }. \label{eq:sigma_B2}
	\end{align}
	By tensor product the right-hand side of \eqref{eq:tau_A1} and \eqref{eq:tau_A2},  and 
	similarly tensor product the right-hand side of \eqref{eq:sigma_B1} and \eqref{eq:sigma_B2},
	one can immediately check that both the product states $\tau_{A_1 A_2}^\star := \tau_{A_1}^\star \otimes \tau_{A_2}^\star$ and 	$\sigma_{B_1 B_2}^\star := \sigma_{B_1}^\star \otimes \sigma_{B_2}^\star$
	satisfy the following fixed-point equations:
	\begin{align}
		\tau_{A_1 A_2}^\star &= \frac{ \Tr_{B_1 B_2} \left[ \left( (\tau_{A_1 A_2}^\star\otimes \sigma_{B_1 B_2}^\star)^{\frac{1-\alpha}{2\alpha}} \rho_{A_1 B_1} \otimes \omega_{A_2 B_2} (\tau_{A_1 A_2}^\star\otimes \sigma_{B_1 B_2}^\star)^{\frac{1-\alpha}{2\alpha}} \right)^\alpha \right] }{ \Tr \left[ \left( (\tau_{A_1 A_2}^\star\otimes \sigma_{B_1 B_2}^\star)^{\frac{1-\alpha}{2\alpha}} \rho_{A_1 B_1} \otimes \omega_{A_2 B_2} (\tau_{A_1 A_2}^\star\otimes \sigma_{B_1 B_2}^\star)^{\frac{1-\alpha}{2\alpha}} \right)^\alpha \right]  }, \label{eq:tau_12}\\
		\sigma_{B_1 B_2}^\star &= \frac{ \Tr_{A_1 A_2} \left[ \left( (\tau_{A_1 A_2}^\star\otimes \sigma_{B_1 B_2}^\star)^{\frac{1-\alpha}{2\alpha}} \rho_{A_1 B_1} \otimes \omega_{A_2 B_2} (\tau_{A_1 A_2}^\star\otimes \sigma_{B_1 B_2}^\star)^{\frac{1-\alpha}{2\alpha}} \right)^\alpha \right] }{ \Tr \left[ \left( (\tau_{A_1 A_2}^\star\otimes \sigma_{B_1 B_2}^\star)^{\frac{1-\alpha}{2\alpha}} \rho_{A_1 B_1} \otimes \omega_{A_2 B_2} (\tau_{A_1 A_2}^\star\otimes \sigma_{B_1 B_2}^\star)^{\frac{1-\alpha}{2\alpha}} \right)^\alpha \right]  }. \label{eq:sigma_12}
	\end{align}
	This again with the fixed-point property (Item~\ref{item:fixed-point}) ensure that the product states $(\tau_{A_1}^\star \otimes \tau_{A_2}^\star, \sigma_{B_1}^\star \otimes \sigma_{B_2}^\star)$ must be one of the minimizers of $I^{\downarrow \downarrow}(A_1 A_2:B_1 B_2)_{\rho\otimes \omega} $. This then completes the proof.
\end{proof}

\begin{lemm}[A log-convexity property] \label{lemm:log-convexity}
%	Let $h$ be a non-negative operator monotone decreasing function on $(0,\infty)$, and 
	Let $f$ be a positive linear functional on $\mathcal{B}(\mathsf{A}\otimes \mathsf{B})$.
	Then, for $s\in[-1,0)$, the functional 
	\begin{align}
		(\tau_A, \sigma_B) \mapsto \log f\left( (\tau_A\otimes \sigma_B)^s \right)
	\end{align}
	is jointly convex on the interior of $\mathcal{B}(\mathsf{A})\times \mathcal{B}(\mathsf{B})$.
\end{lemm}

\begin{remark}
	Lemma~\ref{lemm:log-convexity} for a special case that one of the Hilbert space is void was proved by Ando and Hiai \cite[Proposition 1.1]{AH10}, which can be used to prove the convexity of the quantum sandwiched R'enyi divergence in its second argument; see e.g.~Ref.~\cite[Proposition 3.18]{MO17} by Mosonyi and Ogawa.
\end{remark}

\begin{proof}[Proof of Lemma~\ref{lemm:log-convexity}]
	Denote by 
	\begin{align}
		X \# Y := X^\frac12 \left(X^{-\frac12} Y X^{-\frac12} \right)^{\frac12} X^{\frac12}
	\end{align}
	be the symmetric operator geometric mean.
	We first prove an operator inequality that will be used in the proof.
	For any positive definite operators $X$ and $Y$, we have 
	\begin{align}
		\left( \frac{X+Y}{2} \right)^s &= \left(\left( \frac{X+Y}{2} \right)^{-s}\right)^{-1} \\
		&\overset{\textnormal{(a)}}{\leq} \left( \frac{X^{-s} + Y^{-s}}{2} \right)^{-1} \\
		&\overset{\textnormal{(b)}}{\leq} \left( X^{-s} \# Y^{-s} \right)^{-1} \\
		&\overset{\textnormal{(c)}}{=} X^s \# Y^s, \label{eq:AG}
	\end{align}
	where (a) follows from the operator concavity of the power function $(\cdot)^{-s}$ for $-s\in(0,1]$, i.e.~$(\frac{X+Y}{2})^{-s} \geq \frac{X^{-s}+Y^{-s}}{2}$ and that inverse is operator monotone decreasing;
	(b) follows from the operator arithmetic-geometric mean inequality, i.e.~$\frac{X^{-s}+Y^{-s}}{2} \geq X^{-s} \# Y^{-s}$ and again that inverse is operator monotone decreasing;
	and (c) is by the definition of the operator geometric mean.

	Next, let $\tau_A, \tau_A' \in \mathcal{B}(\mathsf{A})$ and $\sigma_B , \sigma_B' \in \mathcal{B}(\mathsf{B})$ be positive semi-definite.
	Then, by applying the operator inequality in \eqref{eq:AG} twice and note that tensor product $(X,Y)\mapsto X\otimes Y$ is operator monotone increasing, we get
	\begin{align}
		\left( \left(\frac{\tau_A+\tau_A'}{2}\right) \otimes \left(\frac{\sigma_B+\sigma_B'}{2}\right) \right)^s 
		&= \left(\frac{\tau_A+\tau_A'}{2}\right)^s \otimes \left(\frac{\sigma_B+\sigma_B'}{2}\right)^s \label{eq:original}\\
%		&\leq \left(\tau_A\#\tau_A'\right)^s \otimes \left(\sigma_B\#\sigma_B'\right)^s  \\
%		&\overset{\textnormal{(a)}}{=} 	
	&\leq (\tau_A)^s\#(\tau_A')^s \otimes (\sigma_B)^s\#(\sigma_B')^s \\
		&= \left( (\tau_A)^s \otimes  (\sigma_B)^s \right) \# \left( (\tau_A')^s \otimes  (\sigma_B')^s \right) \\
		&= (\tau_A\otimes \sigma_B)^s \# (\tau_A ' \otimes \sigma_B')^s. \label{eq:goal}
	\end{align}
%	where (a) only holds for $s=-1$.
	
	Then, by a standard argument, for any $t>0$,
	\begin{align}
		\left( \left(\frac{\tau_A+\tau_A'}{2}\right) \otimes \left(\frac{\sigma_B+\sigma_B'}{2}\right) \right)^s 
		&\leq (\tau_A\otimes \sigma_B)^s \# (\tau_A ' \otimes \sigma_B')^s \\
		&= t(\tau_A\otimes \sigma_B)^s \# t^{-1}(\tau_A ' \otimes \sigma_B')^s \\
		&\leq \frac{ t(\tau_A\otimes \sigma_B)^s + t^{-1}(\tau_A ' \otimes \sigma_B')^s  }{2}.
	\end{align}
	This implies that, for any $t>0$,
	\begin{align}
	f\left(\left( \left(\frac{\tau_A+\tau_A'}{2}\right) \otimes \left(\frac{\sigma_B+\sigma_B'}{2}\right) \right)^s\right) &\leq \frac{ tf\left((\tau_A\otimes \sigma_B)^s\right) + t^{-1}f\left((\tau_A ' \otimes \sigma_B')^s\right)  }{2}.%, \quad \forall\, t>0.
	\end{align}	
	Optimizing over $t>0$, gives
	\begin{align}
		f\left(\left( \left(\frac{\tau_A+\tau_A'}{2}\right) \otimes \left(\frac{\sigma_B+\sigma_B'}{2}\right) \right)^s\right) \leq \sqrt{ f\left( (\tau_A\otimes \sigma_B)^s\right) f\left( (\tau_A'\otimes \sigma_B')^s\right) }.
	\end{align}
	After taking logarithm, we complete the proof via the mid-point convexity.
\end{proof}

%%%%%%%%%%%%%%%%%%%%%%%%%%%%%%%%%%%%%%%%%%%%%%%%%%%%%%%%%

{\larger
\bibliographystyle{myIEEEtran}
\bibliography{reference.bib}
}

\end{document}